\documentclass[11pt,journal,draft,letterpaper,onecolumn]{IEEEtran}
\newtheorem{definition}{Definition}[section]
\newtheorem{theorem}{Theorem}[section]
\newtheorem{lemma}{Lemma}[section]

\newtheorem{remark}{Remark}[section]

\newtheorem{corollary}{Corollary}[section]
\newcommand{\argmax}{\operatornamewithlimits{argmax}}

\renewcommand{\thefootnote}{}
\newcommand{\bydef} {{\buildrel{\triangle}\over =}}

\usepackage[dvips]{graphicx}
\usepackage{latexsym}
\usepackage{rotate}
\usepackage{amsmath,amssymb}
\usepackage{epsf,epsfig}
\DeclareGraphicsExtensions{.ps}
\usepackage{amsfonts}
\usepackage{graphicx}
\usepackage{amsmath}
\usepackage{pstricks}

\usepackage{dsfont}
\begin{document}
\title{\huge On the Multiple Access Channel with Asymmetric Noisy State Information at the Encoders}
\author{Nevroz~\c{S}en,~ Fady~Alajaji,~Serdar~Y\"{u}ksel and Giacomo~Como }
\maketitle
\renewcommand{\thefootnote}{}
\footnotetext{This work was supported in part by the Natural Sciences and Engineering Research Council of Canada (NSERC).}
\footnotetext{The material in this paper was presented in part at Forty-Ninth Annual Allerton Conference on Communication, Control, and Computing, Monticello, IL, September 2011.}
\footnote{N. \c{S}en, F. Alajaji and S. Y\"{u}ksel are with the Department of Mathematics and Statistics, Queen's University, Kingston, ON K7L 3N6, Canada (email:\{nsen,fady,yuksel\}@mast.queensu.ca).}
\footnote{G. Como is with the Department of Automatic Control, Lund University, SE-221 00 Lund, Sweden (email:giacomo@control.lth.se).}
\vspace{-0.5in}
\begin{abstract}
We consider the problem of reliable communication over multiple-access channels (MAC) where the channel is driven by an independent and identically distributed state process and the encoders and the decoder are provided with various degrees of asymmetric noisy channel state information (CSI). For the case where the encoders observe causal, asymmetric noisy CSI and the decoder observes complete CSI, we provide inner and outer bounds to the capacity region, which are tight for the sum-rate capacity. We then observe that, under a Markov assumption, similar capacity results also hold in the case where the receiver observes noisy CSI. Furthermore, we provide a single letter characterization for the capacity region when the CSI at the encoders are asymmetric deterministic functions of the CSI at the decoder and the encoders have non-causal noisy CSI (its causal version is recently solved in \cite{como-yuksel}). When the encoders observe asymmetric noisy CSI with asymmetric delays and the decoder observes complete CSI, we provide a single letter characterization for the capacity region. Finally, we consider a cooperative scenario with common and private messages, with asymmetric noisy CSI at the encoders and complete CSI at the decoder. We provide a single letter expression for the capacity region for such channels. For the cooperative scenario, we also note that as soon as the common message encoder does not have access to CSI, then in any noisy setup, covering the cases where no CSI or noisy CSI at the decoder, it is possible to obtain a single letter characterization for the capacity region. The main component in these results is a generalization of a converse coding approach, recently introduced in \cite{como-yuksel} for the MAC with asymmetric quantized CSI at the encoders and herein considerably extended and adapted for the noisy CSI setup.
\end{abstract}
\renewcommand{\thefootnote}{\arabic{footnote}}
\setcounter{footnote}{0}
\section{Introduction and Literature Review}
Modeling communication channels with a state process, which governs the channel behavior, fits well for many physical scenarios. For single-user channels, the characterization of the capacity with various degrees of channel state information at the transmitter (CSIT) and at the receiver (CSIR) is well understood. Among them, Shannon \cite{Sha_Side} provides the capacity formula for a discrete memoryless channel with causal noiseless CSIT, where the state process is independent and identically distributed (i.i.d.), in terms of Shannon strategies (random functions from the state space to the channel input space). In \cite{gelfand-pinsker} Gel'fand and Pinsker consider the same problem with non-causal side information and establish a single-letter capacity formula. In \cite{Salehi}, noisy state observation available at both the transmitter and the receiver is considered and the capacity under such a setting is derived. Later, in \cite{caire-shamai} this result is shown to be a special case of Shannon's model and the authors also prove that when CSIT is a deterministic function of CSIR optimal codes can be constructed directly on the input alphabet. In \cite{erez-zamir}, the authors examine the discrete modulo-additive noise channel with casual CSIT which governs the noise distribution, and they determine the optimal strategies that achieve channel capacity. In \cite{gold-var} fading channels with perfect channel state information at the transmitter is considered and it is shown that with instantaneous and perfect CSI, the transmitter can adjust the data rates for each channel state to maximize the average transmission rate. In \cite{yuk-tat}, a single letter characterization of the capacity region for single-user finite-state Markovian channels with quantized state information available at the transmitter and full state information at the decoder is provided. In a closely related direction, finite-state channels (with memory) with output feedback is investigated in \cite{tat-mit}. In particular, \cite{tat-mit} shows that it is possible to formulate the computation of feedback capacity as a stochastic control problem. In \cite{per-weiss-gold}, finite-state channels with feedback, where feedback is a time-invariant deterministic function of the output samples, is considered.

The literature on finite state multiple access channels (FS-MAC) with different assumptions of CSIR and CSIT (such as causal vs non-causal, perfect vs imperfect) is extensive and the main contributions of the current paper have several interactions with the available results in the literature, which we present in Subsection \ref{subsec:main-cont}. Hence, we believe that in order to suitably highlight the contributions of this paper, it is worth to discuss the relevant literature for the multi-user setting in more detail. To start, \cite{das-narayan} provides a multi-letter characterization of the capacity region of time-varying MACs with general channel statistics (with/without memory) under a general state process (not necessarily stationary or ergodic) and with various degrees of CSIT and CSIR. In \cite{das-narayan}, it is also shown that when the channel is memoryless, if the encoders use only the past $k$ asymmetric partial (but not noisy) CSI and the decoder has complete CSI, then it is possible to simplify the multi-letter characterization to a single letter one \cite[Theorem 4]{das-narayan}. In \cite{Jafar}, a general framework for the capacity region of MACs with causal and non-causal CSI is presented. In particular, an achievable rate region is presented for the memoryless FS-MAC with correlated CSI and the sum-rate capacity is established under the condition that the state information available to each encoder are independent. In \cite{Cemal}, MACs with complete CSIR and noncausal, partial, rate limited CSITs are considered. In particular, for the degraded case, i.e., the case where the CSI available at one of the encoders is a subset of the CSI available at the other encoder, a single letter formula for the capacity region is provided and when the CSITs are not degraded, inner and outer bounds are derived, see \cite[Theorems 1, 2]{Cemal}. In \cite{lap-ste-1}, memoryless FS-MACs with two independent states, each known causally and strictly causally to one encoder, is considered and an achievable rate region, which is shown to contain an achievable region where each user applies Shannon strategies, is proposed. In \cite{LiYen}, another achievable rate region for the same problem is proposed and in \cite{lap-ste-3} it is shown that this region can be strictly larger than the one proposed in \cite{lap-ste-1}. In \cite{lap-ste-1} it is also shown that strictly casual CSI does not increase the sum-rate capacity. In \cite{basher} the finite-state Markovian MAC with asymmetric delayed CSITs is studied and its capacity region is determined. Another active research direction on the FS-MAC regards the so-called cooperative FS-MAC where there exists a degraded condition on the message sets. In particular, \cite{som-sha-verdu} and \cite{kotagari} characterize the capacity region of the cooperative FS-MAC with states non-causally and causally available at the transmitters. For more recent results on the cooperative FS-MAC problem see references \cite{Zaidi} and \cite{perm-sha-som}. Finally, for a comprehensive survey on channel coding with side information see \cite{keshet}.

The most relevant work to this paper is \cite{como-yuksel}, which presents a single letter characterization of the capacity region for memoryless FS-MAC in which transmitters observe asymmetric partial quantized CSI causally, and the receiver has full CSI. In the converse part of this work, which we discuss in more detail below, the authors use team decision theoretic methods \cite{Wits} (see also \cite{Yuksel}, \cite{maha-tene} and \cite{Nayyar} for recent team decision and control theoretic approaches). When a comparison of this result with the previously mentioned results is made, we observe the following: $i)$ it shows that when the state process is i.i.d. there is no loss of optimality if the encoders use a window size of $k=1$ in \cite[Theorem 3]{das-narayan}, $ii)$ it extends the causal part of result \cite[Theorem 5]{Jafar} to the case where CSITs are not independent, and finally, $iii)$ it partially answers the setup in \cite[Theorem 2]{Cemal} with the assumption that CSITs are causal.
\subsection{Main Contributions and Connections with the Literature}\label{subsec:main-cont}
We consider several scenarios where the encoders and the decoder observe various degrees of noisy CSI. The essential requirement we impose is that the noisy CSI available to the decision makers is realized via the corruption of CSI by different noise processes, which give a realistic physical structure of the communication setup. We herein note that the asymmetric noisy CSI assumption is acceptable as typically the feedback links are imperfect and sufficiently far from each other so that the information carried through them is corrupted by different (independent) noise processes. Finally, what makes (asymmetric) noisy setups particularly interesting are the facts that
\begin{itemize}
\item [(a)] No transmitter CSI contains the CSI available to the other one;
\item [(b)] CSI available to the decoder does not contain any of the CSI available to the two encoders.
\end{itemize}
When existing results, which provide a single letter capacity formulation, are examined, it can be observed that most of them do not satisfy $(a)$ or $(b)$ or both (e.g., \cite{como-yuksel}, \cite{das-narayan}, \cite{Jafar}, \cite{Cemal}, \cite{basher}). Nonetheless, among these, \cite{das-narayan} discusses the situation with noisy CSI and the authors make the observation that the situation where the CSI at the encoders and decoder are noisy versions of $S_t$ can be accommodated by their models. However, they also note that if the noises corrupting transmitters and receiver CSI are different, then the encoder CSI will, in general, not be contained in the decoder CSI. Hence, motivated by similar observations in the literature (e.g., \cite{Jafar}), we partially treat the scenarios below and provide inner and outer bounds, which are tight for the sum-rate capacity, for the scenarios $(1)$ and $(1a)$ and provide a single letter characterization for the capacity region of the latter scenarios:
\begin{itemize}
\item [(1)] The memoryless FS-MAC in which each of the transmitters has an asymmetric causal noisy CSI and the receiver has complete CSI (Theorems \ref{theo:inbnd-cau-fl}, \ref{theo:outbnd-cau-fl} and Corollary \ref{cor:sumra-cau-fl}).
    \begin{itemize}
    \item [(1a)] The memoryless FS-MAC in which each of the transmitters has an asymmetric causal noisy CSI and the receiver has also noisy CSI (Corollaries \ref{cor:inbnd-cau-rn}, \ref{cor:outbnd-cau-rn} and \ref{cor:sumra-cau-rn}).
    \item [(1b)] The memoryless FS-MAC in which each of the transmitters has an asymmetric causal and non-causal noisy CSIT which is a deterministic function of the noisy CSIR at the receiver (Theorem \ref{theo:main-inp-ncas}).
    \end{itemize}
\item [(2)] The memoryless FS-MAC in which each of the transmitters has an asymmetrically delayed and asymmetric noisy CSI and the receiver has complete CSI (Theorem \ref{theo:main-inp}).
\item [(3)] The cooperative memoryless FS-MAC in which both transmitters transmit a common message and one transmitter (informed transmitter) transmits a private message. The informed transmitter has causal noisy CSI, the other encoder has a delayed noisy CSI and the receiver has various degrees of CSI (Theorems \ref{theo:main-corr} and \ref{theo:main-corr-gen}).
\end{itemize}

Let us now briefly position these contributions with respect to the available results in the literature. The sum-rate capacity determined in $(1)$ and $(1a)$ can be thought as an extension of \cite[Theorem 4]{Jafar} to the case where the encoders have correlated CSI. The causal setup of $(1b)$, with the observation of the existence of an equivalent channel, is solved in \cite{como-yuksel}. The solution that we provide to the non-causal case partially solves \cite{Cemal} and extends \cite[Theorem 5]{Jafar} to the case where the encoders have correlated CSI. Furthermore, since the causal and non-causal capacities are identical for scenario $(1b)$, the causal solution can be considered as an extension of \cite[Proposition 1]{caire-shamai} to a noisy multi-user case. Finally, $(3)$ is an extension of \cite[Theorem 4]{som-sha-verdu} to a noisy setup.

\subsection{The Converse Coding Approach}
In this work, we adopt and expand on the converse technique presented in \cite{como-yuksel} and use it in a noisy setup. The converse coding approach of \cite{como-yuksel} is based on using \textit{memoryless stationary team policies} which play a key role in showing that the past information is irrelevant. This is obtained by showing that under any policy that one can achieve using an arbitrary decentralized coding policy, the same performance can be achieved by using memoryless stationary team policies. More specifically, this is accomplished in two steps. In the first step, it is shown that any achievable rate pair can be approximated with the convex combinations of conditional mutual information terms which are indexed by the past CSIR. In the second step, the conditional probability distribution, for which these conditional mutual information terms are a function of, is examined. With the observation that the past CSIR only affects the ``controls,'' i.e., memoryless stationary team policies, taking the convex hull associated to all
possible such controls completes the converse part. However, as the authors mention in \cite[Remark 2]{como-yuksel}, for the validity of the above arguments, it would suffice that the state information available at the
decoder contains the one available at the two transmitters. In this way, the decoder does not need to estimate the coding policies used in decentralized time-sharing.

For the noisy setup, we need to modify this approach to account for the fact that the decoder does not have access to the state information at the encoders, and that the past state information does not lead to a tractable recursion. This difficulty is overcome by showing that a product form on the team policies exists in the noisy setup as well.

The rest of the paper is organized as follows. In Sections \ref{sec:causal}, \ref{sec:delay} and \ref{sec:corr}, we formally state scenarios (1)-(1b), (2) and (3), respectively, and present the main results and several observations. In Section \ref{sec:example}, we provide two examples in one of which (the modulo-additive FS-MAC) we apply the result of \cite{erez-zamir} and get the full capacity region by only considering the tightness of the sum-rate capacity. Finally, in Section \ref{sec:conc}, we present concluding remarks.

Throughout the paper, we will use the following notations. A random variable will be denoted by an upper case letter $X$ and its particular realization by a lower case letter $x$. For a vector $v$, and a positive integer $i$, $v_{i}$ will denote the $i$-th entry of $v$, while $v_{[i]} = (v_1, \cdots, v_i)$ will denote the vector of the first $i$ entries and $v_{[i,j]} = (v_{i}, \cdots, v_j),\enspace i\leq j$ will denote the vector of entries between $i,j$ of $v$. For a finite set $\mathcal{A}$, $\mathcal{P}(\mathcal{A})$ will denote the simplex of probability distributions over $\mathcal{A}$. Probability distributions are denoted by $P(\cdot)$ and subscripted by the name of the random variables and conditioning, e.g., $P_{U,T|V,S}(u,t|v,s)$ is the conditional probability of $(U=u,T=t)$ given $(V=v,S=s)$. Finally, for a positive integer $n$, we shall denote by $\mathcal{A}^{(n)} := \bigcup_{0<s<n} \mathcal{A}^{s}$ the set of $\mathcal{A}$-strings of length smaller than $n$. We denote the indicator function of an event $E$ by $1_{\{E\}}$. All sets considered hereafter are finite.

\section{Asymmetric Causal Noisy CSIT and Complete CSIR}\label{sec:causal}
Consider a two-user memoryless FS-MAC, with two encoders, $a, b$, and two independent message sources $W_a$ and $W_b$ which are uniformly distributed in the finite sets $\mathcal W_a$ and $\mathcal W_b$, respectively. The channel inputs from the encoders are $X^{a} \in {\cal X}_a$ and $X^{b} \in {\cal X}_b$, respectively, and the channel output is $Y \in {\cal Y}$. The channel state process is modeled as a sequence $\{S_t\}_{t=1}^{\infty}$ of random variables in some finite space $\cal S$. The two encoders have access to a causal noisy version of the state information $S_t$ at each time $t\geq1$, modeled by $S_{t}^{a} \in {\cal S}_{a}$, $S_{t}^{b} \in {\cal S}_{b}$, respectively, where the joint distribution of $(S_t,S_t^a,S_t^b)$ factorizes as
\begin{eqnarray}
P_{S_{t}^{a}, S_{t}^{b},S_t}(s_{t}^{a}, s_{t}^{b},s_t)=P_{S_{t}^{a}|S_t}(s_{t}^{a}|s_t)P_{S_{t}^{b}|S_t}(s_{t}^{b}|s_t)P_{S_t}(s_t)\label{eq:noi-recfl-mrk}.
\end{eqnarray}
We also assume that $S_t$ is fully available at the receiver (see Fig. \ref{fig:mac-noi-cau}) and that $\{(S_t,S_t^a,S_t^b)\}_{t=1}^{\infty}$ is  a sequence of independent and identically distributed triples, independent from $(W_a,W_b)$. Therefore, we have that for any $n\geq1$,
\begin{eqnarray}
P_{S_{[n]},S_{[n]}^a,S_{[n]}^b,W_a,W_b}(s_{[n]},s_{[n]}^a,s_{[n]}^b,w_a,w_b)=\prod_{t=1}^n \frac{1}{|{\cal W}_a|}\frac{1}{|{\cal W}_b|}P_{S_{t}^{a}|S_t}(s_{t}^{a}|s_t)P_{S_{t}^{b}|S_t}(s_{t}^{b}|s_t)P_{S_t}(s_t)\label{eq:sta-iid}.
\end{eqnarray}
The channel inputs at time $t$, i.e., $X_t^{a}$ and $X_t^{b}$, are functions of the locally available information $(W_a,S_{[t]}^{a})$ and $(W_b,S_{[t]}^{b})$, respectively. Let $\mathbf{W}:=(W_a,W_b)$ and $\mathbf{X_t}:=(X_t^a,X_t^b)$, respectively. Then, the laws governing $n$-sequences of state, input and output letters are given by
\begin{eqnarray}
P_{Y_{[n]}|\mathbf{W},\mathbf{X}_{[n]},S_{[n]},S_{[n]}^a,S_{[n]}^b}(y_{[n]}|\mathbf{w},\mathbf{x}_{[n]},s_{[n]},s_{[n]}^a,s_{[n]}^b)=\prod _{t=1}^nP_{Y_t|X_{t}^{a}, X_{t}^{b}, S_t}(y_t|x_{t}^{a}, x_{t}^{b}, s_t), \label{eq:ch-recfl}
\end{eqnarray}
where the channel's transition probability distribution, $P_{Y_t|X_{t}^{a}, X_{t}^{b}, S_t}(y_t|x_{t}^{a}, x_{t}^{b}, s_t)$, is given a priori.
\begin{definition}\label{def:maccode-causal}
An $(n, 2^{nR_a}, 2^{nR_b})$ code with block length $n$ and rate pair $(R_a, R_b)$ for an FS-MAC with causal noisy state feedback consists of
\begin{itemize}
\item [(1)] A sequence of mappings for each encoder
\begin{center}
$\phi_{t}^{(a)}: \mathcal{S}_{a}^{t} \times \mathcal{W}_a  \rightarrow {\cal X}_a, \enspace t=1,2,...n$;\\
\end{center}
\begin{center}
$\phi_{t}^{(b)}: \mathcal{S}_{b}^{t} \times \mathcal{W}_b \rightarrow {\cal X}_b, \enspace t=1,2,...n$.
\end{center}
\item [2)] An associated decoding function
\begin{center}
$\psi: \mathcal{S}^n\times{\cal Y}^n\rightarrow \mathcal{W}_a \times \mathcal{W}_b$.\\
\end{center}
\end{itemize}
\end{definition}
\begin{figure}
\setlength{\unitlength}{3947sp}%
\begingroup\makeatletter\ifx\SetFigFont\undefined%
\gdef\SetFigFont#1#2#3#4#5{%
  \reset@font\fontsize{#1}{#2pt}%
  \fontfamily{#3}\fontseries{#4}\fontshape{#5}%
  \selectfont}%
\fi\endgroup%
\begin{picture}(5945,2853)(0,-3377)
\thinlines
{\color[rgb]{0,0,0}\put(1083,-2549){\vector( 1, 0){569}}
}%
{\color[rgb]{0,0,0}\put(1658,-2909){\framebox(1200,712){}}
}%
{\color[rgb]{0,0,0}\put(1658,-1725){\framebox(1200,712){}}
}%
{\color[rgb]{0,0,0}\put(1083,-1327){\vector( 1, 0){569}}
}%
{\color[rgb]{0,0,0}\put(3283,-2296){\framebox(1521,705){}}
}%
{\color[rgb]{0,0,0}\put(5281,-2311){\framebox(1149,712){}}
}%
{\color[rgb]{0,0,0}\put(6430,-1753){\vector( 1, 0){536}}
}%
{\color[rgb]{0,0,0}\put(6430,-2201){\vector( 1, 0){536}}
}%
{\color[rgb]{0,0,0}\put(2860,-1317){\line( 1, 0){233}}
\put(3092,-1317){\line( 0,-1){391}}
\put(3092,-1708){\vector( 1, 0){190}}
}%
{\color[rgb]{0,0,0}\put(2860,-2538){\line( 1, 0){233}}
\put(3092,-2538){\line( 0, 1){371}}
\put(3092,-2160){\vector( 1, 0){190}}
}%
{\color[rgb]{0,0,0}\put(2207,-536){\vector( 0,-1){468}}
\put(2207,-536){\line( 1, 0){1735}}
\put(3940,-536){\line( 0,-1){1053}}
}%
{\color[rgb]{0,0,0}\put(2208,-3365){\vector( 0, 1){461}}
\put(2208,-3365){\line( 1, 0){1746}}
\put(3951,-3365){\line( 0, 1){1064}}
}%
{\color[rgb]{0,0,0}\put(4811,-1762){\vector( 1, 0){470}}
}%
{\color[rgb]{0,0,0}\put(4811,-2159){\vector( 1, 0){470}}
}%
\put(1356,-1230){\makebox(0,0)[b]{\smash{{\SetFigFont{11}{13.2}{\rmdefault}{\mddefault}{\updefault}{\color[rgb]{0,0,0}$W_a$}%
}}}}
\put(2240,-1281){\makebox(0,0)[b]{\smash{{\SetFigFont{11}{13.2}{\rmdefault}{\mddefault}{\updefault}{\color[rgb]{0,0,0}Encoder}%
}}}}
\put(2250,-1520){\makebox(0,0)[b]{\smash{{\SetFigFont{11}{13.2}{\rmdefault}{\mddefault}{\updefault}{\color[rgb]{0,0,0}$\phi_{t}^{(a)}(W_a,S_{[t]}^a)$}%
}}}}
\put(2240,-2691){\makebox(0,0)[b]{\smash{{\SetFigFont{11}{13.2}{\rmdefault}{\mddefault}{\updefault}{\color[rgb]{0,0,0}$\phi_{t}^{(b)}(W_b,S_{[t]}^b)$}%
}}}}
\put(2240,-2484){\makebox(0,0)[b]{\smash{{\SetFigFont{11}{13.2}{\rmdefault}{\mddefault}{\updefault}{\color[rgb]{0,0,0}Encoder}%
}}}}
\put(3995,-1802){\makebox(0,0)[b]{\smash{{\SetFigFont{11}{13.2}{\rmdefault}{\mddefault}{\updefault}{\color[rgb]{0,0,0}Channel}%
}}}}
\put(4091,-2110){\makebox(0,0)[b]{\smash{{\SetFigFont{11}{13.2}{\rmdefault}{\mddefault}{\updefault}{\color[rgb]{0,0,0}$P(Y_t|X_t^a, X_t^b,S_t)$}%
}}}}
\put(5890,-2132){\makebox(0,0)[b]{\smash{{\SetFigFont{11}{13.2}{\rmdefault}{\mddefault}{\updefault}{\color[rgb]{0,0,0}$\psi(Y_{[n]},S_{[n]})$}%
}}}}
\put(5834,-1885){\makebox(0,0)[b]{\smash{{\SetFigFont{11}{13.2}{\rmdefault}{\mddefault}{\updefault}{\color[rgb]{0,0,0}Decoder}%
}}}}
\put(6703,-1710){\makebox(0,0)[b]{\smash{{\SetFigFont{11}{13.2}{\rmdefault}{\mddefault}{\updefault}{\color[rgb]{0,0,0}$\hat{W}_a$}%
}}}}
\put(6708,-2149){\makebox(0,0)[b]{\smash{{\SetFigFont{11}{13.2}{\rmdefault}{\mddefault}{\updefault}{\color[rgb]{0,0,0}$\hat{W}_b$}%
}}}}
\put(1363,-2460){\makebox(0,0)[b]{\smash{{\SetFigFont{11}{13.2}{\rmdefault}{\mddefault}{\updefault}{\color[rgb]{0,0,0}$W_b$}%
}}}}
\put(3099,-1268){\makebox(0,0)[b]{\smash{{\SetFigFont{11}{13.2}{\rmdefault}{\mddefault}{\updefault}{\color[rgb]{0,0,0}$X_t^a$}%
}}}}
\put(3063,-2763){\makebox(0,0)[b]{\smash{{\SetFigFont{11}{13.2}{\rmdefault}{\mddefault}{\updefault}{\color[rgb]{0,0,0}$X_t^b$}%
}}}}
\put(2438,-827){\makebox(0,0)[b]{\smash{{\SetFigFont{11}{13.2}{\rmdefault}{\mddefault}{\updefault}{\color[rgb]{0,0,0}$S_{t}^a$}%
}}}}
\put(2444,-3159){\makebox(0,0)[b]{\smash{{\SetFigFont{11}{13.2}{\rmdefault}{\mddefault}{\updefault}{\color[rgb]{0,0,0}$S_{t}^b$}%
}}}}
\put(5009,-1700){\makebox(0,0)[b]{\smash{{\SetFigFont{11}{13.2}{\rmdefault}{\mddefault}{\updefault}{\color[rgb]{0,0,0}$Y_t$}%
}}}}
\put(5042,-2050){\makebox(0,0)[b]{\smash{{\SetFigFont{11}{13.2}{\rmdefault}{\mddefault}{\updefault}{\color[rgb]{0,0,0}$S_{t}$}%
}}}}
\end{picture}
\caption{The multiple-access channel with asymmetric causal noisy state feedback.}
\label{fig:mac-noi-cau}
\vspace{-0.2in}
\end{figure}
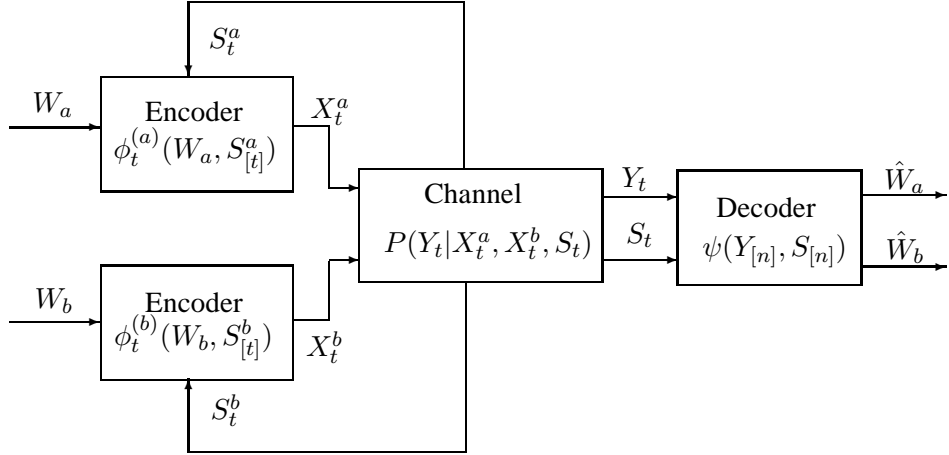
The system's probability of error, $P_{e}^{(n)}$, is given by
\begin{eqnarray}
P_{e}^{(n)}=\frac{1}{2^{n(R_a+R_b)}}\sum_{w_a=1}^{2^{nR_a}}\sum_{w_b=1}^{2^{nR_b}}P\left(\psi(Y_{[n]},S_{[n]})\neq (w_a,w_b)| \mathbf{W}=\mathbf{w}\right).\nonumber
\end{eqnarray}
A rate pair $(R_a, R_b)$ is achievable if for any $\epsilon > 0$, there exists, for all $n$ sufficiently large an $(n, 2^{nR_a}, 2^{nR_b})$ code such that $\frac{1}{n}\log |{\cal W}_a| \geq R_a > 0$, $\frac{1}{n}\log |{\cal W}_b| \geq R_b > 0$ and $P_{e}^{(n)} \leq \epsilon$. The capacity region of the FS-MAC, ${\cal C}_{FS}$, is the closure of the set of all achievable rate pairs $(R_a, R_b)$ and the sum-rate capacity is defined as ${\cal C}_{FS}^{\sum}:=\max_{(R_a,R_b) \in {\cal C}_{FS}}(R_a+R_b)$.

Before proceeding with the main result, we introduce \textit{memoryless stationary team policies} \cite{como-yuksel} and their associated rate regions. Let the set of all possible functions from ${\cal S}_{a}$ to ${\cal X}_{a}$ and ${\cal S}_{b}$ to ${\cal X}_{b}$  be denoted by ${\cal T}_{a}:={{\cal X}_{a}}^{{\cal S}_a}$ and ${\cal T}_{b}:={{\cal X}_{b}}^{{\cal S}_b}$, respectively.
We shall refer to ${\cal T}_a$-valued and ${\cal T}_b$-valued random vectors as Shannon strategies.
\begin{definition}\cite{como-yuksel}\label{def:team-pol-shst}
A memoryless stationary (in time) team policy is a family
\begin{eqnarray}
\Pi=\left\{\pi=\left(\pi_{T^a}(\cdot),\pi_{T^b}(\cdot)\right)\in {\cal P}({\cal T}_a)\times{\cal P}({\cal T}_b)\right\}\label{eq:team-pol-shst}
\end{eqnarray}
of probability distribution pairs on $({\cal T}_a,{\cal T}_b)$.

For every memoryless stationary team policy $\pi$, let $\mathcal{R}_{FS}(\pi)$ denote the region of all rate pairs $R=(R_a,R_b)$ satisfying
\begin{eqnarray}
R_a &<& I(T^{a};Y|T^{b},S) \quad \label{eq:ra1-shst-fl}\\
R_b &<& I(T^{b};Y|T^{a},S) \quad  \label{eq:ra2-shst-fl}\\
R_a+R_b &<& I(T^a,T^{b};Y|S) \label{eq:ra3-shst-fl}
\end{eqnarray}
where $S$, $T^a$, $T^b$ and $Y$ are random variables taking values in ${\cal S}$, ${\cal T}_{a}$, ${\cal T}_{b}$ and $\cal Y$, respectively, and whose joint probability distribution factorizes as
\begin{eqnarray}
P_{S,T^a,T^b,Y}(s,t^a,t^b,y)=P_S(s)P_{Y|T^a,T^b,S}(y|t^a,t^b,s)\pi_{T^a}(t^a)\pi_{T^b}(t^b)\label{eq:joi-dist-shst-fl}.
\end{eqnarray}
\end{definition}
Let ${\cal C}_{IN}:=\overline{co}\bigg(\bigcup_{\pi}\mathcal{R}_{FS}(\pi)\bigg)$ denote the closure of the convex hull of the rate regions $\mathcal{R}_{FS}(\pi)$ given by (\ref{eq:ra1-shst-fl})-(\ref{eq:ra3-shst-fl}) associated to all possible memoryless stationary team polices as defined in (\ref{eq:team-pol-shst}). We now present an inner bound and an outer bound to the capacity region. The latter bound is obtained by providing a tight converse to the sum-rate capacity.
\begin{theorem}[Inner Bound to ${\cal C}_{FS}$]\label{theo:inbnd-cau-fl}
 ${\cal C}_{IN}\subseteq {\cal C}_{FS}$.
\end{theorem}
The achievability proof follows the standard arguments of joint $\epsilon$-typical $n$-sequences \cite[Section 15.2]{cover}.
\begin{definition}\cite{cover}\label{def:jnt-typ}
Fix integer $k\geq1$. The set ${\cal A}_{\epsilon}^{n}$ of $\epsilon$-typical $n$-sequences $\{(x^1_{[n]},\cdots,x^k_{[n]})\}$ with respect to the distribution $P_{X^1,\cdots,X^k}(x^1, \cdots, x^k)=\prod_{i=1}^k P_{X^i}(x^i)$ is defined by
\begin{eqnarray}
{\cal A}_{\epsilon}^{n}=\left\{(x^1_{[n]},\cdots,x^k_{[n]})\in {\cal X}^1\times \cdots {\cal X}^k:|-\frac{1}{n}\log\left(P(\textbf{u})\right)-H(U)|<\epsilon, \forall U \subseteq \{X^1,\cdots,X^k\}\right\} \nonumber
\end{eqnarray}
where $\textbf{u}$ denotes an ordered sequence in $x^1_{[n]},\cdots,x^k_{[n]}$ corresponding to $U$.
\end{definition}
\begin{proof}[Proof of Theorem \ref{theo:inbnd-cau-fl}]
Fix $(R_a,R_b)\in \mathcal{R}_{FS}(\pi)$.

\textbf{\textit{Codebook Generation}}
Fix $\pi_{T^a}(t^a)$ and $\pi_{T^b}(t^b)$. For each $w_a \in \{1,\cdots, 2^{nR_a}\}$, randomly generate its corresponding $n$-tuple $t_{[n],w_a}^{a}$, each according to $\prod_{i=1}^n\pi_{T_{i}^{a}}(t_{i,w_a}^{a})$. Similarly, For each $w_b \in \{1,\cdots, 2^{nR_b}\}$, randomly generate its corresponding $n$-tuple $t_{[n],w_b}^{b}$, each according to $\prod_{i=1}^n\pi_{T_{i}^{b}}(t_{i,w_b}^{b})$. The set of these codeword pairs form the codebook, which is revealed to the decoder while codewords $t_{i,w_l}^{l}$ are revealed to encoder $l$, $l=\{a,b\}$.

\textbf{\textit{Encoding}}
Define the encoding functions as follows: $x_{i}^{a}(w_a)=\phi_{i}^{a}(w_a,s_{[i]}^a)=t_{i,w_a}^a(s_i^a)$ and $x_{i}^{b}(w_b)=\phi_{i}^{b}(w_b,s_{[i]}^b)=t_{i,w_b}^b(s_i^b)$ where $t_{i,w_a}^a$ and $t_{i,w_b}^b$ denote the $i$th component of $t_{[n],w_a}^{a}$ and $t_{[n],w_b}^{b}$, respectively, and $s_i^a$ and $s_i^b$ denote the last components of $s_{[i]}^a$ and $s_{[i]}^b$, respectively, $i=1,\cdots,n$.
Therefore, to send the messages $w_a$ and $w_b$, we simply transmit the corresponding $t_{[n],w_a}^{a}$ and $t_{[n],w_b}^{b}$, respectively.

\textbf{\textit{Decoding}}
After receiving $(y_{[n]},s_{[n]})$, the decoder looks for the only pair $(w_a,w_b)$ such that $(t_{[n],w_a}^{a},t_{[n],w_b}^{b}$ $,y_{[n]},s_{[n]})$ are jointly $\epsilon-$typical and declares this pair as its estimate $(\hat{w}_a,\hat{w}_b)$.

\textbf{\textit{Error Analysis}}
Without loss of generality, we can assume that $(w_a,w_b)=(1,1)$ was sent. An error occurs, if the correct codewords are not typical with the received sequence or there is a pair of incorrect codewords that are typical with the received sequence. Define the events $E_{\alpha,\beta}\bydef\big\{(T_{[n],\alpha}^a,T_{[n],\beta}^b,Y_{[n]},S_{[n]})\in {\cal A}_{\epsilon}^n\big\}$, $\alpha\in\{1,\cdots,2^{nR_a}\}$ and $\beta\in\{1,\cdots,2^{nR_b}\}$.
Then, by the union bound we get{\allowdisplaybreaks
\begin{eqnarray}
P_{e}^{n}&=&P\big(E_{1,1}^c\bigcup_{(\alpha,\beta)\neq(1,1)}E_{\alpha,\beta}\big)\nonumber\\
&\leq& P(E_{1,1}^c)+\sum_{\alpha=1,\beta\neq1}P(E_{\alpha,\beta}) + \sum_{\alpha\neq1,\beta=1}P(E_{\alpha,\beta})+ \sum_{\alpha\neq1,\beta\neq1}P(E_{\alpha,\beta})\label{eq:erbound-cau-fl}
\end{eqnarray}}where $E_{1,1}^c$ denotes the complement set of $E_{1,1}$. It can easily be verified that $\{Y_i,S_i,T_i^a,T_i^b\}_{i=1}^{\infty}$ is an i.i.d. sequence and by \cite[Theorem 15.2.1]{cover}, $ P(E_{1,1}^c)\rightarrow0$ as $n\rightarrow \infty$.
Next, let us consider the second term
\begin{eqnarray}
\sum_{\alpha=1,\beta\neq1}P(E_{\alpha=1,\beta\neq1})&=&\sum_{\alpha=1,\beta\neq1}P((T_{[n],1}^a,T_{[n],\beta}^b,Y_{[n]},S_{[n]})\in A_{\epsilon}^n)\nonumber\\
&\overset{(i)}{=}&\sum_{\alpha=1,\beta\neq1}\sum_{(t_{[n]}^a,t_{[n]}^b,y_{[n]},s_{[n]})\in A_{\epsilon}^n}P_{T_{[n]}^b}(t_{[n]}^b)P_{T_{[n]}^a, Y_{[n]},S_{[n]}}(t_{[n]}^a, y_{[n]},s_{[n]})\nonumber\\
&\overset{}{\leq}&\sum_{\alpha=1,\beta\neq1}|A_{\epsilon}^n|2^{-n[H(T^b)-\epsilon]}2^{-n[H(T^a,Y,S)-\epsilon]}\label{eq:erbo0-cau-fl}\\
&\leq&2^{nR_b}2^{-n[H(T^b)+ H(T^a,Y,S) - H(T^a,T^b,Y,S)-3\epsilon]}\nonumber\\
&\overset{(ii)}{=}&2^{n[R_b-I(T^b;Y|S,T^a)-3\epsilon]}\label{eq:erbo1-cau-fl}
\end{eqnarray}
where $(i)$ holds since for $\beta\neq 1$, $T_{[n],\beta}^b$ is independent of $(T_{[n],1}^a,Y_{[n]},S_{[n]})$ and $(ii)$ follows since $T^b$ and $(T^a,S)$ are independent and $I(T^b;Y,T^a,S)=I(T^b;T^a,S)+I(T^b;Y|T^a,S)$
$=I(T^b;Y|T^a,S)$, where $I(T^b;T^a,S)=0$. Following the same steps for $(\alpha\neq1,\beta=1)$ and $(\alpha\neq1,\beta\neq1)$ we get
\begin{eqnarray}
\sum_{\alpha\neq1,\beta=1}P(E_{\alpha,\beta})\leq2^{n[R_a-I(T^a;Y|T^b,S)-3\epsilon]}, \enspace\sum_{\alpha\neq1,\beta\neq1}P(E_{\alpha,\beta})\leq2^{n[R_a+R_b-I(T^a,T^b;Y|S)-3\epsilon]} \label{eq:erbo2-cau-fl},
\end{eqnarray}
and the rate conditions of the $\mathcal{R}_{FS}(\pi)$ imply that each term tends in (\ref{eq:erbound-cau-fl}) tends to zero as $n \rightarrow \infty$. This shows the achievability of a rate pair $(R_a, R_b) \in \mathcal{R}_{FS}(\pi)$. Achievability of any rate pair in ${\cal C}_{IN}$ follows from a standard time-sharing argument.
\end{proof}
Let
\begin{eqnarray}
{\cal C}_{OUT}:=\biggl\{(R_a,R_b)\in \mathds{R}^{+}\times\mathds{R}^{+}: R_a+R_b\leq\sup_{\pi_{T^a}(t^a)\pi_{T^b}(t^b)}I(T^a,T^b;Y|S)\biggr\}\nonumber,
\end{eqnarray}
where $\mathds{R}^{+}$ is the set of positive reals.
\begin{theorem}[Outer Bound to ${\cal C}_{FS}$]\label{theo:outbnd-cau-fl}
${\cal C}_{FS}\subseteq {\cal C}_{OUT}$.
\end{theorem}
As a consequence of Theorems \ref{theo:inbnd-cau-fl} and \ref{theo:outbnd-cau-fl}, we have the following corollary which can be thought of as an extension of \cite[Theorem 4]{Jafar} to the case where the encoders have correlated CSI.
\begin{corollary}\label{cor:sumra-cau-fl}
\begin{eqnarray}
{\cal C}^{FS}_{\sum}=\sup_{\pi_{T^a}(t^a)\pi_{T^b}(t^b)}I(T^a,T^b;Y|S)\label{eq:sumra-cau-fl}.
\end{eqnarray}
\end{corollary}
\begin{proof}[Proof of Theorem \ref{theo:outbnd-cau-fl}]
We need to show that all achievable rates satisfy
\begin{eqnarray}
R_a+R_b \leq \sup_{\pi_{T^a}(t^a)\pi_{T^b}(t^b)}I(T^a,T^b;Y|S),\nonumber
\end{eqnarray}
i.e., a converse for the sum-rate capacity. Following \cite{como-yuksel}, let
\begin{eqnarray}
\alpha_{\mathbf{\mu}}:= \frac{1}{n}P_{S_{[t-1]}}(\mathbf{\mu})\enspace\mbox{and}\enspace \eta(\epsilon):=\frac{\epsilon}{1-\epsilon}\log|{\cal Y}|+\frac{H(\epsilon)}{1-\epsilon}\label{eq:eta-epsilon}.
\end{eqnarray}
Observe that $\lim_{\epsilon\rightarrow 0}\eta(\epsilon)=0$ and
\begin{eqnarray}
\sum_{\mathbf{\mu} \in {{\cal S}}^{(n)}}\alpha_{\mathbf{\mu}}=\frac{1}{n}\sum_{1\leq t \leq n}\enspace \sum_{\mathbf{\mu} \in {{\cal S}}^{(t-1)}}P_{S_{[t-1]}}(\mathbf{\mu})=1\nonumber,
\end{eqnarray}
where ${{\cal S}}^{(n)}$ and ${{\cal S}}^{(t-1)}$ are the sets of all ${\cal S}$-strings of length $n$ and $(t-1)$, respectively.

First recall that, $\forall t\ge1$, $X_t^a=\phi_{t}^{(a)} \left(W_a,S_{[t]}^a\right)=\phi_{t}^{(a)} \left(W_a,S_{[t-1]}^a,S_t^a\right)$ and $X_t^b=\phi_{t}^{(b)} \left(W_b,S_{[t]}^b\right)=\phi_{t}^{(b)} \left(W_b,S_{[t-1]}^b, S_t^b\right)$.
Then, we can define the Shannon strategies $T_t^a\in\mathcal T_a$ and $T_t^b\in\mathcal T_b$ by putting, for every $s_a\in\mathcal S_a$ and $s_b\in\mathcal S_b$,
\begin{equation}
T_t^a(s_a):=\phi_{t}^{(a)}\left(W_a,S_{[t-1]}^a,s_a\right),\qquad T_t^b(s_b):=\phi_{t}^{(b)}\left(W_b,S_{[t-1]}^b,s_b\right) . \label{eq:conv-cau-fl-1}
\end{equation}
We now show that the sum of any achievable rate pair can be written as the convex combinations of conditional mutual information terms which are indexed by the realization of past complete state information.
\begin{lemma}\label{lem:conv-cau-fl}
Let $T_t^{a} \in {\cal T}_{a}$ and $T_t^b \in {\cal T}_{b}$ be the Shannon strategies induced by $\phi_{t}^{(a)}$ and $\phi_{t}^{(b)}$, respectively, as shown in (\ref{eq:conv-cau-fl-1}). Assume that a rate pair $R=(R_a,R_b)$, with block length $n\geq1$ and a constant $\epsilon \in (0,1/2)$, is achievable. Then,
\begin{eqnarray}
R_a+R_b\leq \sum_{\mathbf{\mu} \in {\cal S}^{(n)}}\alpha_{\mathbf{\mu}} I(T_t^a,T_t^b;Y_t|S_t,S_{[t-1]}=\mathbf{\mu})+\eta(\epsilon) \label{eq:conv-cau-fl-2}.
\end{eqnarray}
\end{lemma}
\begin{proof}
Let $\mathbf{T}_t:=(T_t^a,T_t^b)$. By Fano's inequality, we get
\begin{eqnarray}
H(\mathbf{W}|Y_{[n]},S_{[n]})\leq H(\epsilon) + \epsilon\log(|{\cal W}_a||{\cal W}_b|). \label{eq:conv-cau-fl-3}
\end{eqnarray}
Observing that
\begin{eqnarray}
I(\mathbf{W};Y_{[n]}, S_{[n]})&=&H(\mathbf{W})-H(\mathbf{W}|Y_{[n]},S_{[n]})\nonumber\\
&=&\log(|{\cal W}_a||{\cal W}_b|)-H(\mathbf{W}|Y_{[n]},S_{[n]}).\label{eq:conv-cau-fl-4}
\end{eqnarray}
Combining (\ref{eq:conv-cau-fl-3}) and (\ref{eq:conv-cau-fl-4}) gives
\begin{eqnarray}
(1-\epsilon)\log(|{\cal W}_a||{\cal W}_b|)\leq I(\mathbf{W};Y_{[n]}, S_{[n]})+H(\epsilon)\nonumber
\end{eqnarray}
and
\begin{eqnarray}
R_a+R_b&\leq& \frac{1}{n}\log(|{\cal W}_a||{\cal W}_b|)\leq \frac{1}{1-\epsilon}\frac{1}{n}\left(I(\mathbf{W};Y_{[n]}, S_{[n]})+H(\epsilon)\right)\label{eq:conv-cau-fl-5}.
\end{eqnarray}
Furthermore,
\begin{eqnarray}
I(\mathbf{W};Y_{[n]},S_{[n]})&=&\sum_{t=1}^{n}\left[H(Y_t,S_t|S_{[t-1]},Y_{[t-1]})-H(Y_t,S_t|\mathbf{W},S_{[t-1]},Y_{[t-1]})\right]\nonumber\\
&\overset{(i)}{=}&\sum_{t=1}^{n}\left[H(Y_t|S_{[t]},Y_{[t-1]})-H(Y_t|\mathbf{W},S_{[t]},Y_{[t-1]})\right]\nonumber\\
&\overset{(ii)}{\leq}&\sum_{t=1}^{n}\left[H(Y_t|S_{[t]})-H(Y_t|\mathbf{W},S_{[t]},Y_{[t-1]},\mathbf{T}_t)\right]\nonumber\\
&\overset{(iii)}{=}&\sum_{t=1}^{n}\left[H(Y_t|S_{[t]})-H(Y_t|S_{[t]},\mathbf{T}_t)\right]\nonumber\\
&=&\sum_{t=1}^{n}I(\mathbf{T}_t;Y_t|S_{[t]})\label{eq:conv-cau-fl-6}
\end{eqnarray}
where $(i)$ is implied by (\ref{eq:sta-iid}), in $(ii)$ $\mathbf{T}_t:=(T_t^a,T_t^b)$ are Shannon strategies whose realizations are mappings $t_t^i:S_t^{i}\rightarrow X_t^i$ for $i=\{a,b\}$ and thus $(ii)$ holds since conditioning reduces entropy. Finally, $(iii)$ follows since
\begin{eqnarray}
&&\hspace{-1in}P_{Y_t|\mathbf{W},S_t,S_{[t-1]},Y_{[t-1]},T_t^a,T_t^b}(y_t|\mathbf{w},s_t,s_{[t-1]},y_{[t-1]},t_t^a,t_t^b)\nonumber\\
&=&\sum_{s_t^a,s_t^b}P_{Y_t|S_t,S_t^a,S_t^b,T_t^a,T_t^b}(y_t|s_t,s_t^a,s_t^b,t_t^a,t_t^b) P_{S_t^a,S_t^b|S_t}(s_t^a,s_t^b|s_t)\nonumber\\
&=&P_{Y_t|S_t,T_t^a,T_t^b}(y_t|s_t,t_t^a,t_t^b)\label{eq:conv-cau-fl-6a}
\end{eqnarray}
where the first equality is verified by (\ref{eq:ch-recfl}) and (\ref{eq:sta-iid}), where $x_t^i=t_t^i(s_t^i)$ for $i=\{a,b\}$. At this point, it is worth to note that by (\ref{eq:conv-cau-fl-6a}), one can remove $S_{[t-1]}$ from (\ref{eq:conv-cau-fl-6}) in the conditioning. However, we will soon observe why it is crucial to keep it when we prove the product form. Now, let $\chi(\epsilon):=\frac{H(\epsilon)}{n(1-\epsilon)}$ and combining (\ref{eq:conv-cau-fl-5})-(\ref{eq:conv-cau-fl-6}) gives{\allowdisplaybreaks
\begin{eqnarray}
R_a+R_b&\leq&\frac{1}{n}\log(|{\cal W}_a||{\cal W}_b|)\nonumber\\
&\leq&\left(\frac{1}{1-\epsilon}\frac{1}{n}\sum_{t=1}^nI(T_t^a,T_t^b;Y_t|S_{[t]})\right)+\chi(\epsilon)+(n-1)\chi(\epsilon)\nonumber\\
&\overset{(a)}{\leq}&\frac{1}{1-\epsilon}\frac{1}{n}\sum_{t=1}^nI(T_t^a,T_t^b;Y_t|S_{[t]})+\eta(\epsilon)-\frac{\epsilon}{1-\epsilon}\frac{1}{n}\sum_{t=1}^nI(T_t^a,T_t^b;Y_t|S_{[t]})\nonumber\\
&=&\frac{1}{n}\sum_{t=1}^nI(T_t^a,T_t^b;Y_t|S_{[t]})+\eta(\epsilon)\label{eq:conv-cau-fl-7}
\end{eqnarray}}where $(a)$ is valid since $I(T_t^a,T_t^b;Y_t|S_{[t]})\leq \log|{\cal Y}|$. Furthermore,
\begin{eqnarray}
I(T_t^a,T_t^b;Y_t|S_{[t]})=n\sum_{\mathbf{\mu} \in {\cal S}^{(t-1)}}\alpha_{\mathbf{\mu}}I(T_t^a,T_t^b;Y_t|S_t, S_{[t-1]}=\mathbf{\mu}),
\end{eqnarray}
and substituting the above into (\ref{eq:conv-cau-fl-7}) yields (\ref{eq:conv-cau-fl-2}).
\end{proof}
Note that, for any $t\geq1$, $I(T_t^a,T_t^b;Y_t|S_t,S_{[t-1]}=\mathbf{\mu})$ is a function of the joint conditional distribution of channel state $S_t$, inputs $T_t^a$, $T_t^b$ and output $Y_t$ given the past realization $(S_{[t-1]}=\mathbf{\mu})$. Hence, to complete the proof of the outer bound, we need to show that $P_{T_t^a,T_t^b,Y_t,S_t|S_{[t-1]}}(t^a,t^b,y,s|\mathbf{\mu})$ factorizes as in (\ref{eq:joi-dist-shst-fl}). This is done in the lemma below. In particular, it is crucial to observe that the knowledge of the past state at the decoder, $S_{[t-1]}$, is enough to provide a product form on $T^a$ and $T^b$. Let
\begin{eqnarray}
\Upsilon_{\mathbf{\mu_a}}^a(t^a):=\{w_a:\phi_{t}^{(a)}(w_a,s_{[t-1]}^a=\mathbf{\mu_a})=t^a\}, \enspace\Upsilon_{\mathbf{\mu_b}}^b(t^b):=\{w_b:\phi_{t}^{(b)}(w_b,s_{[t-1]}^b=\mathbf{\mu_b})=t^b\}\label{eq:conv-fact-fl-1}
\end{eqnarray}
and
\begin{eqnarray}
\pi_{T^a}^{\mathbf{\mu_a}}(t^a)&:=&\sum_{w_a\in \Upsilon_{\mathbf{\mu_a}}^a(t^a)}\frac{1}{|{\cal W}_a|}, \enspace\pi_{T^b}^{\mathbf{\mu_b}}(t^b)\quad:=\quad\sum_{w_b\in \Upsilon_{\mathbf{\mu_b}}^b(t^b)}\frac{1}{|{\cal W}_b|}, \nonumber\\
\pi_{T^a}^{\mathbf{\mu}}(t^a)&:=&\sum_{\mathbf{\mu_a}}\pi_{T^a}^{\mathbf{\mu_a}}(t^a)P_{S_{[t-1]}^a|S_{[t-1]}}(\mathbf{\mu_a}|\mathbf{\mu}),\nonumber\\ \pi_{T^b}^{\mathbf{\mu}}(t^b)&:=&\sum_{\mathbf{\mu_b}}\pi_{T^b}^{\mathbf{\mu_b}}(t^b)P_{S_{[t-1]}^b|S_{[t-1]}}(\mathbf{\mu_b}|\mathbf{\mu}),\label{eq:conv-fact-fl-2}
\end{eqnarray}
where $\mathbf{\mu_a}$ and $\mathbf{\mu_b}$ denote particular realizations of $S_{[t-1]}^a$ and $S_{[t-1]}^b$, respectively.
\begin{lemma}\label{lem:fact-cau-fl}
For every $1\leq t\leq n$ and $\mathbf{\mu} \in {\cal S}^{t-1}$, the following holds
\begin{eqnarray}
P_{T_t^a,T_t^b,Y_t,S_t|S_{[t-1]}}(t^a,t^b,y,s|\mathbf{\mu})&=&P_S(s)P_{Y|S,T^a,T^b}(y|s,t^a,t^b)\pi_{T^a}^{\mathbf{\mu}}(t^a)\pi_{T^b}^{\mathbf{\mu}}(t^b).\label{eq:conv-fact-fl-3}
\end{eqnarray}
\end{lemma}
\begin{proof}
Let $\textbf{S}:=(S_t,S_t^a,S_t^b)$ and $\textbf{s}:=(s,s_t^a,s_t^b)$. Observe that{\allowdisplaybreaks
\begin{eqnarray}
\hspace{-0.2in}P_{T_t^a,T_t^b,Y_t,S_t|S_{[t-1]}}(t^a,t^b,y,s|\mathbf{\mu})&=&\sum_{s_t^a\in{\cal S}^a}\sum_{s_t^b\in{\cal S}^b}
P_{\textbf{S},T_t^a,T_t^b,Y_t|S_{[t-1]}}(\textbf{s},t^a,t^b,y|\mathbf{\mu})\nonumber\\
&=&\sum_{s_t^a\in{\cal S}^a}\sum_{s_t^b\in{\cal S}^b}P_{Y|\textbf{S},T_t^a,T_t^b}(y|\textbf{s},t^a,t^b) P_{\textbf{S},T_t^a,T_t^b|S_{[t-1]}}(\textbf{s},t^a,t^b|\mathbf{\mu})\label{eq:conv-fact-fl-4}
\end{eqnarray}
where the second equality is shown in (\ref{eq:conv-cau-fl-6a}).} Let us now consider the term $P_{\textbf{S},T_t^a,T_t^b|S_{[t-1]}}(\textbf{s},t^a,t^b|\mathbf{\mu})$ above. We have the following
\begin{eqnarray}
&&\hspace{-0.5in}P_{\textbf{S},T_t^a,T_t^b|S_{[t-1]}}(\textbf{s},t^a,t^b|\mathbf{\mu})\nonumber\\
&=&\sum_{w_a\in {\cal W}_a}\sum_{w_b\in {\cal W}_b}\sum_{\mathbf{\mu_a}}\sum_{\mathbf{\mu_b}}P_{\mathbf{W},S_{[t-1]}^a,S_{[t-1]}^b,\textbf{S},T_t^a,T_t^b|S_{[t-1]}}(\mathbf{w},\mathbf{\mu_a},\mathbf{\mu_b},\textbf{s},t^a,t^b|\mathbf{\mu})\nonumber\\
&\overset{(i)}{=}&P_{\textbf{S}}(\textbf{s})\sum_{w_a\in {\cal W}_a}\sum_{w_b\in {\cal W}_b}\sum_{\mathbf{\mu_a}}\sum_{\mathbf{\mu_b}}P_{\mathbf{W},S_{[t-1]}^a,S_{[t-1]}^b,T_t^a,T_t^b|S_{[t-1]}}(\mathbf{w},\mathbf{\mu_a},\mathbf{\mu_b},t^a,t^b|\mathbf{\mu})\nonumber\\
&\overset{(ii)}{=}&P_{\textbf{S}}(\textbf{s})\sum_{w_a\in {\cal W}_a}\sum_{w_b\in {\cal W}_b}\sum_{\mathbf{\mu_a}}\sum_{\mathbf{\mu_b}}1_{\{t^l=\phi_{t}^{(l)}(w_l,\mathbf{\mu_l}),\enspace l=a,b\}} P_{\mathbf{W},S_{[t-1]}^a,S_{[t-1]}^b|S_{[t-1]}}(\mathbf{w},\mathbf{\mu_a},\mathbf{\mu_b}|\mathbf{\mu})\nonumber\\
&\overset{(iii)}{=}&P_{\textbf{S}}(\textbf{s})\sum_{w_a\in {\cal W}_a}\sum_{w_b\in {\cal W}_b}\sum_{\mathbf{\mu_a}}\sum_{\mathbf{\mu_b}}1_{\{t^l=\phi_{t}^{(l)}(w_l,\mathbf{\mu_l}),\enspace l=a,b\}}\frac{1}{|{\cal W}_a|}\frac{1}{|{\cal W}_b|}P_{S_{[t-1]}^a,S_{[t-1]}^b|S_{[t-1]}}(\mathbf{\mu_a},\mathbf{\mu_b}|\mathbf{\mu})\nonumber\\
&\overset{(iv)}{=}&P_{\textbf{S}}(\textbf{s})\sum_{\mathbf{\mu_a}}P_{S_{[t-1]}^a|S_{[t-1]}}(\mathbf{\mu_a}|\mathbf{\mu})\sum_{\mathbf{\mu_b}}P_{S_{[t-1]}^b|S_{[t-1]}}(\mathbf{\mu_b}|\mathbf{\mu})\nonumber\\
&&\hspace{1.4in}\sum_{w_a\in {\cal W}_a} \frac{1}{|{\cal W}_a|}1_{\{t^a=\phi_{t}^{(a)}(w_a,\mathbf{\mu_a})\}}\sum_{w_b\in {\cal W}_b} \frac{1}{|{\cal W}_b|}1_{\{t^b=\phi_{t}^{(b)}(w_b,\mathbf{\mu_b})\}}\nonumber\\
&\overset{(v)}{=}&P_{\textbf{S}}(\textbf{s})\sum_{\mathbf{\mu_a}}P_{S_{[t-1]}^a|S_{[t-1]}}(\mathbf{\mu_a}|\mathbf{\mu}) \sum_{w_a\in \Upsilon_{\mathbf{\mu_a}}^a(t^a)}\frac{1}{|{\cal W}_a|}\sum_{\mathbf{\mu_b}}P_{S_{[t-1]}^b|S_{[t-1]}}(\mathbf{\mu_b}|\mathbf{\mu})\sum_{w_b\in\Upsilon_{\mathbf{\mu_b}}^b(t^b)}\frac{1}{|{\cal W}_b|}\nonumber\\
&\overset{(vi)}{=}&P_{\textbf{S}}(\textbf{s})\sum_{\mathbf{\mu_a}}P_{S_{[t-1]}^a|S_{[t-1]}}(\mathbf{\mu_a}|\mathbf{\mu})\pi_{T^a}^{\mathbf{\mu_a}}(t^a)\sum_{\mathbf{\mu_b}}P_{S_{[t-1]}^b|S_{[t-1]}}(\mathbf{\mu_b}|\mathbf{\mu})\pi_{T^b}^{\mathbf{\mu_b}}(t^b)\nonumber\\
&\overset{(vii)}{=}&P_{\textbf{S}}(\textbf{s})\pi_{T^a}^{\mathbf{\mu}}(t^a)\pi_{T^b}^{\mathbf{\mu}}(t^b)\label{eq:conv-fact-fl-5}
\end{eqnarray}
where $(i)$ is due to (\ref{eq:sta-iid}) and (\ref{eq:conv-cau-fl-1}), $(ii)$ is valid by (\ref{eq:conv-cau-fl-1}), $(iii)$ is due to (\ref{eq:sta-iid}), $(iv)$ is valid by (\ref{eq:noi-recfl-mrk}) and (\ref{eq:conv-cau-fl-1}), $(v)$ is valid due to (\ref{eq:conv-fact-fl-1}) and $(vi)-(vii)$ is valid due to (\ref{eq:conv-fact-fl-2}). Substituting (\ref{eq:conv-fact-fl-5}) into (\ref{eq:conv-fact-fl-4}) proves the lemma.
\end{proof}

We can now complete the proof of Theorem \ref{theo:outbnd-cau-fl}. With Lemma \ref{lem:conv-cau-fl} it is shown that the sum of any achievable rate pair can be approximated by the convex combinations of rate conditions given in (\ref{eq:ra3-shst-fl}) which are indexed by $\mathbf{\mu} \in {\cal S}^{(n)}$ and satisfy (\ref{eq:joi-dist-shst-fl}) for joint state-input-output distributions.
More explicitly, we have{\allowdisplaybreaks
\begin{eqnarray}
R_a+R_b&\leq& \sum_{\mathbf{\mu} \in {\cal S}^{(n)}}\alpha_{\mathbf{\mu}} I(T_t^a,T_t^b;Y_t|S_t,S_{[t-1]}=\mathbf{\mu})+\eta(\epsilon)\nonumber\\
&=&\sum_{\mathbf{\mu} \in {\cal S}^{(n)}}\alpha_{\mathbf{\mu}} I(T_t^a,T_t^b;Y_t|S_t)_{\pi_{T^a}^{\mathbf{\mu}}(t^a)\pi_{T^b}^{\mathbf{\mu}}(t^b)}+\eta(\epsilon)\nonumber\\
&\leq&\sup_{\left(\pi_{T^a}^{\mathbf{\mu}}(t^a)\pi_{T^b}^{\mathbf{\mu}}(t^b),\enspace\mathbf{\mu}\right)}I(T_t^a,T_t^b;Y_t|S_t)+\eta(\epsilon)\nonumber\\
&\leq&\sup_{\left(\pi_{T^a}(t^a)\pi_{T^b}(t^b)\in \Pi\right)}I(T_t^a,T_t^b;Y_t|S_t)+\eta(\epsilon)\nonumber,
\end{eqnarray}}where $I(T_t^a,T_t^b;Y_t|S_t)_{\pi_{T^a}^{\mathbf{\mu}}(t^a)\pi_{T^b}^{\mathbf{\mu}}(t^b)}$ denotes the mutual information induced by the product distribution $\pi_{T^a}^{\mathbf{\mu}}(t^a)\pi_{T^b}^{\mathbf{\mu}}(t^b)$ and the second step is valid since $I(T_t^a,T_t^b;Y_t|S_t,S_{[t-1]}=\mathbf{\mu})$ is a function of the joint conditional distribution of channel state $S_t$, inputs $T_t^a, T_t^b$ and output $Y_t$ given the past realization $(S_{[t-1]}=\mathbf{\mu})$. Hence, since $\lim_{\epsilon \rightarrow 0}\eta(\epsilon)=0$, any achievable pair satisfies $R_a+R_b \leq \sup_{\pi_{T^a}(t^a)\pi_{T^b}(t^b)}I(T^a,T^b;Y|S)$.
\end{proof}
Having achievability and converse proof in hand, we can now prove Corollary \ref{cor:sumra-cau-fl}.
\begin{proof}[Proof of Corollary \ref{cor:sumra-cau-fl}]
We need to show that $\exists \enspace (R_a,R_b) \in {\cal C}_{IN}$ achieving (\ref{eq:sumra-cau-fl}). We follows steps akin to \cite[p.535]{cover} where discrete memoryless MACs are considered. Let us fix $\pi_{T^a}(t^a)\pi_{T^b}(t^b)$ and consider the rate constraints given in ${\cal C}_{IN}$
\begin{eqnarray}
I(T^a;Y|T^b,S)&=&H(T^a|T^b,S)-H(T^a|T^b,Y,S)=H(T^a)-H(T^a|T^b,Y,S)\label{eq:cor-sumra-fl-1} \\
I(T^b;Y|T^a,S)&=&H(T^b|T^a,S)-H(T^b|T^a,Y,S)=H(T^b)-H(T^b|T^a,Y,S)\label{eq:cor-sumra-fl-2}
\end{eqnarray}
and
\begin{eqnarray}
I(T^a,T^b;Y|S)&=&H(T^a,T^b)-H(T^a,T^b|Y,S)\nonumber\\
&=&H(T^a)+H(T^b)-H(T^a|T^b,Y,S)-H(T^b|Y,S),\label{eq:cor-sumra-fl-3}
\end{eqnarray}
where (\ref{eq:cor-sumra-fl-1}), (\ref{eq:cor-sumra-fl-2}) and (\ref{eq:cor-sumra-fl-3}) are valid since $T^a$ and $T^b$ are independent of each other and independent of $S$. Observe now that for any $\pi_{T^a}(t^a)\pi_{T^b}(t^b)$, $I(T^a;Y|T^b,S) + I(T^b;Y|T^a,S)\geq I(T^a,T^b;Y|S)$ since $H(T^b|Y,S)\geq H(T^b|T^a,Y,S)$. Therefore, the sum-rate constraint in $ {\cal C}_{IN}$ is always active and hence, there exists $(R_a,R_b) \in {\cal C}_{IN}$ achieving (\ref{eq:sumra-cau-fl}).
\end{proof}
We now present a number of remarks.
\begin{remark}\label{rem:outbnd-cau-fl}
One essential step in the proof of Theorem \ref{theo:outbnd-cau-fl} is that, once we have the complete CSI, conditioning on which allows a product form on $T^a$ and $T^b$, there is no loss of optimality (for the sum-rate capacity) in using associated memoryless team policies instead of using all the past information at the receiver.
\end{remark}
\begin{remark}\label{rem:corr-cau-fl}
For the validity of Corollary \ref{cor:sumra-cau-fl}, it is crucial to have the product form on $(T^a,T^b)$. If this is not the case, we would get that $I(T^a;Y|T^b,S)+I(T^b;Y|T^a,S)=H(T^a|T^b)+H(T^b|T^a)-H(T^a|T^b,Y,S)-H(T^b|T^a,Y,S)$ and
$I(\mathbf{T};Y|S)=H(T^a|T^b)+H(T^b)-H(T^a|T^b,Y,S)-H(T^b|Y,S)$. Therefore, it is possible to get an obsolete sum-rate constraint in ${\cal C}_{IN}$ and hence, achievability of ${\cal C}_{FS}^{\sum}$ is not guaranteed.
\end{remark}
\begin{remark}\label{rem:diff-us-cy}
The main difference between the problem that we consider here and the one considered in \cite{como-yuksel} is the encoders' information at the decoder. More explicitly, in \cite{como-yuksel}, the information at the encoders are available at the decoder. From this perspective, the main contribution of the result of this section can be thought as showing that when the decoder has no knowledge of encoders' CSI, by enlarging the input space, there is no loss of optimality (for the sum-rate capacity) if the optimization is performed by ignoring the past CSI at the encoders given that the decoder has complete CSI.
\end{remark}
\subsection{Asymmetric Causal Noisy CSIT and Noisy CSIR}\label{sec:noisy}
In many practical applications, CSI first needs to be estimated by the receiver, such as using training methods, and then the receiver feeds back this information to the transmitters. This motivates us to consider a scenario where the decoder is first provided with noisy CSI (where the noise models the estimation error) and then, it feeds back this noisy CSI to the encoders thorough independent but noisy feedback links as shown in Fig. \ref{fig:mac-rec-noi}, where $N_a$ and $N_b$ denote independent noise processes. The two encoders have causal noisy versions of the state information $S_t$ at each time $t\geq1$, $S_{t}^{a} \in {\cal S}_{a}$, $S_{t}^{b} \in {\cal S}_{b}$, respectively, and the decoder has access to noisy CSI at time $t$ ,$S_t^r\in {\cal S}_{r}$. Based on the physical setup, the joint distribution of $(S_t,S_t^a,S_t^b,S_t^r)$ satisfies
\begin{eqnarray}
P_{S_{t}^{a}, S_{t}^{b},S_t^r,S_t}(s_{t}^{a}, s_{t}^{b}, s_t^r, s_t)=P_{S_{t}^{a}|S_t^r}(s_{t}^{a}|s_t^r)P_{S_{t}^{b}|S_t^r}(s_{t}^{b}|s_t^r)P_{S_t|S_t^r}(s_t|s_t^r)P_{S_t^r}(s_t^r)\label{eq:noi-recnoi-mrk}.
\end{eqnarray}
We also assume that the channel is memoryless (i.e., (\ref{eq:ch-recfl}) holds) and that
\begin{eqnarray}
&&\hspace{-1in}P_{S_{[n]},S_{[n]}^a,S_{[n]}^b,S_{[n]}^r,W_a,W_b}(s_{[n]},s_{[n]}^a,s_{[n]}^b,s_{[n]}^r, w_a,w_b)\nonumber\\
&&\quad\quad\quad=\hspace{0.1in}\prod_{t=1}^n \frac{1}{|{\cal W}_a|}\frac{1}{|{\cal W}_b|}P_{S_{t}^{a}|S_t^r}(s_{t}^{a}|s_t^r)P_{S_{t}^{b}|S_t^r}(s_{t}^{b}|s_t^r)P_{S_t,S_t^r}(s_t,s_t^r)\label{eq:star-iid}.
\end{eqnarray}
We first provide inner and outer bounds on the capacity region and an expression for the sum-rate capacity, akin to Theorems \ref{theo:inbnd-cau-fl}, \ref{theo:outbnd-cau-fl} and Corollary \ref{cor:sumra-cau-fl}, respectively, when the feedback links are noisy. In the next subsection, by assuming that the CSITs are asymmetric deterministic functions of CSIR we obtain the full capacity region.

A code can be defined as in Definition \ref{def:maccode-causal}, except $\psi: \mathcal{S}_{r}^n\times{\cal Y}^n\rightarrow \mathcal{W}_a \times \mathcal{W}_b$. $P_e^{(n)}$, achievable rates and the capacity region, ${\cal \tilde{C}}_{NS}$, are defined similarly. The sum-rate capacity is denoted by ${\cal \tilde{C}}_{NS}^{\sum}$. We also keep Definition \ref{def:team-pol-shst} and slightly change the associated rate region.
\begin{figure}
%
%
\setlength{\unitlength}{3947sp}%
\begingroup\makeatletter\ifx\SetFigFont\undefined%
\gdef\SetFigFont#1#2#3#4#5{%
  \reset@font\fontsize{#1}{#2pt}%
  \fontfamily{#3}\fontseries{#4}\fontshape{#5}%
  \selectfont}%
\fi\endgroup
\begin{picture}(6097,3566)(0,-3329)
\thinlines
{\color[rgb]{0,0,0}\put(3072,-1851){\framebox(1521,705){}}
}%
{\color[rgb]{0,0,0}\put(5070,-1866){\framebox(1149,712){}}
}%
{\color[rgb]{0,0,0}\put(6216,-1308){\vector( 1, 0){536}}
}%
{\color[rgb]{0,0,0}\put(6219,-1756){\vector( 1, 0){572}}
}%
{\color[rgb]{0,0,0}\put(2452,-872){\line( 1, 0){429}}
\put(2881,-872){\line( 0,-1){391}}
\put(2881,-1263){\vector( 1, 0){190}}
}%
{\color[rgb]{0,0,0}\put(2452,-2093){\line( 1, 0){429}}
\put(2881,-2093){\line( 0, 1){391}}
\put(2881,-1700){\vector( 1, 0){190}}
}%
{\color[rgb]{0,0,0}\put(4600,-1317){\vector( 1, 0){477}}
}%
{\color[rgb]{0,0,0}\put(4591,-1714){\vector( 1, 0){482}}
}%
{\color[rgb]{0,0,0}\put(718,-871){\vector( 1, 0){569}}
}%
{\color[rgb]{0,0,0}\put(729,-2097){\vector( 1, 0){569}}
}%
{\color[rgb]{0,0,0}\put(1299,-2450){\framebox(1149,712){}}
}%
{\color[rgb]{0,0,0}\put(1298,-1280){\framebox(1149,712){}}
}%
{\color[rgb]{0,0,0}\put(3831,-2857){\line(-1, 0){1939}}
\put(1892,-2857){\vector( 0, 1){401}}
}%
{\color[rgb]{0,0,0}\put(3813,-212){\line(-1, 0){1928}}
\put(1885,-212){\vector( 0,-1){355}}
}%
{\color[rgb]{0,0,0}\put(4000,-212){\line( 1, 0){1669}}
\put(5669,-212){\line( 0,-1){940}}
}%
{\color[rgb]{0,0,0}\put(4000,-2852){\line( 1, 0){1669}}
\put(5669,-2852){\line( 0, 1){980}}
}%
{\color[rgb]{0,0,0}\put(3910,190){\vector( 0,-1){317}}
}%
{\color[rgb]{0,0,0}\put(3910,-3255){\vector( 0, 1){317}}
}%
\put(3784,-1357){\makebox(0,0)[b]{\smash{{\SetFigFont{11}{13.2}{\rmdefault}{\mddefault}{\updefault}{\color[rgb]{0,0,0}Channel}%
}}}}
\put(3880,-1665){\makebox(0,0)[b]{\smash{{\SetFigFont{11}{13.2}{\rmdefault}{\mddefault}{\updefault}{\color[rgb]{0,0,0}$P(Y_t|X_t^a, X_t^b,S_t)$}%
}}}}
\put(5679,-1687){\makebox(0,0)[b]{\smash{{\SetFigFont{11}{13.2}{\rmdefault}{\mddefault}{\updefault}{\color[rgb]{0,0,0}$\psi(Y_{[n]},S_{[n]}^r)$}%
}}}}
\put(5623,-1440){\makebox(0,0)[b]{\smash{{\SetFigFont{11}{13.2}{\rmdefault}{\mddefault}{\updefault}{\color[rgb]{0,0,0}Decoder}%
}}}}
\put(6492,-1265){\makebox(0,0)[b]{\smash{{\SetFigFont{11}{13.2}{\rmdefault}{\mddefault}{\updefault}{\color[rgb]{0,0,0}$\hat{W}_a$}%
}}}}
\put(6497,-1704){\makebox(0,0)[b]{\smash{{\SetFigFont{11}{13.2}{\rmdefault}{\mddefault}{\updefault}{\color[rgb]{0,0,0}$\hat{W}_b$}%
}}}}
\put(2688,-823){\makebox(0,0)[b]{\smash{{\SetFigFont{11}{13.2}{\rmdefault}{\mddefault}{\updefault}{\color[rgb]{0,0,0}$X_t^a$}%
}}}}
\put(2688,-2318){\makebox(0,0)[b]{\smash{{\SetFigFont{11}{13.2}{\rmdefault}{\mddefault}{\updefault}{\color[rgb]{0,0,0}$X_t^b$}%
}}}}
\put(4798,-1291){\makebox(0,0)[b]{\smash{{\SetFigFont{11}{13.2}{\rmdefault}{\mddefault}{\updefault}{\color[rgb]{0,0,0}$Y_t$}%
}}}}
\put(4837,-1867){\makebox(0,0)[b]{\smash{{\SetFigFont{11}{13.2}{\rmdefault}{\mddefault}{\updefault}{\color[rgb]{0,0,0}$S_{t}^r$}%
}}}}
\put(2094,-470){\makebox(0,0)[b]{\smash{{\SetFigFont{11}{13.2}{\rmdefault}{\mddefault}{\updefault}{\color[rgb]{0,0,0}$S_{t}^a$}%
}}}}
\put(2094,-2718){\makebox(0,0)[b]{\smash{{\SetFigFont{11}{13.2}{\rmdefault}{\mddefault}{\updefault}{\color[rgb]{0,0,0}$S_{t}^b$}%
}}}}
\put(1017,-781){\makebox(0,0)[b]{\smash{{\SetFigFont{11}{13.2}{\rmdefault}{\mddefault}{\updefault}{\color[rgb]{0,0,0}$W_a$}%
}}}}
\put(1017,-1994){\makebox(0,0)[b]{\smash{{\SetFigFont{11}{13.2}{\rmdefault}{\mddefault}{\updefault}{\color[rgb]{0,0,0}$W_b$}%
}}}}
\put(1863,-860){\makebox(0,0)[b]{\smash{{\SetFigFont{11}{13.2}{\rmdefault}{\mddefault}{\updefault}{\color[rgb]{0,0,0}Encoder}%
}}}}
\put(1872,-1088){\makebox(0,0)[b]{\smash{{\SetFigFont{11}{13.2}{\rmdefault}{\mddefault}{\updefault}{\color[rgb]{0,0,0}$\phi_{t}^{(a)}(W_a,S_{[t]}^a)$}%
}}}}
\put(1854,-2047){\makebox(0,0)[b]{\smash{{\SetFigFont{11}{13.2}{\rmdefault}{\mddefault}{\updefault}{\color[rgb]{0,0,0}Encoder}%
}}}}
\put(1880,-2275){\makebox(0,0)[b]{\smash{{\SetFigFont{11}{13.2}{\rmdefault}{\mddefault}{\updefault}{\color[rgb]{0,0,0}$\phi_{t}^{(b)}(W_b,S_{[t]}^b)$}%
}}}}
\put(3800,-2905){\makebox(0,0)[lb]{\smash{{\SetFigFont{12}{14.4}{\rmdefault}{\mddefault}{\updefault}{\color[rgb]{0,0,0}$\bigoplus$}%
}}}}
\put(3800,-262){\makebox(0,0)[lb]{\smash{{\SetFigFont{12}{14.4}{\rmdefault}{\mddefault}{\updefault}{\color[rgb]{0,0,0}$\bigoplus$}%
}}}}
\put(4000,-3200){\makebox(0,0)[lb]{\smash{{\SetFigFont{12}{14.4}{\rmdefault}{\mddefault}{\updefault}{\color[rgb]{0,0,0}$N_b$}%
}}}}
\put(4000, 63){\makebox(0,0)[lb]{\smash{{\SetFigFont{12}{14.4}{\rmdefault}{\mddefault}{\updefault}{\color[rgb]{0,0,0}$N_a$}%
}}}}
\end{picture}
\caption{The multiple-access channel with causal noisy CSIT and noisy CSIR.}
\label{fig:mac-rec-noi}
\vspace{-0.2in}
\end{figure}
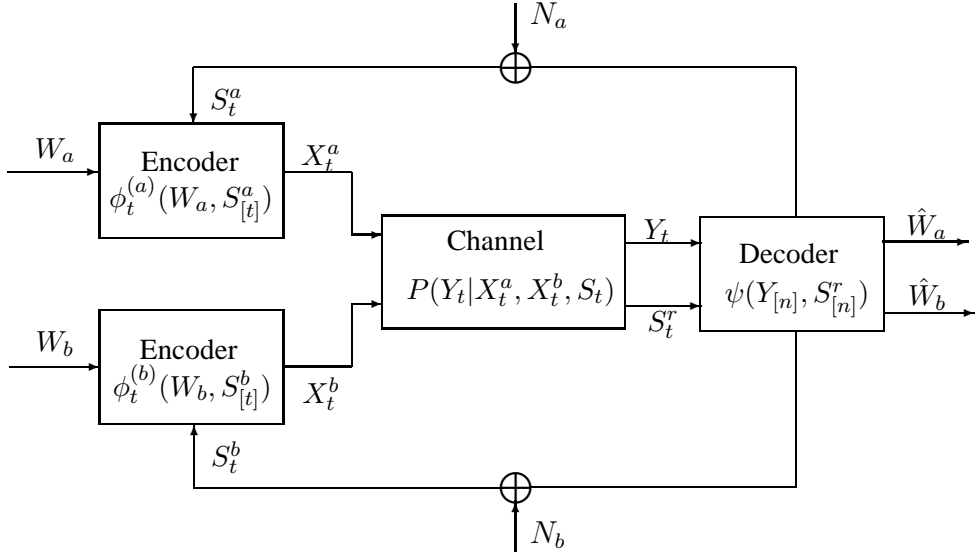

For every memoryless stationary team policy $\pi$ defined in (\ref{def:team-pol-shst}), let $\mathcal{R}_{NS}(\pi)$ denote the region of all rate pairs $R=(R_a,R_b)$ satisfying
\begin{eqnarray}
R_a &<& I(T^{a};Y|T^{b},S^r) \quad \label{eq:ra1-shst-rn}\\
R_b &<& I(T^{b};Y|T^{a},S^r) \quad  \label{eq:ra2-shst-rn}\\
R_a+R_b &<& I(T^a,T^{b};Y|S^r) \label{eq:ra3-shst-rn}
\end{eqnarray}
where $S^r$, $T^a$, $T^b$ and $Y$ are random variables taking values in ${\cal S}_r$, ${\cal T}_{a}$, ${\cal T}_{b}$ and $\cal Y$, respectively and whose joint probability distribution factorizes as
\begin{eqnarray}
P_{S^r,T^a,T^b,Y}(s^r,t^a,t^b,y)=P_{S^r}(s^r)P_{Y|T^a,T^b,S^r}(y|t^a,t^b,s^r)\pi_{T^a}(t^a)\pi_{T^b}(t^b)\label{eq:joi-dist-shst-rn}.
\end{eqnarray}
Let ${\cal \tilde{C}}_{IN}:=\overline{co}\bigg(\bigcup_{\pi}\mathcal{R}_{NS}(\pi)\bigg)$ denotes the closure of the convex hull of the rate regions $\mathcal{R}_{NS}(\pi)$ given by (\ref{eq:ra1-shst-rn})-(\ref{eq:ra3-shst-rn}) associated to all possible memoryless team policies as defined in (\ref{eq:team-pol-shst}).

\begin{remark}\label{rem:eqiuv-cha}
It should be observed that once we have the Markov property (\ref{eq:noi-recnoi-mrk}), the setup with noisy CSIR described above is no more general then the setup with complete CSIR. This is because, one can define an equivalent channel with conditional output probability
\begin{eqnarray}
P^{eq}_{Y|X^a, X^b, S^r}(y|x^a,x^b,s^r)=\sum_{s \in {\cal S}}P_{Y|X^a,X^b,S}(y|x^a,x^b,s)P_{S|S^r}(s|s^r)\label{eq:eqiv-chn}
\end{eqnarray}
which follows from (\ref{eq:star-iid}). With (\ref{eq:eqiv-chn}), the noisy CSIR problem  reduces to the complete CSIR problem since we can now define a new channel with state $S^r$ and
\begin{eqnarray}
P_{S_{t}^{a}, S_{t}^{b},S_t^r}(s_{t}^{a}, s_{t}^{b},s_t^r)=P_{S_{t}^{a}|S_t^r}(s_{t}^{a}|s_t^r)P_{S_{t}^{b}|S_t^r}(s_{t}^{b}|s_t^r)P_{S_t^r}(s_t^r)\nonumber.
\end{eqnarray}
Hence, the proofs of the corollaries below follow directly from the complete CSIR case.
\end{remark}

\begin{corollary}[Inner Bound to ${\cal \tilde{C}}_{NS}$]\label{cor:inbnd-cau-rn}
 ${\cal \tilde{C}}_{IN}\subseteq {\cal \tilde{C}}_{NS}$.
\end{corollary}
\begin{corollary}[Outer Bound to ${\cal \tilde{C}}_{NS}$]\label{cor:outbnd-cau-rn}
${\cal \tilde{C}}_{NS}\subseteq {\cal \tilde{C}}_{OUT}$, where
\begin{eqnarray}
{\cal \tilde{C}}_{OUT}:=\biggl\{(R_a,R_b)\in \mathds{R}^{+}\times\mathds{R}^{+}: R_a+R_b\leq\sup_{\pi_{T^a}(t^a)\pi_{T^b}(t^b)}I(T^a,T^b;Y|S^r)\biggr\}\nonumber.
\end{eqnarray}
\end{corollary}
\begin{corollary}\label{cor:sumra-cau-rn}
\begin{eqnarray}
{\cal \tilde{C}}_{NS}^{\sum}=\sup_{\pi_{T^a}(t^a)\pi_{T^b}(t^b)}I(T^a,T^b;Y|S^r)\label{eq:sumra-cau-rn}.
\end{eqnarray}
\end{corollary}
This corollary indicates the fact that even if we have noisy CSIR, Shannon strategies are still optimal for the sum-rate capacity as long as conditioning on the past CSIR gives a product form on these strategies.

\subsection{CSITs are Deterministic Functions of CSIR:Causal and Non-Causal Cases}\label{subsec:csit-fnc-rn}
Since the computation of optimal strategies in Corollaries \ref{cor:sumra-cau-fl} and \ref{cor:sumra-cau-rn} requires an optimization over extended input alphabets, it is worth to consider the case in which optimization can be performed over the input alphabets ${\cal X}_a$ and ${\cal X}_b$. The usual approach is to assume that the transmitters have access to partial (through a deterministic function such as a quantizer) state information at the decoder. In particular, let $S_t^i=f^i(S_t^r)$, where $f^i:{\cal S}_r\rightarrow {\cal S}_{i}$, $i=\{a,b\}$.

The equivalent channel defined in (\ref{eq:eqiv-chn}) shows that the causal setup of this problem is no more general than \cite{como-yuksel}. Hence, the main contribution of this subsection is to provide a single letter characterization for the capacity region for the non-causal case. The expression shows that the result of \cite{como-yuksel} also holds for non-causal coding.

We keep the channel codes definition identical for the causal and non-causal cases, except for the non-causal case we have; $\phi_{t}^{(i)}: \mathcal{S}_{i}^n \times \mathcal{W}_i  \rightarrow {\cal X}_i^n$, $i=\{a,b\}$, $t=1,\cdots,n$.

Let ${\cal C}^{Q}_{NS}$ and ${\cal C}^{Q'}_{NS}$ denote the capacity region for the causal and non-causal cases, respectively. We need to modify Definition \ref{def:team-pol-shst} in order to take the current CSI into account.
\begin{definition}\label{def:team-pol-inp-ncas}
A memoryless stationary (in time) team policy is a family
\begin{eqnarray}
\bar{\Pi}=\left\{\bar{\pi}=\left(\pi_{X^a|S^a}(\cdot|f^a(s^r)),\pi_{X^b|S^b}(\cdot|f^b(s^r))\right)\in {\cal P}({\cal X}_a)\times{\cal P}({\cal X}_b)\right\}.\label{eq:team-pol-inp-ncas}
\end{eqnarray}
\end{definition}
For every $\bar{\pi}$ defined in (\ref{eq:team-pol-inp-ncas}), $\mathcal{R}_{NS}^Q(\bar{\pi})$ denotes the region of all rate pairs $R=(R_a,R_b)$ satisfying{\allowdisplaybreaks
\begin{eqnarray}
R_a &<& I(X^{a};Y|X^{b},S^r) \quad \label{eq:ra1-inp-ncas}\\
R_b &<& I(X^{b};Y|X^{a},S^r) \quad  \label{eq:ra2-inp-ncas}\\
R_a+R_b &<& I(X^a,X^{b};Y|S^r) \label{eq:ra3-inp-ncas}
\end{eqnarray}
where $S^r$, $X^a$, $X^b$ and $Y$ are random variables taking values in ${\cal S}_r$, ${\cal X}_{a}$, ${\cal X}_{b}$ and $\cal Y$, respectively, and whose joint probability distribution factorizes as
\begin{eqnarray}
&&\hspace{-0.5in}P_{S^r,X^a,X^b,Y}(s^r,x^a,x^b,y)\nonumber\\
&=&P_{S^r}(s^r)P_{Y|X^a,X^b,S^r}(y|x^a,x^b,s^r)\pi_{X^a|S^a}(x^a|f^a(s^r))\pi_{X^b|S^b}(x^b|f^b(s^r)).\label{eq:joi-dist-inp-ncas}
\end{eqnarray}}
Let $\overline{co}\bigg(\bigcup_{\bar{\pi}}\mathcal{R}_{NS}^Q(\bar{\pi})\bigg)$ denote the closure of the convex hull of the rate regions $\mathcal{R}_{NS}^Q(\bar{\pi})$ given by (\ref{eq:ra1-inp-ncas})-(\ref{eq:ra3-inp-ncas}) associated to all possible memoryless stationary team polices as defined in (\ref{eq:team-pol-inp-ncas}).
\begin{theorem}\label{theo:main-inp-ncas}
${\cal C}^{Q}_{NS}={\cal C}^{Q'}_{NS}=\overline{co}\bigg(\bigcup_{\bar{\pi}}\mathcal{R}_{NS}^Q(\bar{\pi})\bigg)$.
\end{theorem}
For the achievability proof, see \cite[Section III]{como-yuksel} and observe that any rate which is achievable with causal CSI is also achievable with non-causal CSI. For the converse proof of the non-causal case see Appendix \ref{app:0}. The proof for the non-causal case is realized by observing that there is no loss of optimality if not only the past, as shown in \cite{como-yuksel}, but also the future CSI is ignored given that the receiver is provided with complete CSI.

It should be noted that the causal result can be thought of as an extension of \cite[Propositon 1]{caire-shamai} to a multi-user case and the non-causal case is also considered in \cite[Theorem 3]{Cemal} where inner and outer bounds are provided.

\begin{remark}\label{rem:ncas-1}
Following \cite[Remark 1]{como-yuksel}, it is worth to emphasize that for the above argument to work, it is crucial that the past and future state realizations only affect the team policies and that the state information available at the decoder contains the one available at the two transmitters.
\end{remark}
In particular, the latter fact plays a role in the converse part of the coding theorem by enabling the decoder to ignore the past channel outputs, given that the channel is memoryless, without any loss of optimality.

Let us investigate this remark via considering the setup in Section \ref{sec:causal} in order to observe that for the non-causal case the optimality of Shannon strategies are not guaranteed. Recall that, we have
\begin{eqnarray}
I(\mathbf{W};Y_{[n]},S_{[n]}){\leq}\sum_{t=1}^{n}\left[H(Y_t|S_{[n]},Y_{[t-1]})-H(Y_t|\mathbf{W},S_{[n]},Y_{[t-1]},\mathbf{T}_t)\right]\label{eq:cas-sumrate}
\end{eqnarray}
where $\mathbf{T}_t:=(T_t^a,T_t^b)$. Consider now the right hand side of (\ref{eq:cas-sumrate}) and observe that
\begin{eqnarray}
&&\hspace{-0.5in}P_{Y_t|\mathbf{W},S_{[n]},Y_{[t-1]},T_t^a,T_t^b}(y_t|\mathbf{w},s_{[n]},y_{[t-1]},t_t^a,t_t^b)\nonumber\\
&&=\sum_{s_t^a,s_t^b}P_{Y_t|S_t,S_t^a,S_t^b,T_t^a,T_t^b}(y_t|s_t,s_t^a,s_t^b,t_t^a,t_t^b) P_{S_t^a,S_t^b|Y_{[t-1]},S_t}(s_t^a,s_t^b|y_{[t-1]},s_t),\nonumber
\end{eqnarray}
and therefore, the past channel outputs cannot be ignored. Recall that in the causal setup the conditional probability $P_{S_t^a,S_t^b|Y_{[t-1]},S_t}(s_t^a,s_t^b|y_{[t-1]},s_t)$ is independent of past channel outputs.
\section{Asymmetric Delayed, Asymmetric Noisy CSIT and Complete CSIR}\label{sec:delay}
Consider the problem defined in Section \ref{sec:causal} where the two encoders have accesses to asymmetrically delayed, where delays are $d_a\geq 1$ and $d_b\geq1$, respectively, and noisy versions of the state information $S_t$ at each time $t\geq1$, modeled by $S_{t-d_a}^{a} \in {\cal S}_{a}$, $S_{t-d_b}^{b} \in {\cal S}_{b}$, respectively. The rest of the channel model is identical and hence, (\ref{eq:noi-recfl-mrk}), (\ref{eq:sta-iid}) and (\ref{eq:ch-recfl}) are valid throughout the section. We also assume that $S_t$ is fully available at the receiver. A code can be defined as in Definition \ref{def:maccode-causal}, except now
\begin{center}
$\phi_{t}^{(a)}: \mathcal{S}_{a}^{t-d_a} \times \mathcal{W}_a  \rightarrow {\cal X}_a, \enspace t=1,2,...n$;\\
\end{center}
\begin{center}
$\phi_{t}^{(b)}: \mathcal{S}_{b}^{t-d_b} \times \mathcal{W}_b \rightarrow {\cal X}_b, \enspace t=1,2,...n$.\footnote{Obviously, when $d_l\geq t, \enspace l=a,b$ then $X_t^a=\phi_t^{(a)}(W_a)$ and $X_t^b=\phi_t^{(b)}(W_b)$.}\\
\end{center}
Let ${\cal C}_{DN}$ denote the capacity region of the delayed setup.

In the main result of this section the team policies are composed of probability distributions on the channel inputs rather than Shannon strategies.
\begin{definition}\label{def:team-pol-inp}
A memoryless stationary (in time) team policy is a family
\begin{eqnarray}
\tilde{\Pi}=\left\{\tilde{\pi}=\left(\pi_{X^a}(\cdot),\pi_{X^b}(\cdot)\right)\in {\cal P}({\cal X}^a)\times{\cal P}({\cal X}^b)\right\}\label{eq:team-pol-inp}.
\end{eqnarray}
\end{definition}

For every memoryless stationary team policy $\tilde{\pi}$, $\mathcal{R}_{DN}(\tilde{\pi})$ denotes the region of all rate pairs $R=(R_a,R_b)$ satisfying
\begin{eqnarray}
R_a &<& I(X^{a};Y|X^{b},S) \quad \label{eq:ra1-inp}\\
R_b &<& I(X^{b};Y|X^{a},S) \quad  \label{eq:ra2-inp}\\
R_a+R_b &<& I(X^a,X^{b};Y|S) \label{eq:ra3-inp}
\end{eqnarray}
where $S$, $X^a$, $X^b$ and $Y$ are random variables taking values in ${\cal S}$, ${\cal X}^{a}$, ${\cal X}^{b}$ and $\cal Y$, respectively and whose joint probability distribution factorizes as
\begin{eqnarray}
P_{S,X^a,X^b,Y}(s,x^a,x^b,y)=P_S(s)P_{Y|X^a,X^b,S}(y|x^a,x^b,s)\pi_{X^a}(x^a)\pi_{X^b}(x^b)\label{eq:joi-dist-inp}.
\end{eqnarray}
Let $\overline{co}\bigg(\bigcup_{\tilde{\pi}}\mathcal{R}_{DN}(\tilde{\pi})\bigg)$ denotes the closure of the convex hull of the rate regions $\mathcal{R}_{DN}(\tilde{\pi})$ given by (\ref{eq:ra1-inp})-(\ref{eq:ra3-inp}) associated to all possible memoryless stationary team polices as defined in (\ref{eq:team-pol-inp}).
\begin{theorem}\label{theo:main-inp}
${\cal C}_{DN}=\overline{co}\bigg(\bigcup_{\tilde{\pi}}\mathcal{R}_{DN}(\tilde{\pi})\bigg)$.
\end{theorem}
Achievability can be shown by random coding arguments. For the converse, see Appendix \ref{app:1}.
\begin{remark}[Strictly Causal Case]
When $d_a=d_b=1$, Theorem \ref{theo:main-inp} is the capacity region of the setup with strictly causal CSITs. In \cite{lap-ste-1} and \cite{lap-ste-3}, achievable rate regions are provided for the case when the channel is driven by two independent states (with no CSIT). When the encoders have strictly causal CSI (not noisy/not asymmetric), the authors proposed a region which is based on sending a compressed version of the state information available at the encoders to the decoder. Theorem \ref{theo:main-inp} verifies that since the full CSI is available at the receiver and since the decoder does not need to access the current CSI at the encoders, there exists no loss of optimality if the past information at the encoders are ignored.
\end{remark}
\section{Cooperative FS-MAC with Noisy CSIT}\label{sec:corr}
We now consider the last scenario of the paper. Assume a common message is provided to both encoders and one of the encoders has its own private message. Assume further that the encoder with the private message causally observes noisy state information, whereas the encoder with the common message only observes noisy state information with delay $d_a\geq1$. Let the common and the private messages be $W_a$ and $W_b$, respectively, and $S_{[t-d_a]}^a$, $d_a\geq1$, and $S_{[t]}^b$ denote the CSI at encoder $a$, $b$, respectively, where $(S_t,S_t^a,S_t^b)$ satisfies (\ref{eq:noi-recfl-mrk}) and (\ref{eq:sta-iid}). Hence, $X_t^a=\phi_{t}^{(a)}(W_a,S_{[t-d_a]}^a)$ and $X_t^b=\phi_{t}^{(b)}(W_a, W_b,S_{[t]}^b)$; see Fig. \ref{fig:mac-noi-corr}. Let ${\cal C}_{C}$ denote the capacity region for this channel. Recall that ${\cal T}_b=\mathcal X_b^{\mathcal S_b}$.
%
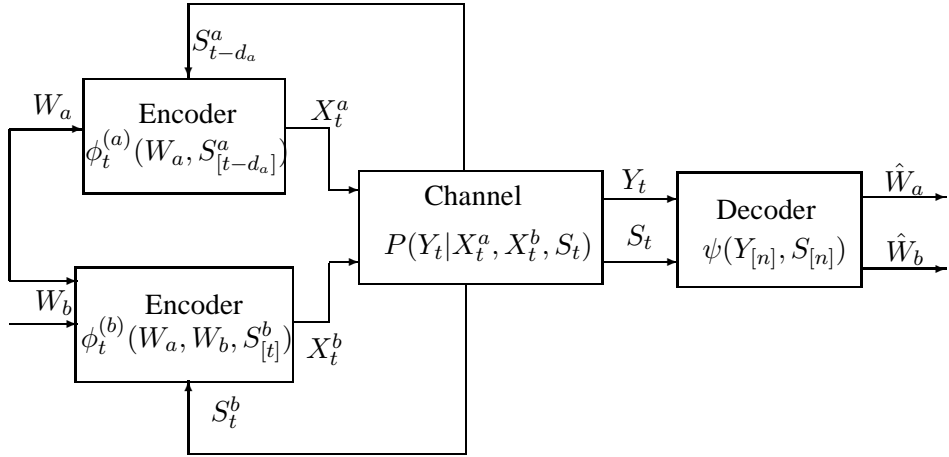
\begin{figure}
\setlength{\unitlength}{3947sp}%
\begingroup\makeatletter\ifx\SetFigFont\undefined%
\gdef\SetFigFont#1#2#3#4#5{%
  \reset@font\fontsize{#1}{#2pt}%
  \fontfamily{#3}\fontseries{#4}\fontshape{#5}%
  \selectfont}%
\fi\endgroup%
\begin{picture}(5945,2853)(0,-3377)
\thinlines
{\color[rgb]{0,0,0}\put(1083,-2549){\vector( 1, 0){425}}
}%
{\color[rgb]{0,0,0}\put(1500,-2909){\framebox(1358,712){}}
}%
{\color[rgb]{0,0,0}\put(1550,-1725){\framebox(1258,712){}}
}%
{\color[rgb]{0,0,0}\put(1083,-1327){\vector( 1, 0){461}}
}%
{\color[rgb]{0,0,0}\put(1083,-1327){\line( 0,-1){951}}
\put(1083,-2278){\vector( 1, 0){426}}
}%
{\color[rgb]{0,0,0}\put(3283,-2296){\framebox(1521,705){}}
}%
{\color[rgb]{0,0,0}\put(5281,-2311){\framebox(1149,712){}}
}%
{\color[rgb]{0,0,0}\put(6430,-1753){\vector( 1, 0){536}}
}%
{\color[rgb]{0,0,0}\put(6430,-2201){\vector( 1, 0){536}}
}%
{\color[rgb]{0,0,0}\put(2810,-1317){\line( 1, 0){283}}
\put(3092,-1317){\line( 0,-1){391}}
\put(3092,-1708){\vector( 1, 0){190}}
}%
{\color[rgb]{0,0,0}\put(2860,-2538){\line( 1, 0){233}}
\put(3092,-2538){\line( 0, 1){371}}
\put(3092,-2160){\vector( 1, 0){190}}
}%
{\color[rgb]{0,0,0}\put(2207,-536){\vector( 0,-1){468}}
\put(2207,-536){\line( 1, 0){1735}}
\put(3940,-536){\line( 0,-1){1053}}
}%
{\color[rgb]{0,0,0}\put(2208,-3365){\vector( 0, 1){461}}
\put(2208,-3365){\line( 1, 0){1746}}
\put(3951,-3365){\line( 0, 1){1064}}
}%
{\color[rgb]{0,0,0}\put(4811,-1762){\vector( 1, 0){470}}
}%
{\color[rgb]{0,0,0}\put(4811,-2159){\vector( 1, 0){470}}
}%
\put(1356,-1230){\makebox(0,0)[b]{\smash{{\SetFigFont{11}{13.2}{\rmdefault}{\mddefault}{\updefault}{\color[rgb]{0,0,0}$W_a$}%
}}}}
\put(2200,-1281){\makebox(0,0)[b]{\smash{{\SetFigFont{11}{13.2}{\rmdefault}{\mddefault}{\updefault}{\color[rgb]{0,0,0}Encoder}%
}}}}
\put(2200,-1520){\makebox(0,0)[b]{\smash{{\SetFigFont{11}{13.2}{\rmdefault}{\mddefault}{\updefault}{\color[rgb]{0,0,0}$\phi_{t}^{(a)}(W_a, S_{[t-d_a]}^a)$}%
}}}}
\put(2200,-2691){\makebox(0,0)[b]{\smash{{\SetFigFont{11}{13.2}{\rmdefault}{\mddefault}{\updefault}{\color[rgb]{0,0,0}$\phi_{t}^{(b)}(W_a,W_b,S_{[t]}^b)$}%
}}}}
\put(2240,-2484){\makebox(0,0)[b]{\smash{{\SetFigFont{11}{13.2}{\rmdefault}{\mddefault}{\updefault}{\color[rgb]{0,0,0}Encoder}%
}}}}
\put(3995,-1802){\makebox(0,0)[b]{\smash{{\SetFigFont{11}{13.2}{\rmdefault}{\mddefault}{\updefault}{\color[rgb]{0,0,0}Channel}%
}}}}
\put(4091,-2110){\makebox(0,0)[b]{\smash{{\SetFigFont{11}{13.2}{\rmdefault}{\mddefault}{\updefault}{\color[rgb]{0,0,0}$P(Y_t|X_t^a, X_t^b,S_t)$}%
}}}}
\put(5890,-2132){\makebox(0,0)[b]{\smash{{\SetFigFont{11}{13.2}{\rmdefault}{\mddefault}{\updefault}{\color[rgb]{0,0,0}$\psi(Y_{[n]},S_{[n]})$}%
}}}}
\put(5834,-1885){\makebox(0,0)[b]{\smash{{\SetFigFont{11}{13.2}{\rmdefault}{\mddefault}{\updefault}{\color[rgb]{0,0,0}Decoder}%
}}}}
\put(6703,-1710){\makebox(0,0)[b]{\smash{{\SetFigFont{11}{13.2}{\rmdefault}{\mddefault}{\updefault}{\color[rgb]{0,0,0}$\hat{W}_a$}%
}}}}
\put(6708,-2149){\makebox(0,0)[b]{\smash{{\SetFigFont{11}{13.2}{\rmdefault}{\mddefault}{\updefault}{\color[rgb]{0,0,0}$\hat{W}_b$}%
}}}}
\put(1363,-2460){\makebox(0,0)[b]{\smash{{\SetFigFont{11}{13.2}{\rmdefault}{\mddefault}{\updefault}{\color[rgb]{0,0,0}$W_b$}%
}}}}
\put(3099,-1268){\makebox(0,0)[b]{\smash{{\SetFigFont{11}{13.2}{\rmdefault}{\mddefault}{\updefault}{\color[rgb]{0,0,0}$X_t^a$}%
}}}}
\put(3063,-2763){\makebox(0,0)[b]{\smash{{\SetFigFont{11}{13.2}{\rmdefault}{\mddefault}{\updefault}{\color[rgb]{0,0,0}$X_t^b$}%
}}}}
\put(2438,-827){\makebox(0,0)[b]{\smash{{\SetFigFont{11}{13.2}{\rmdefault}{\mddefault}{\updefault}{\color[rgb]{0,0,0}$S_{t-d_a}^a$}%
}}}}
\put(2444,-3159){\makebox(0,0)[b]{\smash{{\SetFigFont{11}{13.2}{\rmdefault}{\mddefault}{\updefault}{\color[rgb]{0,0,0}$S_{t}^b$}%
}}}}
\put(5009,-1700){\makebox(0,0)[b]{\smash{{\SetFigFont{11}{13.2}{\rmdefault}{\mddefault}{\updefault}{\color[rgb]{0,0,0}$Y_t$}%
}}}}
\put(5042,-2050){\makebox(0,0)[b]{\smash{{\SetFigFont{11}{13.2}{\rmdefault}{\mddefault}{\updefault}{\color[rgb]{0,0,0}$S_{t}$}%
}}}}
\end{picture}
\caption{Cooperative multiple-access channel with noisy state feedback.}
\label{fig:mac-noi-corr}
\vspace{-0.2in}
\end{figure}
\begin{definition}\label{def:team-pol-corr}
A memoryless stationary (in time) team policy is a family
\begin{eqnarray}
\hat{\Pi}=\left\{\hat{\pi}=\left(\pi_{X^a,T^b}(\cdot,\cdot)\right)\in {\cal P}({\cal X}^a)\times{\cal P}({\cal T}^b)\right\}\label{eq:team-pol-corr}
\end{eqnarray}
of probability distributions on $({\cal X}_a, {\cal T}_b)$.
\end{definition}

Let for every $\hat{\pi}$, $\mathcal{R}_{C}(\hat{\pi})$ denote the region of all rate pairs $R=(R_a,R_b)$ satisfying
\begin{eqnarray}
R_b &<& I(T^{b};Y|X^{a},S) \quad  \label{eq:ra2-corr}\\
R_a+R_b &<& I(X^a,T^{b};Y|S) \label{eq:ra3-corr}
\end{eqnarray}
where $S$, $X^a$, $T^b$ and $Y$ are random variables taking values in ${\cal S}$, ${\cal X}_{a}$, ${\cal T}_{b}$ and $\cal Y$, respectively and whose joint probability distribution factorizes as
\begin{eqnarray}
P_{S,X^a,T^b,Y}(s,x^a,t^b,y)=P_S(s)P_{Y|X^a,T^b,S}(y|x^a,t^b,s)\pi_{X^a, T^b}(x^a,t^b)\label{eq:joi-dist-corr}.
\end{eqnarray}
Let $\overline{co}\bigg(\bigcup_{\hat{\pi}}\mathcal{R}_{C}(\hat{\pi})\bigg)$ denotes the closure of the convex hull of the rate regions $\mathcal{R}_{C}(\hat{\pi})$ given by (\ref{eq:ra2-corr}) and (\ref{eq:ra3-corr}) associated to all possible memoryless stationary team polices as defined in (\ref{eq:team-pol-corr}).
\begin{theorem}\label{theo:main-corr}
${\cal C}_{C}=\overline{co}\bigg(\bigcup_{\hat{\pi}}\mathcal{R}_{C}(\hat{\pi})\bigg)$.
\end{theorem}
See Appendix \ref{app:2} for the proof.
\begin{remark}\label{rem:coop-causal-fl}
Theorem \ref{theo:main-corr} shows that when the common message encoder has no access to the current noisy CSI (since the delay $d_a \geq 1$), by enlarging the optimization space of the other encoder, via Shannon strategies, the past CSI can be ignored without loss of optimality if the decoder is provided with complete CSI.
\end{remark}
One important observation to be made in the cooperative scenario is that we do not require a product form on the pair $(X^a,T^b)$ (see (\ref{eq:joi-dist-corr})). In connection with this observation, let us consider the following noisy CSIR setup.

Let the encoder with the private message causally observe noisy state information, whereas the encoder with the common message has no CSI, i.e., $X_t^a=\phi_{t}^{(a)}(W_a)$ and $X_t^b=\phi_{t}^{(b)}(W_a, W_b,S_{[t]}^b)$, and the decoder also has access to noisy CSI at time $t$, $S_t^r \in {\cal S}_r$; see Fig. \ref{fig:mac-noi-corr-gen}. Let ${\cal C}^G_{C}$ denote the capacity region for this setup.
%
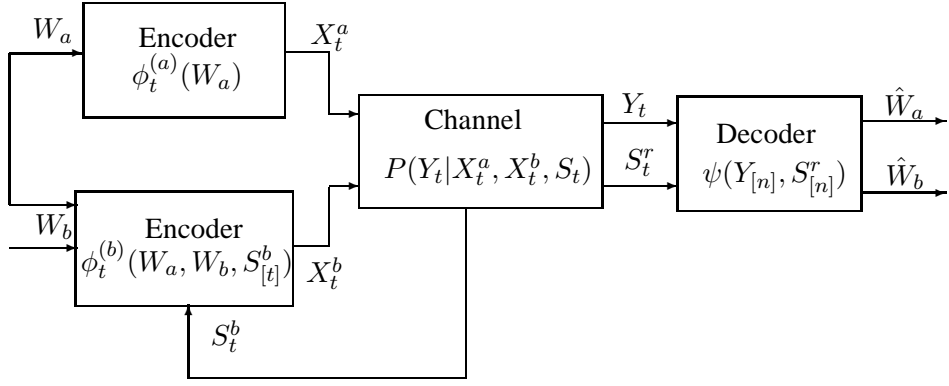
\begin{figure}
\setlength{\unitlength}{3947sp}%
\begingroup\makeatletter\ifx\SetFigFont\undefined%
\gdef\SetFigFont#1#2#3#4#5{%
  \reset@font\fontsize{#1}{#2pt}%
  \fontfamily{#3}\fontseries{#4}\fontshape{#5}%
  \selectfont}%
\fi\endgroup%
\begin{picture}(5945,2853)(0,-3377)
\thinlines
{\color[rgb]{0,0,0}\put(1083,-2549){\vector( 1, 0){425}}
}%
{\color[rgb]{0,0,0}\put(1500,-2909){\framebox(1358,712){}}
}%
{\color[rgb]{0,0,0}\put(1550,-1725){\framebox(1258,712){}}
}%
{\color[rgb]{0,0,0}\put(1083,-1327){\vector( 1, 0){461}}
}%
{\color[rgb]{0,0,0}\put(1083,-1327){\line( 0,-1){951}}
\put(1083,-2278){\vector( 1, 0){426}}
}%
{\color[rgb]{0,0,0}\put(3283,-2296){\framebox(1521,705){}}
}%
{\color[rgb]{0,0,0}\put(5281,-2311){\framebox(1149,712){}}
}%
{\color[rgb]{0,0,0}\put(6430,-1753){\vector( 1, 0){536}}
}%
{\color[rgb]{0,0,0}\put(6430,-2201){\vector( 1, 0){536}}
}%
{\color[rgb]{0,0,0}\put(2810,-1317){\line( 1, 0){283}}
\put(3092,-1317){\line( 0,-1){391}}
\put(3092,-1708){\vector( 1, 0){190}}
}%
{\color[rgb]{0,0,0}\put(2860,-2538){\line( 1, 0){233}}
\put(3092,-2538){\line( 0, 1){371}}
\put(3092,-2160){\vector( 1, 0){190}}
}%
{\color[rgb]{0,0,0}\put(2208,-3365){\vector( 0, 1){461}}
\put(2208,-3365){\line( 1, 0){1746}}
\put(3951,-3365){\line( 0, 1){1064}}
}%
{\color[rgb]{0,0,0}\put(4811,-1762){\vector( 1, 0){470}}
}%
{\color[rgb]{0,0,0}\put(4811,-2159){\vector( 1, 0){470}}
}%
\put(1356,-1230){\makebox(0,0)[b]{\smash{{\SetFigFont{11}{13.2}{\rmdefault}{\mddefault}{\updefault}{\color[rgb]{0,0,0}$W_a$}%
}}}}
\put(2200,-1281){\makebox(0,0)[b]{\smash{{\SetFigFont{11}{13.2}{\rmdefault}{\mddefault}{\updefault}{\color[rgb]{0,0,0}Encoder}%
}}}}
\put(2200,-1520){\makebox(0,0)[b]{\smash{{\SetFigFont{11}{13.2}{\rmdefault}{\mddefault}{\updefault}{\color[rgb]{0,0,0}$\phi_{t}^{(a)}(W_a)$}%
}}}}
\put(2200,-2691){\makebox(0,0)[b]{\smash{{\SetFigFont{11}{13.2}{\rmdefault}{\mddefault}{\updefault}{\color[rgb]{0,0,0}$\phi_{t}^{(b)}(W_a,W_b,S_{[t]}^b)$}%
}}}}
\put(2240,-2484){\makebox(0,0)[b]{\smash{{\SetFigFont{11}{13.2}{\rmdefault}{\mddefault}{\updefault}{\color[rgb]{0,0,0}Encoder}%
}}}}
\put(3995,-1802){\makebox(0,0)[b]{\smash{{\SetFigFont{11}{13.2}{\rmdefault}{\mddefault}{\updefault}{\color[rgb]{0,0,0}Channel}%
}}}}
\put(4091,-2110){\makebox(0,0)[b]{\smash{{\SetFigFont{11}{13.2}{\rmdefault}{\mddefault}{\updefault}{\color[rgb]{0,0,0}$P(Y_t|X_t^a, X_t^b,S_t)$}%
}}}}
\put(5890,-2132){\makebox(0,0)[b]{\smash{{\SetFigFont{11}{13.2}{\rmdefault}{\mddefault}{\updefault}{\color[rgb]{0,0,0}$\psi(Y_{[n]},S_{[n]}^r)$}%
}}}}
\put(5834,-1885){\makebox(0,0)[b]{\smash{{\SetFigFont{11}{13.2}{\rmdefault}{\mddefault}{\updefault}{\color[rgb]{0,0,0}Decoder}%
}}}}
\put(6703,-1710){\makebox(0,0)[b]{\smash{{\SetFigFont{11}{13.2}{\rmdefault}{\mddefault}{\updefault}{\color[rgb]{0,0,0}$\hat{W}_a$}%
}}}}
\put(6708,-2149){\makebox(0,0)[b]{\smash{{\SetFigFont{11}{13.2}{\rmdefault}{\mddefault}{\updefault}{\color[rgb]{0,0,0}$\hat{W}_b$}%
}}}}
\put(1363,-2460){\makebox(0,0)[b]{\smash{{\SetFigFont{11}{13.2}{\rmdefault}{\mddefault}{\updefault}{\color[rgb]{0,0,0}$W_b$}%
}}}}
\put(3099,-1268){\makebox(0,0)[b]{\smash{{\SetFigFont{11}{13.2}{\rmdefault}{\mddefault}{\updefault}{\color[rgb]{0,0,0}$X_t^a$}%
}}}}
\put(3063,-2763){\makebox(0,0)[b]{\smash{{\SetFigFont{11}{13.2}{\rmdefault}{\mddefault}{\updefault}{\color[rgb]{0,0,0}$X_t^b$}%
}}}}
\put(2444,-3159){\makebox(0,0)[b]{\smash{{\SetFigFont{11}{13.2}{\rmdefault}{\mddefault}{\updefault}{\color[rgb]{0,0,0}$S_{t}^b$}%
}}}}
\put(5009,-1700){\makebox(0,0)[b]{\smash{{\SetFigFont{11}{13.2}{\rmdefault}{\mddefault}{\updefault}{\color[rgb]{0,0,0}$Y_t$}%
}}}}
\put(5042,-2050){\makebox(0,0)[b]{\smash{{\SetFigFont{11}{13.2}{\rmdefault}{\mddefault}{\updefault}{\color[rgb]{0,0,0}$S_{t}^r$}%
}}}}
\end{picture}
\caption{Cooperative multiple-access channel with noisy CSIT and CSIR.}
\label{fig:mac-noi-corr-gen}
\vspace{-0.2in}
\end{figure}
Let for every memoryless stationary team policy $\hat{\pi}$ defined in (\ref{eq:team-pol-corr}), $\mathcal{R}^G_{C}(\hat{\pi})$ denote the region of all rate pairs $R=(R_a,R_b)$ satisfying,
\begin{eqnarray}
R_b &<& I(T^{b};Y|X^{a},S^r) \quad  \label{eq:ra2-corr-gen}\\
R_a+R_b &<& I(X^a,T^{b};Y|S^r) \label{eq:ra3-corr-gen}
\end{eqnarray}
where $S^r$, $X^a$, $T^b$ and $Y$ are random variables taking values in ${\cal S}_r$, ${\cal X}_{a}$, ${\cal T}_{b}$ and $\cal Y$, respectively and whose joint probability distribution factorizes as
\begin{eqnarray}
P_{S^r,X^a,T^b,Y}(s^r,x^a,t^b,y)=P_{S^r}(s^r)P_{Y|X^a,T^b,S^r}(y|x^a,t^b,s^r)\pi_{X^a, T^b}(x^a,t^b)\label{eq:joi-dist-corr-gen}.
\end{eqnarray}
Let $\overline{co}\bigg(\bigcup_{\hat{\pi}}\mathcal{R}^G_{C}(\hat{\pi})\bigg)$ denotes the closure of the convex hull of the rate regions $\mathcal{R}^G_{C}(\hat{\pi})$ given by (\ref{eq:ra2-corr-gen}) and (\ref{eq:ra3-corr-gen}) associated to all possible $\hat{\pi}$ as defined in (\ref{eq:team-pol-corr}).
\begin{theorem}\label{theo:main-corr-gen}
${\cal C}^G_{C}=\overline{co}\bigg(\bigcup_{\hat{\pi}}\mathcal{R}^G_{C}(\hat{\pi})\bigg)$.
\end{theorem}
\begin{proof}
The achievability proof is identical to that of Theorem \ref{theo:main-corr}. Converse proof is also similar and therefore, we only provide a sketch. In particular, observe the following lines of equations for the converse proof of the condition on $R_b$:
{\allowdisplaybreaks
\begin{eqnarray}
I(W_b;Y_{[n]}, S_{[n]}^r)&\leq&I(W_b;Y_{[n]}, S_{[n]}^r|W_a)\nonumber\\
&=&\sum_{t=1}^n\left[H(Y_t,S_t^r|S_{[t-1]}^r,Y_{[t-1]},W_a)-H(Y_t,S_t^r|S_{[t-1]}^r,Y_{[t-1]}, W_a, W_b)\right]\nonumber\\
&\overset{(i)}{=}&\sum_{t=1}^n\left[H(Y_t|S_{[t]}^r,Y_{[t-1]},W_a)-H(Y_t|S_{[t]}^r,Y_{[t-1]}, W_a, W_b)\right]\nonumber\\
&{=}&\sum_{t=1}^n\left[H(Y_t|S_{[t]}^r,Y_{[t-1]},W_a,X_{t}^a)-H(Y_t|S_{[t]}^r,Y_{[t-1]}, W_a, W_b, X_{t}^a)\right]\nonumber\\
&\overset{(ii)}{\leq}&\sum_{t=1}^n\left[H(Y_t|S_{[t]}^r,X_{t}^a)-H(Y_t|S_{[t]}^r,Y_{[t-1]},W_a, W_b,X_{t}^a,T_{t}^b)\right]\nonumber\\
&\overset{(iii)}{=}&\sum_{t=1}^n\left[H(Y_t|S_{[t]}^r,X_{t}^a)-H(Y_t|S_{[t]}^r,X_{t}^a,T_t^b)\right]\nonumber\\
&=&\sum_{t=1}^n I(T_t^b;Y_t|X_t^a,S_{[t]}^r)\label{eq:conv-ena-corr-gen-1}
\end{eqnarray}}where $(i)$ follows since state is i.i.d., where $T_t^b$ is the Shannon strategy induced by encoder $b$ at time $t$ as shown in (\ref{eq:sha-str-ind-enb-coop}), and $(ii)$ is valid since conditioning reduces entropy, and $(iii)$ is valid since state is i.i.d. and can be shown along the similar lines as (\ref{eq:conv-sr-corr-3a}). Hence, one can directly obtain that{\allowdisplaybreaks
\begin{eqnarray}
R_b&\leq& \sum_{\mathbf{\mu_r} \in {\cal S}_r^{(n)}}\alpha_{\mathbf{\mu_r}} I(T_t^b;Y_t|X_t^a,S_t^r,S_{[t-1]}^r=\mathbf{\mu_r})+\eta(\epsilon)\label{eq:co-com-corr-gen-2}\\
R_a+R_b&\leq& \sum_{\mathbf{\mu_r} \in {\cal S}_r^{(n)}}\alpha_{\mathbf{\mu_r}} I(X_t^a,T_t^b;Y_t|S_t^r,S_{[t-1]}^r=\mathbf{\mu_r})+\eta(\epsilon)\label{eq:co-com-corr-gen-3}
\end{eqnarray}}where $\alpha_{\mathbf{\mu_r}}:= \frac{1}{n}P_{S_{[t-1]}^r}(\mathbf{\mu_r})$ and $\eta(\epsilon)$ is given in (\ref{eq:eta-epsilon}). We now need to show that the joint distribution $P_{X_t^a,T_t^b,Y_t,S_t^r|S_{[t-1]}^r}(x^a,t^b,y,s^r|\mathbf{\mu_r})$ satisfies (\ref{eq:joi-dist-corr-gen}). Let $\pi_{X^a, T^b}^{\mathbf{\mu_r}}(x^a,t^b):=P_{X_t^a,T_t^b|S_{[t-1]}^r}(x^a,t^b|\mathbf{\mu_r})$ and observe that{\allowdisplaybreaks
\begin{eqnarray}
&&\hspace{-0.5in}P_{X_t^a,T_t^b,Y_t,S_t^r|S_{[t-1]}^r}(x^a,t^b,y,s^r|\mathbf{\mu_r})\nonumber\\
&=&\sum_{s_t^b \in {\cal S}^b}\sum_{s_t \in {\cal S}}P_{Y_t|X_t^a,X_t^b,S_t}(y|x^a,t^b(s_t^b),s)P_{S_t^b,S_t,S_r}(s_t^b,s_t,s^r)P_{X_t^a,T_t^b|S_{[t-1]}^r}(x^a,t^b|\mathbf{\mu_r})\nonumber\\
&=&\pi_{X^a, T^b}^{\mathbf{\mu}}(x^a,t^b)P_{S_t^r}(s^r)P_{Y_t|X_t^a,T_t^b,S_t^r}(y|x^a,t^b,s^r)\label{eq:conv-fact-corr-gen-4}
\end{eqnarray}}where the first equality is verified by (\ref{eq:ch-recfl}) and by the fact that $(X_t^a,T_t^b)$ is independent of $(S_t,S_t^b,S_t^r)$.
\end{proof}
\begin{remark}\label{rem:coop-gen-1}
It should be observed that unlike Theorem \ref{theo:main-corr} and results in the previous sections, for the validity of Theorem \ref{theo:main-corr-gen}, it is not required to have a Markov condition on $P_{S_t,S_t^b, S_t^r}(s_t,s_t^b,s_t^r)$ such as the one given in (\ref{eq:noi-recnoi-mrk}). Furthermore, the result also holds with no CSIT, i.e., ${\cal S}_r=\emptyset$ is allowed, and in this case Theorem \ref{theo:main-corr-gen} is as an extension of \cite[Theorem 4]{som-sha-verdu} to a noisy setup.
\end{remark}

Note that for the setup given in \cite[Theorem 4]{som-sha-verdu}, Theorem \ref{theo:main-corr-gen} provides an equivalent characterization. Recall that in \cite[Theorem 4]{som-sha-verdu} the informed encoder has full CSI, i.e., $X_t^b=\phi_{t}^{(b)}(W_a, W_b, S_{[t]})$, both the uniformed encoder and the decoder has no CSI and the capacity region, ${\cal C}_{AS}$, is given as the closure of all rate pairs $(R_a,R_b)$ satisfying
\begin{eqnarray}
R_b&<&I(U;Y|X^a) \label{eq:rate-somek-enb}\\
R_b+R_a &<& I(U,X^a;Y) \label{eq:rate-somek}
\end{eqnarray}
for some joint measure on ${\cal S}\times{\cal X}_a\times{\cal X}_b\times{\cal Y}\times{\cal U}$ having the form
\begin{eqnarray}
P_{Y|X^a,X^b,S}(y|x^a,x^b,s)P_{X^b|U,X^a,S}(x^b|u,x^a,s)P_{S}(s)P_{X^a,U}(x^a,u)\label{eq:dist-somek},
\end{eqnarray}
where $|{\cal U}|\leq|{\cal S}||{\cal X}_a||{\cal X}_b|+1$. On the other hand, for this setup, Theorem \ref{theo:main-corr-gen} gives the capacity region, ${\cal C}^G_{FS}$, as $\overline{co}\bigg(\bigcup_{\hat{\pi}}\mathcal{R}_{C}^{'}(\hat{\pi})\bigg)$ where $\mathcal{R}_{C}^{'}(\hat{\pi})$ denotes the region of all rate pairs $R=(R_a,R_b)$ satisfying
\begin{eqnarray}
R_b &<& I(T;Y|X^{a}) \quad  \label{eq:ra2-corr-alt}\\
R_a+R_b &<& I(T,X^a;Y) \label{eq:ra3-corr-alt}
\end{eqnarray}
where $P_{Y,T,X^a,X^b,S}(y,t,x^a,x^b,s)$ factorizes as
\begin{eqnarray}
P_{Y|X^a,X^b,S}(y|x^a,x^b,s)P_{X^b|S,T}(x^b|s,t)P_{S}(s)\hat{\pi}_{X^a, T}(x^a,t)\label{eq:joi-dist-corr-alt},
\end{eqnarray}
and $T: {\cal S}\rightarrow{\cal X}_b$.

Although the relation between an auxiliary variable and Shannon strategies is well understood for the single-user case (e.g., see \cite[Section 3.2]{keshet}), we believe that it requires more attention in the multi user case; in particular, note the difference between $|{\cal U}|$ and $|{\cal T}|$. Hence, we provide a proof for ${\cal C}^G_{FS}={\cal C}_{AS}$, see Appendix \ref{app:3}.

We conclude this section with the following remark.
\begin{remark}\label{rem:coop-gen-2}
For the validity of converse proof of Theorem \ref{theo:main-corr-gen} it is crucial that $X_t^a$ only depends on $W_a$. To be more explicit, let us assume ${\cal S}_r=\emptyset$ and consider the following steps of the converse
\begin{eqnarray}
I(W_b;Y_{[n]})&\leq& \sum_{t=1}^nH(Y_t|Y_{[t-1]},X_{[n]}^a)-H(Y_t|Y_{[t-1]},W_a,W_b,X_{[n]}^a,T_t^b)\nonumber\\
&=&\sum_{t=1}^nH(Y_t|Y_{[t-1]},X_{[n]}^a)-H(Y_t|Y_{[t-1]},X_{t}^a,T_t^b)\label{eq:corr-alt-cv5}.
\end{eqnarray}
Since $S_t$ is not available to the decoder, the above equality is valid if and only if $X_{[n]}^a$ does not provide any information about $S_t$. Hence, in other words, whether CSITs are noisy or not, if there is no CSI or noisy CSI at the decoder, the arguments above would fail if the uninformed encoder observes some degree of CSI, i.e., $d_a<\infty$ so that $X_{[n]}^a$ carry some information about $(S_t, S_t^b, S_t^r)$.
\end{remark}
\section{Examples}\label{sec:example}
We present two examples. In the first example we discuss the state dependent modulo-additive MAC with noisy CSIT and complete CSIR (as in Section \ref{sec:causal}) and show that the proposed inner and outer bounds are tight and yield the capacity region. In the second example we consider the problem defined in Section \ref{subsec:csit-fnc-rn} where the channel is a binary multiplier MAC with state being an interference sequence.

\subsection{Modulo-Additive FS-MAC with Noisy CSIT and Complete CSIR}\label{sec:ex-modmac}
Recall that both the achievable regions and the sum-rate capacities of Sections \ref{sec:causal} and \ref{sec:noisy} are given in terms of Shannon-strategies. Hence, their computation requires an optimization over an extended space of the input alphabet to a space of strategies and is often hard; in fact, very few explicit solutions exist even in the single-user case. In \cite{erez-zamir} symmetric, modulo-additive, single-user finite-state channel with complete CSIT is considered and a closed-form solution for the capacity is derived. Based on this result, we now consider the modulo-additive FS-MAC with asymmetric noisy CSIT and show that for the sum-rate capacity, the optimal set of strategies has uniform distribution. This enable us to determine the entire capacity region by observing that under the uniform distribution both inner and outer bounds are tight.

To be more explicit, we consider a two-user FS-MAC in which the channel noise, defined by a process $\{Z_t\}_{t=1}^{\infty}$, is correlated with the state process. The channel is given by $Y=X^a\oplus X^b\oplus Z$ where ${\cal X}_a={\cal X}_b={\cal Y}={\cal Z}=\{0,\cdots,q-1\}$ and ${Z}$, is conditionally independent of $(X^a,X^b)$ given the state $S$ and in the sequel addition (and subtraction) is understood to be performed mod-$q$. Assume further that we have the setup of Section \ref{sec:causal}. The following theorem is the main result of this example and can be though as an extension of \cite[Theorem 1]{erez-zamir} to a noisy multi-user setting.
\begin{theorem}\label{theo:mod-mac}
The capacity region of the modulo-additive FS-MAC defined above is given by the closure of the rate pairs $(R_a,R_b)$ satisfying{\allowdisplaybreaks
\begin{eqnarray}
R_a<\log q - H_{\min}\nonumber\\
R_b<\log q - H_{\min} \nonumber\\
R_a+R_b<\log q - H_{\min} \label{eq:mod-mac1}
\end{eqnarray}}where $H_{\min}:=\min_{t^a,t^b}H(Z+t^a(S^a)+t^b(S^b)|S)$.
\end{theorem}
\begin{proof}
First, recall the rate condition given in Theorem \ref{theo:outbnd-cau-fl};
\begin{eqnarray}
R_a+R_b&\leq& H(Y|S)-H(Y|T^a,T^b,S)\label{eq:mod-mac2}.
\end{eqnarray}
The sketch of the proof is to first determine the optimal distributions of $t^a,t^b$, the distributions achieving the sum-rate capacity, and then concluding with the fact that these distributions yield the same inner bound. Let us first consider $H(Y|T^a,T^b,S)$. Clearly, $P_{Y|X^a,X^b,S}(y|x^a,x^b,s)=P_{Z|S}(y-x^a-x^b|s)$ and  $H(Y|T^a,T^b,S)\geq \min_{t^a,t^b}H(Y|T^a=t^a,T^b=t^b,S)$. Observe that
\begin{eqnarray}
P_{Y|T^a,T^b,S}(y|t^a,t^b,s)&=&\sum_{s^a,s^b}P_{Y|T^a,T^b,S^a,S^b,S}(y|t^a,t^b,s^a,s^b,s)P_{S^a,S^b|S}(s^a,s^b|s)\nonumber\\
&=&\sum_{s^a,s^b}P_{Z|S}(Z=y-t^a(s^a)-t^b(s^b)|s)P_{S^a,S^b|S}(s^a,s^b|s)\nonumber\\
&=&P_{Z+t^a(S^a)+t^b(S^b)|S}(y|s).\label{eq:mod-mac3}
\end{eqnarray}
where the second step is valid since $Z$ is conditionally independent of $(S^a,S^b)$ given $S$. Therefore, $H(Y|T^a=t^a,T^b=t^b,S)=H(Z+t^a(S^a)+t^b(S^b)|S)$. Let $(t^{a*},t^{b*})$ be two mappings from ${\cal S}_a$ to ${\cal X}_a$ and  ${\cal S}_b$ to ${\cal X}_b$ for which $H(Y|T^a=t^{a*},T^b=t^{b*},S)=H_{\min}$. Now, by Corollary \ref{cor:sumra-cau-fl}, we have{\allowdisplaybreaks
\begin{eqnarray}
{\cal C}^{FS}_{\sum}&=&\sup_{\pi_{T^a}(t^a)\pi_{T^b}(t^b)}\left[H(Y|S)-H(Y|T^a,T^b,S)\right]\nonumber\\
&\leq&\sup_{\pi_{T^a}(t^a)\pi_{T^b}(t^b)}H(Y|S) - H_{\min} \label{eq:mod-mac4},
\end{eqnarray}}and we now determine the policies $\{\pi_{T^a}(t^a), \enspace t^a \in {\cal T}_a\}$ and $\{\pi_{T^b}(t^b), \enspace t^b \in {\cal T}_b\}$ achieving the supremum above. Let us first define the following class of strategies
\begin{eqnarray}
{\cal T}_a^{*}&:=&\{t^{a}_{\tau}\}, \enspace \mbox{where} \enspace t^{a}_{\tau}(s^a)=t^{a*}(s^a)+\tau, \enspace \tau=1,\cdots, q \label{eq:mod-mac5}\\
{\cal T}_b^{*}&:=&\{t^{b}_{\tau}\}, \enspace \mbox{where} \enspace t^{b}_{\tau}(s^b)=t^{b*}(s^b)-\tau, \enspace \tau=1,\cdots, q \label{eq:mod-mac6}.
\end{eqnarray}
It should be noted that $H(Y|T^a=t^{a*},T^b=t^{b*},S)=H(Y|T^a=t^{a}_{\tau},T^b=t^{b}_{\tau},S)$ since $H(Y|T^a=t^a,T^b=t^b,S)=H(Z+t^a(S^a)+t^b(S^b)|S)$. Note that $H(Y|S)\leq \log |{\cal Y}|=\log q$, but if we choose $T^a$ and $T^b$ uniformly distributed within ${\cal T}_a^{*}$ and ${\cal T}_b^{*}$, respectively (with zero mass on strategies not in ${\cal T}_a^{*}$ and ${\cal T}_b^{*}$), we would get
{\allowdisplaybreaks
\begin{eqnarray}
P_{Y|S}(y|s)&\overset{(i)}{=}&\sum_{s^a,s^b}\sum_{t^a \in {\cal T}_a^{*}}\sum_{t^b \in {\cal T}_b^{*}}P_{Y|T^b,T^b,S^a,S^b,S}(y|t^a,t^b,s^a,s^b,s)\frac{1}{q^2}P_{S^a,S^b|S}(s^a,s^b|s)\nonumber\\
&=&\sum_{s^a,s^b}P_{S^a,S^b|S}(s^a,s^b|s)\frac{1}{q^2}\sum_{t^a \in {\cal T}_a^{*}}\sum_{t^b \in {\cal T}_b^{*}}P_{Z|S}(y-t^a(s^a)-t^b(s^b)|s)\nonumber\\
&\overset{(ii)}{=}&\sum_{s^a,s^b}P_{S^a,S^b|S}(s^a,s^b|s)\frac{1}{q^2}\sum_{t^a \in {\cal T}_a^{*}}1\nonumber\\
&\overset{(iii)}{=}&\frac{1}{q}
\end{eqnarray}}where $(i)$ valid since $T^a$ and $T^b$ are uniformly distributed, $(ii)$ is due to (\ref{eq:mod-mac6}) (i.e., follows from the fact that ${t^b \in {\cal T}_b^{*}}$ traces all possible values of $Z$) and finally, $(iii)$ is valid since $|{\cal T}_a^{*}|=q$. Therefore, we get that ${\cal C}_{FS}^{\sum}=\log q- H_{\min}$ which is achieved by
\begin{eqnarray}
\pi_{T^a}(t^a)=\frac{1}{q}, \enspace \forall t^a \in {\cal T}_a^{*}, \enspace \pi_{T^b}(t^b)=\frac{1}{q}, \enspace \forall t^b \in {\cal T}_b^{*}\label{eq:mod-mac7}.
\end{eqnarray}
Let us now consider the inner bound. In particular, we need to show that the sets of policies in (\ref{eq:mod-mac7}) give $H(Y|T^a,S)=H(Y|T^b,S)=\log q$. Consider $H(Y|T^a,S)$ and observe that{\allowdisplaybreaks
\begin{eqnarray}
P_{Y|T^a,S}(y|t^a,s)&\overset{(iv)}{=}&\sum_{s^a,s^b}\sum_{t^b \in {\cal T}_b^{*}}P_{Y|T^b,T^b,S^a,S^b,S}(y|t^a,t^b,s^a,s^b,s)\frac{1}{q}P_{S^a,S^b|S}(s^a,s^b|s)\nonumber\\
&=&\sum_{s^a,s^b}P_{S^a,S^b|S}(s^a,s^b|s)\frac{1}{q}\sum_{t^b \in {\cal T}_b^{*}}P_{Z|S}(y-t^a(s^a)-t^b(s^b)|s)\nonumber\\
&\overset{(v)}{=}&\sum_{s^a,s^b}P_{S^a,S^b|S}(s^a,s^b|s)\frac{1}{q}\nonumber\\
&=&\frac{1}{q}
\end{eqnarray}}where $(iv)$ is valid since $T^b$ is uniformly distributed and $(v)$ is due to (\ref{eq:mod-mac6}) (i.e., follows from the fact that ${t^b \in {\cal T}_b^{*}}$ traces all possible values of $Z$). Thus, $H(Y|T^a,S)=\log q$. It can be shown similarly that under (\ref{eq:mod-mac7}) $H(Y|T^b,S)=\log q$.
\end{proof}

Finally, it is easy to see that when there is no side information at the encoders and at the decoder the capacity region of modulo-addtive FS-MAC is given by the closure of rate pairs $(R_a,R_b)$ where{\allowdisplaybreaks
\begin{eqnarray}
R_a\leq \log q - H(Z)\nonumber\\
R_b\leq \log q - H(Z)  \nonumber\\
R_a+R_b\leq \log q - H(Z).
\end{eqnarray}} Observe that we have{\allowdisplaybreaks
\begin{eqnarray}
H(Z+t^{a}(S^a)+t^{b}(S^b)|S)&\leq& H(Z|S)+H(t^{a}(S^a)+t^{b}(S^b)|S)\nonumber\\
H_{\min}=\min_{t^a,t^b}H(Z+t^{a}(S^a)+t^{b}(S^b)|S)&\leq& \min_{t^a,t^b}\left[H(Z|S)+H(t^{a}(S^a)+t^{b}(S^b)|S)\right]\nonumber\\
&\overset{(vi)}{=}& H(Z|S)\nonumber\\
&\overset{(vii)}{<}&H(Z)\nonumber
\end{eqnarray}}where $(vi)$ can be achieved with any deterministic mapping and $(vii)$ is valid since $Z$ and $S$ (and hence $S$) are correlated. Therefore, availability of state information strictly increases, by an amount of at least $I(S;Z)$, the capacity region of the modulo-additive FS-MAC.
\subsection{Binary Multiplier FS-MAC with Interference}
Consider the binary multiplier MAC with state process interfering the output, namely $Y=X^aX^b\oplus S$ where ${\cal X}_a={\cal X}_b={\cal Y}={\cal S}=\{0,1\}$. Assume further that the communication setup is given as in Section \ref{subsec:csit-fnc-rn} with $S^r=S\oplus Z^r$ where $Z^r\sim$ $\mbox{Ber}(p_r)$ is Bernoulli with $P(Z^r=1)=p_r$ . We now show that the capacity region, with both causal and non-causal coding, of this channel is given by the closure of $(R_a,R_b)$ where $R_a<1-H(S|S^r)$, $R_b<1-H(S|S^r)$ and $R_a+R_b<1-H(S|S^r)$.

First recall the capacity region given in Theorem \ref{theo:main-inp-ncas} and observe that $H(Y|S^r,X^a,X^b)=H(X^aX^b\oplus S|S^r,X^a,X^b)=H(S|S^r,X^a,X^b)=H(S|S^r)$, where the last equality follows from (\ref{eq:noi-recnoi-mrk}). Hence, input distributions do not effect $H(Y|S^r,X^a,X^b)$. Obviously, $H(Y|S^r)\leq1$, $H(Y|S^r,X^a)\leq1$ and $H(Y|S^r,X^b)\leq1$ and we now show that equalities can be achieved. More explicitly, we have the following optimizing distributions which can be shown using basic inequalities{\allowdisplaybreaks
\begin{eqnarray}
\argmax_{\pi_{X^a|S^a}(x^a|f^a(s^r)), \pi_{X^b|S^b}(x^b|f^b(s^r))}H(Y|S^r)&=&\left\{\pi_{X^a|S^a}(0|f^a(0))=\pi_{X^a|S^a}(0|f^a(1))=0.5,\right.\nonumber\\
&&\hspace{-0.2in}\left.\pi_{X^b|S^b}(0|f^b(0))=\pi_{X^b|S^b}(0|f^b(1))=0.5\right\}\label{eq:bin-mac-opt-pi1}\\
\argmax_{\pi_{X^a|S^a}(x^a|f^a(s^r)), \pi_{X^b|S^b}(x^b|f^b(s^r))}H(Y|S^r,X^a)&=&\left\{\pi_{X^a|S^a}(0|f^a(0))=\pi_{X^a|S^a}(0|f^a(1))=0,\right.\nonumber\\
&&\hspace{-0.2in}\left.\pi_{X^b|S^b}(0|f^b(0))=\pi_{X^b|S^b}(0|f^b(1))=0.5\right\}\label{eq:bin-mac-opt-pi2}\\
\argmax_{\pi_{X^a|S^a}(x^a|f^a(s^r)),
\pi_{X^b|S^b}(x^b|f^b(s^r))}H(Y|S^r,X^b)&=&\left\{\pi_{X^b|S^b}(0|f^b(0))=\pi_{X^b|S^b}(0|f^b(1))=0,\right.\nonumber\\
&&\hspace{-0.2in}\left.\pi_{X^a|S^a}(0|f^b(0))=\pi_{X^b|S^b}(0|f^b(1))=0.5\right\}\label{eq:bin-mac-opt-pi3}
\end{eqnarray}}and in the rest, let us show that these yield the equalities in the conditional entropies. Let us start with $R_a$, i.e., $H(Y|S^r,X^b)$. Note that
\begin{eqnarray}
H(Y|S^r,X^b)=\sum_{s^r \in \{0,1\}}\sum_{x^b\in \{0,1\}}P_{S^r}(s^r)\pi_{X^b|S^b}(x^b|f^b(s^r))H(Y|S^r=s^r,X^b=x^b)\label{eq:bin-mac-ra-1}.
\end{eqnarray}
Substituting (\ref{eq:bin-mac-opt-pi3}) in (\ref{eq:bin-mac-ra-1}) gives
\begin{eqnarray}
H(Y|S^r,X^b)=P_{S^r}(0)H(X^a\oplus S|X^b=1,S^r=0)+P_{S^r}(1)H(X^a\oplus S|X^b=1,S^r=1)\label{eq:bin-mac-ra-2}.
\end{eqnarray}
We next show that under (\ref{eq:bin-mac-opt-pi3}) $H(X^a\oplus S|X^b=1,S^r=0)=1$, for which it is enough to show that $P_{X^a\oplus S|X^b,S^r}(0|1,0)=0.5$. We have{\allowdisplaybreaks
\begin{eqnarray}
&&\hspace{-0.7in}P_{X^a\oplus S|X^b,S^r}(0|1,0)\nonumber\\
&=&\sum_{s \in \{0,1\}}\sum_{x^a \in \{0,1\}}P_{X^a\oplus S|S,X^a,X^b,S^r}(0|s,x^a,1,0)P_{S|S^r}(s|0)\pi_{X^a|S^a}(x^a|f^a(0))\label{eq:bin-mac-ra-3}\\
&=&P_{S|S^r}(0|1)\big[0.5P_{X^a\oplus S|S,X^a,X^b,S^r}(0|0,0,1,0)+0.5P_{X^a\oplus S|S,X^a,X^b,S^r}(0|0,1,1,0)\big]\nonumber\\
&&\hspace{-0.1in}+P_{S|S^r}(1|1)\big[0.5P_{X^a\oplus S|S,X^a,X^b,S^r}(0|1,0,1,0)+0.5P_{X^a\oplus S|S,X^a,X^b,S^r}(0|1,1,1,0)\big]\nonumber\\
&=&0.5\nonumber,
\end{eqnarray}
where (\ref{eq:bin-mac-ra-3}) is due to (\ref{eq:noi-recnoi-mrk}) and (\ref{eq:team-pol-inp-ncas}). We can similarly show that $P_{X^a\oplus S|X^b,S^r}(0|1,1)=0.5$ and hence, $H(X^a\oplus S|X^b=1,S^r=1)=1$. Therefore, $H(Y|S^r,X^b)=1$. Since the above derivation is symmetric, under (\ref{eq:bin-mac-opt-pi2}) $H(Y|X^a,S^r)=1$.}

It now remains to show that with (\ref{eq:bin-mac-opt-pi1}) $H(Y|S^r)$ is equal to one. It should be observed that{\allowdisplaybreaks
\begin{eqnarray}
&&\hspace{-0.4in}P_{X^aX^b\oplus S|S^r}(\cdot|s^r)\nonumber\\
&\overset{(i)}{=}&\sum_{x^a,x^b,s \in \{0,1\}}P_{X^aX^b\oplus S|X^a,X^b,S}(\cdot|x^a,x^b,s)\pi_{X^a|S^a}(x^a|f^a(s^r))\pi_{X^b|S^b}(x^b|f^b(s^r))P_{S|S^r}(s|s^r)\nonumber\\
&\overset{(ii)}{=}&0.25\sum_{s \in \{0,1\}}P_{S|S^r}(s|s^r)\sum_{x^a,x^b \{0,1\}}P_{X^aX^b\oplus S|X^a,X^b,S}(\cdot|x^a,x^b,s)\nonumber\\
&=&0.5\nonumber
\end{eqnarray}}where $(i)$ is due to (\ref{eq:noi-recnoi-mrk}) and (\ref{eq:team-pol-inp-ncas}), $(ii)$ is due to (\ref{eq:bin-mac-opt-pi1}) and the last step is valid since for given $s$, there are only two pairs of $(x^a,x^b)$ for which $P_{X^aX^b\oplus S|X^a,X^b,S}(\cdot|x^a,x^b,s)=1$ (and zero for the other twos). Hence, $H(Y|S^r)=1$.

Finally, it can be easily shown that the capacity region of $Y=X^aX^b\oplus S$ without CSIT and CSIR is given by the closure of $(R_a,R_b)$ where $R_a<1-H(S)$, $R_b<1-H(S)$ and $R_a+R_b<1-H(S)$. Therefore, availability of noisy CSI at the encoders (both causal and non-causal) and at the decoder increases the capacity region by an amount of $I(S;S^r)$.

\section{Conclusion and Remarks}\label{sec:conc}
We have considered several scenarios for the memoryless FS-MAC with asymmetric noisy CSI at the encoders and complete and noisy CSI at the receiver. When the encoders have access to causal noisy CSI, single letter inner and outer bounds, which are tight for the sum-rate capacity, are obtained. Furthermore, under the assumption that CSI at the encoders are provided by the decoder through noisy feedback links, we demonstrate that a tight converse for the sum-rate capacity still holds if the decoder also observes noisy CSI. In order to reduce the space of optimization, from Shannon strategies to channel inputs, we consider the case where CSITs are asymmetric deterministic functions of noisy CSIR. The equivalent channel demonstrates that the causal setup of this problem is considered in \cite{como-yuksel} and a single-letter characterization for capacity region is provided. Hence, we also considered the non-causal setup and showed that the causal and non-causal capacity regions are identical.

When the decoder does not need to access the current CSI at the encoder, which matches with the delayed scenario, we observe that a single letter characterization of the capacity region can be obtained when the channel state is an i.i.d. stochastic process. We further discuss a cooperative scenario and show that when the common message encoder does not have an access to the current noisy CSI, due to delay, it is possible to obtain a single letter expression for the capacity region. Since a product form is not required in a cooperative scenario, we observed that as soon as the common message encoder does not have access to CSI, then in any noisy setup, covering the cases where no CSIR or noisy CSIR, it is possible to obtain the capacity region.

Finally, the following further problems are worth to be explored: the complete characterization of the capacity region for the problem defined in Section \ref{sec:causal} and its non-causal extension, the cooperative FS-MAC where both encoders observe causal noisy CSI and the cooperative FS-MAC where informed encoder observe noisy CSI non-causally and the other encoder observes noisy CSI with delay.
\bigskip

\appendices
\section{Converse Proof of Theorem \ref{theo:main-inp-ncas}: Non-Causal Case}\label{app:0}
\begin{proof}
Let
\begin{eqnarray}
\alpha_{\mathbf{\mu_{p,f}}}:= \frac{1}{n}P_{S_{[1,t-1]}^r,S_{[t+1,n]}^r}(\mathbf{\mu_p}, \mathbf{\mu_f})\label{eq:alp-mupf}.
\end{eqnarray}
Observe that $(\mathbf{\mu_{p}}:\mathbf{\mu_{f}}) \in {\cal S}_r^{n-1}$, where $(v:w)$ denotes the concatenation of two vectors $v$ and $w$, and
\begin{eqnarray}
\sum_{(\mathbf{\mu_{p}}:\mathbf{\mu_{f}}) \in {{\cal S}_r}^{n-1}}\alpha_{\mathbf{\mu_{p,f}}}=\frac{1}{n}\sum_{1\leq t \leq n}\enspace \sum_{\mathbf{\mu_{p}},\mathbf{\mu_f}}P_{S_{[1,t-1]}^r,S_{[t+1,n]}^r}(\mathbf{\mu_p}, \mathbf{\mu_f})=1\nonumber.
\end{eqnarray}
\begin{lemma}\label{lem:conv-ncas}
Assume that a rate pair $R=(R_a,R_b)$, with block length $n\geq1$ and a constant $\epsilon \in (0,1/2)$, is achievable. Then,
{\allowdisplaybreaks
\begin{eqnarray}
R_a&\leq&\sum_{(\mathbf{\mu_{p}}:\mathbf{\mu_{f}})}\alpha_{\mathbf{\mu_{p,f}}} I(X_t^a;Y_t|X_t^b,S_t^r,S_{[t-1]}^r=\mathbf{\mu_p},S_{[t+1,n]}^r=\mathbf{\mu_f})+\eta(\epsilon)\label{eq:co-com-ncas-inp-1}\\
R_b&\leq&\sum_{(\mathbf{\mu_{p}}:\mathbf{\mu_{f}})}\alpha_{\mathbf{\mu_{p,f}}} I(X_t^b;Y_t|X_t^a,S_t^r,S_{[t-1]}^r=\mathbf{\mu_p},S_{[t+1,n]}^r=\mathbf{\mu_f})+\eta(\epsilon)\label{eq:co-com-ncas-inp-2}\\
R_a+R_b&\leq&\sum_{(\mathbf{\mu_{p}}:\mathbf{\mu_{f}})}\alpha_{\mathbf{\mu_{p,f}}} I(X_t^a,X_t^b;Y_t|S_t^r,S_{[t-1]}^r=\mathbf{\mu_p},S_{[t+1,n]}^r=\mathbf{\mu_f})+\eta(\epsilon)\label{eq:co-com-ncas-inp-3}
\end{eqnarray}
}
\end{lemma}
\begin{proof}
Let us first consider the sum-rate. With standard steps, we get
\begin{eqnarray}
R_a+R_b &\leq& \frac{1}{1-\epsilon}\frac{1}{n}\left(I(\mathbf{W};Y_{[n]}, S_{[n]}^r)+H(\epsilon)\right)\label{eq:conv-sr-inp-ncas-2}.
\end{eqnarray}
Note that since $S_{[n]}^r$ is independent of $\mathbf{W}$, we have $I(\mathbf{W};Y_{[n]},S_{[n]}^r)=I(\mathbf{W};Y_{[n]}|S_{[n]}^r)$ and{\allowdisplaybreaks
\begin{eqnarray}
I(\mathbf{W};Y_{[n]}|S_{[n]}^r)&=&\sum_{t=1}^{n}\left[H(Y_t|S_{[n]}^r,Y_{[t-1]})-H(Y_t|\mathbf{W},S_{[n]}^r,Y_{[t-1]})\right]\nonumber\\
&\overset{(i)}{\leq}&\sum_{t=1}^{n}\left[H(Y_t|S_{[n]}^r)-H(Y_t|\mathbf{W},S_{[n]}^r,Y_{[t-1]})\right]\nonumber\\
&\overset{(ii)}{=}&\sum_{t=1}^{n}\left[H(Y_t|S_{[n]}^r)-H(Y_t|\mathbf{W},S_{[n]}^r,Y_{[t-1]},\mathbf{X}_{[n]})\right]\nonumber\\
&\overset{(iii)}{=}&\sum_{t=1}^{n}\left[H(Y_t|S_{[n]}^r)-H(Y_t|S_{[n]}^r,\mathbf{X}_t)\right]\nonumber\\
&=&\sum_{t=1}^{n}I(\mathbf{X}_t;Y_t|S_{[n]}^r)\label{eq:conv-sr-inp-ncas-3}
\end{eqnarray}}where $(i)$ follows since conditioning reduces entropy, $(ii)$ holds since $X_t^i=\phi^{(i)}_t(W_i,f^i(S_{[n]}^r))$, $i=\{a,b\}$, and $(iii)$ is due to (\ref{eq:ch-recfl}). Combining (\ref{eq:conv-sr-inp-ncas-2}) and (\ref{eq:conv-sr-inp-ncas-3}) similar to (\ref{eq:conv-cau-fl-6a}), gives
\begin{eqnarray}
R_a+R_b\leq\frac{1}{n}\sum_{t=1}^nI(X_t^a,X_t^b;Y_t|S_{[n]}^r)+\eta(\epsilon)\label{eq:conv-sr-inp-ncas-4}
\end{eqnarray}
Furthermore,
\begin{eqnarray}
I(X_t^a,X_t^b;Y_t|S_{[n]}^r)=n\sum_{(\mathbf{\mu_{p}}:\mathbf{\mu_{f}})}\alpha_{\mathbf{\mu_{p,f}}}I(X_t^a,X_t^b;Y_t|S_t^r, S_{[t-1]}^r=\mathbf{\mu_p},S_{[t+1,n]}^r=\mathbf{\mu_f}),\label{eq:conv-sr-inp-ncas-4a}
\end{eqnarray}
and substituting the above into (\ref{eq:conv-sr-inp-ncas-4}) yields (\ref{eq:co-com-ncas-inp-3}).

Let us now consider encoder $a$. Using Fano's inequality and standard steps we first get,
\begin{eqnarray}
R_a\leq \frac{1}{1-\epsilon}\frac{1}{n}\left(I(W_a;Y_{[n]}, S_{[n]}^r)+H(\epsilon)\right)\label{eq:conv-sr-inp-ncas-5}.
\end{eqnarray}
Furthermore,
{\allowdisplaybreaks
\begin{eqnarray}
I(W_a;Y_{[n]}, S_{[n]}^r)&\overset{(i)}{\leq}&I(W_a;Y_{[n]}|S_{[n]}^r,W_b)\nonumber\\
&=&\sum_{t=1}^n\left[H(Y_t|S_{[n]}^r,Y_{[t-1]},W_b)-H(Y_t|S_{[n]}^r,Y_{[t-1]},\mathbf{W})\right]\nonumber\\
&\overset{(ii)}{\leq}&\sum_{t=1}^n\left[H(Y_t|S_{[n]}^r,W_b)-H(Y_t|S_{[n]}^r,Y_{[t-1]},\mathbf{W})\right]\nonumber\\
&\overset{(iii)}{=}&\sum_{t=1}^n\left[H(Y_t|S_{[n]}^r,W_b,X_{[n]}^b)-H(Y_t|S_{[n]}^r,Y_{[t-1]},\mathbf{W},\mathbf{X}_{[n]})\right]\nonumber\\
&\overset{(iv)}{\leq}&\sum_{t=1}^n\left[H(Y_t|S_{[n]}^r,X_{t}^b)-H(Y_t|S_{[n]}^r,Y_{[t-1]},\mathbf{W},\mathbf{X}_{[n]})\right]\nonumber\\
&\overset{(v)}{=}&\sum_{t=1}^n\left[H(Y_t|S_{[n]}^r,X_{t}^b)-H(Y_t|S_{[n]}^r,X_{t}^b,X_t^a)\right]\nonumber\\
&=&\sum_{t=1}^n I(X_t^a;Y_t|X_t^b,S_{[n]}^r)\label{eq:conv-sr-inp-ncas-6}
\end{eqnarray}
}where $(i)$ is due to (\ref{eq:sta-iid}) and conditioning reduces entropy, $(ii)$ holds since conditioning reduces entropy, $(iii)$ holds since $X_t^i=\phi^{(i)}(W_i,f^i(S_{[n]}^r))$, $i=\{a,b\}$, $(iv)$ is valid since conditioning reduces entropy and finally, $(v)$ is valid due to (\ref{eq:ch-recfl}) and $S_t^i$, $i=\{a,b\}$, being a function of $S_t^r$.

Now combining (\ref{eq:conv-sr-inp-ncas-5})-(\ref{eq:conv-sr-inp-ncas-6}) and following steps akin to (\ref{eq:conv-sr-inp-ncas-4}) and (\ref{eq:conv-sr-inp-ncas-4a}), we can verify (\ref{eq:co-com-ncas-inp-1}). To verify (\ref{eq:co-com-ncas-inp-2}) for encoder $b$ it is enough to switch the roles of encoder $a$ and $(b)$.
\end{proof}
Observe now that for any $t\geq1$, $I(X_t^a,X_t^b;Y_t|S_t^r,S_{[t-1]}^r=\mathbf{\mu_p}, S_{[t+1,n]}^r=\mathbf{\mu_f})$ is a function of the conditional distribution $P_{X_t^a,X_t^b,Y_t,S_t^r|S_{[t-1]}^r, S_{[t+1,n]}^r}(x_t^a,x_t^b,y_t,s_t^r|\mathbf{\mu_p},\mathbf{\mu_f})$. Hence, we need to show that this distribution factorizes as in (\ref{eq:joi-dist-inp-ncas}). Let
\begin{eqnarray}
\Upsilon_{\mathbf{\mu_p},\mathbf{\mu_f}}^a(x^a,f^a(s^r))&:=&\{w_a:\phi_{t}^{(a)}\left(w_a,f^a(\mathbf{\mu_p},\mathbf{\mu_f}),f^a(s^r)\right)=x^a\}, \nonumber\\
\Upsilon_{\mathbf{\mu_p},\mathbf{\mu_f}}^b(x^b,f^b(s^r))&:=&\{w_b:\phi_{t}^{(b)}\left(w_b,f^b(\mathbf{\mu_p},\mathbf{\mu_f}),f^b(s^r)\right)=x^b\}\label{eq:lem-fact-ncas-set1}
\end{eqnarray}
and
\begin{eqnarray}
\pi_{X^a|S^a}^{\mathbf{\mu_p},\mathbf{\mu_f}}\left(x^a|f^a(s^r)\right)&:=&\sum_{w_a\in \Upsilon_{\mathbf{\mu_p},\mathbf{\mu_f}}^a(x^a,f^a(s^r))}\frac{1}{|{\cal W}_a|}, \nonumber\\
\pi_{X^b|S^b}^{\mathbf{\mu_p},\mathbf{\mu_f}}\left(x^b|f^b(s^r)\right)&:=&\sum_{w_b\in \Upsilon_{\mathbf{\mu_p},\mathbf{\mu_f}}^b(x^b,f^b(s^r))}\frac{1}{|{\cal W}_b|}.\label{eq:lem-fact-ncas-set2}
\end{eqnarray}
\begin{lemma}\label{lem:fact-ncas}
For every $1\leq t\leq n$ and $(\mathbf{\mu_p}:\mathbf{\mu_f}) \in {\cal S}_r^{n-1}$, the following holds
\begin{eqnarray}
&&\hspace{-0.5in}P_{X_t^a,X_t^b,Y_t,S_t^r|S_{[t-1]}^r,S_{[t+1,n]}^r}(x^a,x^b,y,s^r|\mathbf{\mu_p},\mathbf{\mu_f})\nonumber\\
&&\hspace{0.5in}=P_{S^r}(s^r)P_{Y|S^r,X^a,X^b}(y|s^r,x^a,x^b)\pi_{X^a|S^a}^{\mathbf{\mu_p},\mathbf{\mu_f}}(x^a|f^a(s^r))\pi_{X^b|S^b}^{\mathbf{\mu_p},\mathbf{\mu_f}}(x^b|f^b(s^r)).\label{eq:lem-fact-ncas}
\end{eqnarray}
\end{lemma}
\begin{proof}
First observe that due to (\ref{eq:ch-recfl}) we have
\begin{eqnarray}
&&\hspace{-0.8in}P_{X_t^a,X_t^b,Y_t,S_t^r|S_{[t-1]}^r,S_{[t+1,n]}^r}(x^a,x^b,y,s^r|\mathbf{\mu_p},\mathbf{\mu_f})\nonumber\\
&&=P_{Y_t|S_t^r,X_t^a,X_t^b}(y|s^r,x^a,x^b)P_{X_t^a,X_t^b,S_t^r|S_{[t-1]}^r,S_{[t+1,n]}^r}(x^a,x^b,s^r|\mathbf{\mu_p},\mathbf{\mu_f})\label{eq:lem-fact-ncas-1}.
\end{eqnarray}
Let us now consider the second term in (\ref{eq:lem-fact-ncas-1}). We have{\allowdisplaybreaks
\begin{eqnarray}
&&\hspace{-0.6in}P_{X_t^a,X_t^b,S_t^r|S_{[t-1]}^r,S_{[t+1,n]}^r}(x^a,x^b,s^r|\mathbf{\mu_p},\mathbf{\mu_f})\nonumber\\
&=&\sum_{w_a \in {\cal W}_a}\sum_{w_b \in {\cal W}_b}P_{\mathbf{W},X_t^a,X_t^b,S_t^r|S_{[t-1]}^r,S_{[t+1,n]}^r}(\mathbf{w},x^a,x^b,s^r|\mathbf{\mu_p},\mathbf{\mu_f})\nonumber\\
&\overset{(i)}{=}&\sum_{w_a \in {\cal W}_a}\sum_{w_b \in {\cal W}_b}1_{\left\{x^l=\phi^{(l)}\left(w_l,f^l(s^r,\mathbf{\mu_p},\mathbf{\mu_f})\right),\enspace l=a,b\right\}}P_{W_a,W_b,S_t^r|S_{[t-1]}^r,S_{[t+1,n]}^r}(w_a,w_b,s^r|\mathbf{\mu_p},\mathbf{\mu_f})\nonumber\\
&\overset{(ii)}{=}&\sum_{w_a \in {\cal W}_a}\sum_{w_b \in {\cal W}_b}1_{\left\{x^l=\phi^{(l)}\left(w_l,f^l(s^r,\mathbf{\mu_p},\mathbf{\mu_f})\right),\enspace l=a,b\right\}}\frac{1}{|{\cal W}_a|}\frac{1}{|{\cal W}_b|}P_{S_t^r}(s^r)\nonumber\\
&=&P_{S_t^r}(s^r)\sum_{w_a \in {\cal W}_a}\frac{1}{|{\cal W}_a|}1_{\left\{x^a=\phi^{(a)}\left(w_a,f^a(s^r,\mathbf{\mu_p},\mathbf{\mu_f})\right)\right\}}\sum_{w_b \in {\cal W}_b}\frac{1}{|{\cal W}_b|}1_{\left\{x^b=\phi^{(b)}\left(w_b,f^b(s^r,\mathbf{\mu_p},\mathbf{\mu_f})\right)\right\}}\nonumber\\
&\overset{(iii)}{=}&P_{S_t^r}(s^r)\pi_{X^a|S^a}^{\mathbf{\mu_p},\mathbf{\mu_f}}(x^a|f^a(s^r))\pi_{X^b|S^b}^{\mathbf{\mu_p},\mathbf{\mu_f}}(x^b|f^b(s^r))\label{eq:lem-fact-ncas-2}
\end{eqnarray}}where $(i)$ follows since $X_t^i=\phi^{(i)}(W_i,f^i(S_{[n]}^r))$, $i=\{a,b\}$, $(ii)$ is valid since $W_a$ and $W_b$ are independent of $S_{[n]}^r$ and state process being i.i.d. and $(iii)$ follows due to (\ref{eq:lem-fact-ncas-set1}) and (\ref{eq:lem-fact-ncas-set2}). Substituting (\ref{eq:lem-fact-ncas-2}) in (\ref{eq:lem-fact-ncas-1}) completes the proof.
\end{proof}
We can now complete the proof of Theorem \ref{theo:main-inp-ncas}. With Lemma \ref{lem:conv-ncas}, it is shown that any achievable rate pair can be approximated by the convex combinations of rate conditions given in (\ref{eq:ra1-inp-ncas})-(\ref{eq:ra3-inp-ncas}) which are indexed by $(\mathbf{\mu_p},\mathbf{\mu_f})$ and satisfy (\ref{eq:joi-dist-inp-ncas}) for joint state-input-output distributions. Hence, since $\lim_{\epsilon\rightarrow 0}\eta(\epsilon)=0$, any achievable rate pair belongs to $\overline{co}\bigg(\bigcup_{\bar{\pi}}\mathcal{R}_{NS}^Q(\bar{\pi})\bigg)$.
\end{proof}
\section{Converse Proof of Theorem \ref{theo:main-inp}}\label{app:1}
\begin{proof}
In the proof, we will use the fact that the delayed setup can be modeled by taking the last $d_a$, $d_b$ entries of causal setup as empty. Recall that $\alpha_{\mathbf{\mu}}$ is defined in (\ref{eq:alp-mupf}).
\begin{lemma}\label{lem:conv-inp}
Assume that a rate pair $R=(R_a,R_b)$, with block length $n\geq1$ and a constant $\epsilon \in (0,1/2)$, is achievable. Then,
\begin{eqnarray}
R_a&\leq& \sum_{\mathbf{\mu} \in {\cal S}^{(n)}}\alpha_{\mathbf{\mu}} I(X_t^a;Y_t|X_t^b,S_t,S_{[t-1]}=\mathbf{\mu})+\eta(\epsilon)\label{eq:co-com-inp-1}\\
R_b&\leq& \sum_{\mathbf{\mu} \in {\cal S}^{(n)}}\alpha_{\mathbf{\mu}} I(X_t^b;Y_t|X_t^a,S_t,S_{[t-1]}=\mathbf{\mu})+\eta(\epsilon)\label{eq:co-com-inp-2}\\
R_a+R_b&\leq& \sum_{\mathbf{\mu} \in {\cal S}^{(n)}}\alpha_{\mathbf{\mu}} I(X_t^a,X_t^b;Y_t|S_t,S_{[t-1]}=\mathbf{\mu})+\eta(\epsilon)\label{eq:co-com-inp-3}.
\end{eqnarray}
\end{lemma}
\begin{proof}
For the sum-rate, observe that the derivation in (\ref{eq:conv-cau-fl-6}) can be performed to verify (\ref{eq:co-com-inp-3}), as for $d_i\geq1$, $T_t^i=X_t^i$ by taking $S_{[t-d_i+1, t-1]}^i=\emptyset$, $i=\{a,b\}$.

Let us now consider encoder $a$. We have
\begin{eqnarray}
R_a\leq \frac{1}{n}\log(|{\cal W}_a|)\leq \frac{1}{1-\epsilon}\frac{1}{n}\left(I(W_a;Y_{[n]}, S_{[n]})+H(\epsilon)\right)\label{eq:conv-ena-inp-2}.
\end{eqnarray}
Furthermore,{\allowdisplaybreaks
\begin{eqnarray}
I(W_a;Y_{[n]}, S_{[n]})&\overset{(i)}{\leq}&I(W_a;Y_{[n]}, S_{[n]}|W_b,S_{[n]}^b)\nonumber\\
&=&\sum_{t=1}^n\left[H(Y_t,S_t|S_{[t-1]},Y_{[t-1]},W_b,S_{[n]}^b)-H(Y_t,S_t|S_{[t-1]},Y_{[t-1]}, \mathbf{W}, S_{[n]}^b)\right]\nonumber\\
&\overset{(ii)}{=}&\sum_{t=1}^n\left[H(Y_t|S_{[t]},Y_{[t-1]},W_b,S_{[n]}^b)-H(Y_t|S_{[t]},Y_{[t-1]}, \mathbf{W}, S_{[n]}^b)\right]\nonumber\\
&\overset{(iii)}{=}&\sum_{t=1}^n\left[H(Y_t|S_{[t]},Y_{[t-1]},W_b,S_{[n]}^b,X_{[n]}^b)-H(Y_t|S_{[t]},Y_{[t-1]}, \mathbf{W}, S_{[n]}^b,X_{[n]}^b)\right]\nonumber\\
&\overset{(iv)}{\leq}&\sum_{t=1}^n\left[H(Y_t|S_{[t]},X_{t}^b)-H(Y_t|S_{[t]},Y_{[t-1]},\mathbf{W},S_{[n]}^b,X_{[n]}^b,X_{[n]}^a)\right]\nonumber\\
&\overset{(v)}{=}&\sum_{t=1}^n\left[H(Y_t|S_{[t]},X_{t}^b)-H(Y_t|S_{[t]},X_{t}^b,X_t^a)\right]\nonumber\\
&=&\sum_{t=1}^n I(X_t^a;Y_t|X_t^b,S_{[t]})\label{eq:conv-ena-inp-3}
\end{eqnarray}
where $(i)$ is due to (\ref{eq:sta-iid}) and conditioning reduces entropy, $(ii)$ is valid since
\begin{eqnarray}
P_{S_t|S_{t}^b}(s_t|s_{t}^b)&=&P_{S_t|Y_{[t-1]},S_{[t-1]},W_a,W_b,S_{[n]}^b}(s_t|y_{[t-1]},s_{[t-1]},w_a,w_b,s_{[n]}^b)\nonumber\\
&=&P_{S_t|Y_{[t-1]},S_{[t-1]},W_b,S_{[n]}^b}(s_t|y_{[t-1]},s_{[t-1]},w_b,s_{[n]}^b)\label{eq:conv-ena-inp-3a}
\end{eqnarray}
where the second equality is due to (\ref{eq:sta-iid}), $(iii)$ is valid since $X_t^b=\phi_{t}^{(b)} \left(W_b,S_{[t-d_b]}^b\right)$, $(iv)$ is valid since conditioning reduces entropy and finally, $(v)$ is valid by (\ref{eq:ch-recfl}).}

Now, recall that $\chi(\epsilon)=\frac{H(\epsilon)}{n(1-\epsilon)}$ and, combining (\ref{eq:conv-ena-inp-2}) and (\ref{eq:conv-ena-inp-3}) gives
\begin{eqnarray}
R_a&\leq&\frac{1}{n}\sum_{t=1}^nI(X_t^a;Y_t|X_t^b,S_{[t]})+\eta(\epsilon)\label{eq:conv-ena-inp-4}.
\end{eqnarray}
Furthermore,
\begin{eqnarray}
I(X_t^a;Y_t|X_t^b,S_{[t]})=n\sum_{\mathbf{\mu} \in {\cal S}^{(t-1)}}\alpha_{\mathbf{\mu}}I(X_t^a;Y_t|X_t^b,S_t, S_{[t-1]}=\mathbf{\mu}),
\end{eqnarray}
and substituting the above into (\ref{eq:conv-ena-inp-4}) yields (\ref{eq:co-com-inp-1}).

Finally, for encoder $b$, (\ref{eq:co-com-inp-2}) can be verified by following the similar steps of encoder $a$.
\end{proof}
Now since, for any $t\geq1$, conditional mutual information terms given in (\ref{eq:co-com-inp-1})-(\ref{eq:co-com-inp-3}) are functions of $P_{X_t^a,X_t^b,Y_t,S_t|S_{[t-1]}}(x^a,x^b,y,s|\mathbf{\mu})$, in order to complete the proof of the converse, we need to show that this term factorizes as in (\ref{eq:joi-dist-inp}).
\begin{lemma}\label{lem:conv-fact-inp}
For every $1\leq t\leq n$ and $\mathbf{\mu} \in {\cal S}^{t-1}$, the following holds
\begin{eqnarray}
P_{X_t^a,X_t^b,Y_t,S_t|S_{[t-1]}}(x^a,x^b,y,s|\mathbf{\mu})=P_S(s)P_{Y|S,X^a,X^b}(y|s,x^a,x^b)\pi_{X^a}^{\mathbf{\mu}}(x^a)\pi_{X^b}^{\mathbf{\mu}}(x^b).\label{eq:conv-fact-inp-3}
\end{eqnarray}
\end{lemma}
Note that one of the crucial step in verifying the product form for the causal setup, see (\ref{eq:conv-cau-fl-4}) and (\ref{eq:conv-cau-fl-5}), is the independence of Shannon strategies of the current state. This also holds in the delayed setup. Therefore, let
\begin{eqnarray}
&&\Upsilon_{\mathbf{\mu_i}}^i(x^i):=\{w_i:\phi_{t}^{(i)}(w_i,s_{[t-d_i]}^i=\mathbf{\mu_i})=x^i\}, \enspace i=a,b\label{eq:conv-fact-inp-1}
\end{eqnarray}
and
\begin{eqnarray}
\pi_{X^i}^{\mathbf{\mu_i}}(x^i):=\sum_{w_i\in \Upsilon_{\mathbf{\mu_i}}^i(x^i)}\frac{1}{|{\cal W}_i|}, \enspace \pi_{X^i}^{\mathbf{\mu}}(x^i):=\sum_{\mathbf{\mu_i}}\pi_{X^i}^{\mathbf{\mu_i}}(x^i)P_{S_{[t-d_i]}^i|S_{[t-1]}}(\mathbf{\mu_i}|\mathbf{\mu}),\enspace i=a,b.\nonumber
\end{eqnarray}
Hence, (\ref{eq:conv-fact-inp-3}) can be shown following the same steps in Lemma \ref{lem:fact-cau-fl}.

We can now complete the converse proof of Theorem \ref{theo:main-inp}. With Lemma \ref{lem:conv-inp} it is shown that any achievable rate pair can be approximated by the convex combinations of rate conditions which are indexed by $\mathbf{\mu} \in {\cal S}^{(n)}$ and satisfy (\ref{eq:joi-dist-inp}) for joint state-input-output distributions. Hence, any achievable pair $(R_a,R_b) \in \overline{co}\big(\bigcup_{\tilde{\pi}}\mathcal{R}_{DN}(\tilde{\pi})\big)$.
\end{proof}

\section{Achievability and Converse Proofs of Theorem \ref{theo:main-corr}}\label{app:2}
\begin{proof}[Achievability Proof]
Fix $(R_a,R_b)\in \mathcal{R}_{C}(\hat{\pi})$.

\textbf{\textit{Codebook Generation}}
Fix $\pi_{X^a}(x^a)$ and $\pi_{T^b|X^a}(t^b|x^a)$. For each $w_a \in \{1,\cdots, 2^{nR_a}\}$, randomly generate $x_{[n],w_a}^{a}$, each according to $\prod_{i=1}^n\pi_{X_{i}^{a}}(x_{i,w_a}^{a})$. Reveal this codebook to encoder $b$ and, for each $w_b \in \{1,\cdots, 2^{nR_b}\}$, encoder $b$ randomly generates $t_{[n],w_b}^{b}$, each according to $\prod_{i=1}^n\pi_{T_{i}^{b}|X_{i}^a}(t_{i,w_b}^{b}|x_{i,w_a}^{a})$.
These codeword pairs form the codebook, which is revealed to the decoder.

\textbf{\textit{Encoding}}
Define the encoding functions as follows: $x_{i}^{a}(w_a)=\phi_{i}^{a}(w_a,s_{[i-d_a]}^a)$ and $x_{i}^{b}(w_b)=\phi_{i}^{b}(w_b,s_{[i]}^b)=t_{i,w_b}^b(s_i^b)$ where $x_{i,w_a}^a$ and $t_{i,w_b}^b$ denote the $i$th component of $x_{[n],w_a}^{a}$ and $t_{[n],w_b}^{b}$, respectively. Therefore, to send the messages $w_a$ and $w_b$, transmit the corresponding $x_{[n],w_a}^{a}$ and $t_{[n],w_b}^{b}$, respectively.

\textbf{\textit{Decoding}}
After receiving $(y_{[n]}, s_{[n]})$, the decoder looks for the only $(w_a,w_b)$ pair such that $(x_{[n],w_a}^{a}, t_{[n],w_b}^{b},$ $y_{[n]}, s_{[n]})$ are jointly $\epsilon-$typical and declares this pair as its estimate $(\hat{w}_a, \hat{w}_b)$.

\textbf{\textit{Error Analysis}}
Assume that $(w_a,w_b)=(1,1)$ was sent. Let $E_{\alpha,\beta}\bydef\big\{(X_{[n],\alpha}^a,T_{[n],\beta}^b,Y_{[n]},S_{[n]})\in A_{\epsilon}^n\big\}$, $\alpha\in\{1,\cdots,2^{nR_a}\}$ and $\beta\in\{1,\cdots,2^{nR_b}\}$.
Then
\begin{eqnarray}
P_{e}^{n}=P\big(E_{1,1}^c\bigcup_{(\alpha,\beta)\neq(1,1)}E_{\alpha,\beta}\big)\leq P(E_{1,1}^c)+\sum_{\alpha=1,\beta\neq1}P(E_{\alpha,\beta}) + \sum_{\alpha\neq1,\beta\neq1}P(E_{\alpha,\beta})\label{eq:erbound-corr}.
\end{eqnarray}
Since $\{Y_i,S_i,X_i^a,T_i^b\}_{i=1}^{\infty}$ is an i.i.d. sequence hence, $ P(E_{1,1}^c)\rightarrow0$ for $n\rightarrow \infty$. Next, let us consider the second term
{\allowdisplaybreaks
\begin{eqnarray}
\sum_{\alpha=1,\beta\neq1}P(E_{\alpha=1,\beta\neq1})&=&\sum_{\alpha=1,\beta\neq1}P((X_{[n],1}^a,T_{[n],\beta}^b,Y_{[n]},S_{[n]})\in A_{\epsilon}^n)\nonumber\\
&\overset{(i)}{=}&\sum_{\alpha=1,\beta\neq1}\sum_{(x_{[n]}^a,t_{[n]}^b,y_{[n]},s_{[n]})\in A_{\epsilon}^n}P_{T_{[n]}^b|X_{[n]}^a}(t_{[n]}^b|x_{[n]}^a) P_{X_{[n]}^a, Y_{[n]},S_{[n]}}(x_{[n]}^a, y_{[n]},s_{[n]})\nonumber\\
&\overset{}{\leq}&\sum_{\alpha=1,\beta\neq1}|A_{\epsilon}^n|2^{-n[H(T^b|X^a)-\epsilon]}2^{-n[H(X^a,Y,S)-\epsilon]}\nonumber\\
&\leq&2^{nR_b}2^{-n[H(T^b|X^a)+ H(X^a,Y,S) - H(X^a,T^b,Y,S)-3\epsilon]}\nonumber\\
&\overset{(ii)}{=}&2^{n[R_b-I(T^b;Y|S,X^a)-3\epsilon]}\label{eq:erbo1-corr}
\end{eqnarray}
}where $(i)$ holds since $T_{[n],\beta}^b$ is independent of $(Y_{[n]},S_{[n]})$ given $X_{[n],1}^a$ and $(ii)$ follows since{\allowdisplaybreaks
\begin{eqnarray}
&&\hspace{-1.3in}H(T^b|X^a)+ H(X^a,Y,S) - H(X^a,T^b,Y,S)\nonumber\\
&=&H(T^b|X^a)+ H(X^a,Y,S)-H(Y|X^a,T^b,S)-H(X^a,T^b,S)\nonumber\\
&=&H(X^a,Y,S)-H(Y|X^a,T^b,S)-H(X^a,S)\nonumber\\
&=&I(T^b;Y|S,X^a)\nonumber
\end{eqnarray}}where the second equality follows since $T^b$ and $S$ are independent given $X^a$. Finally, {\allowdisplaybreaks
\begin{eqnarray}
\sum_{\alpha\neq1,\beta\neq1}P(E_{\alpha\neq1,\beta\neq 1})&=&\sum_{\alpha\neq1,\beta\neq 1}P((X_{[n],\alpha}^a,T_{[n],\beta}^b,Y_{[n]},S_{[n]})\in A_{\epsilon}^n)\nonumber\\
&\overset{(iii)}{=}&\sum_{\alpha\neq1,\beta\neq 1}\sum_{(x_{[n]}^a,t_{[n]}^b,y_{[n]},s_{[n]})\in A_{\epsilon}^n}P_{T_{[n]}^b,X_{[n]}^a}(t_{[n]}^b,x_{[n]}^a) P_{Y_{[n]},S_{[n]}}(y_{[n]},s_{[n]})\nonumber\\
&\overset{}{\leq}&\sum_{\alpha\neq1,\beta\neq1}|A_{\epsilon}^n|2^{-n[H(T^b,X^a)-\epsilon]}2^{-n[H(Y,S)-\epsilon]}\nonumber\\
&\leq&2^{n(R_a+R_b)}2^{-n[H(T^b,X^a)+ H(Y,S) - H(X^a,T^b,Y,S)-3\epsilon]}\nonumber\\
&\overset{(iv)}{=}&2^{n[R_a+R_b-I(X^a,T^b;Y|S)-3\epsilon]}\label{eq:erbo1-corr-ra}
\end{eqnarray}}where $(iii)$ holds since for $\alpha, \beta \neq 1$, $(T_{[n],\beta}^b, X_{[n],\alpha}^a)$ is independent of $(Y_{[n]},S_{[n]})$ and $(iv)$ follows since
\begin{eqnarray}
&&\hspace{-1.3in}H(T^b,X^a)+ H(Y,S)-H(X^a,T^b,Y,S)\nonumber\\
&=&H(T^b,X^a)+ H(Y,S)-H(Y|X^a,S,T^b)-H(X^a,S,T^b)\nonumber\\
&=&H(T^b,X^a)+ H(Y,S)-H(Y|X^a,S,T^b)-H(X^a,T^b)-H(S)\nonumber\\
&=&I(X^a,T^b;Y|S)\nonumber,
\end{eqnarray}
and the rate conditions of the $\mathcal{R}_{C}(\hat{\pi})$ imply that each term tends in (\ref{eq:erbound-corr}) tends to zero as $n \rightarrow \infty$.
\end{proof}
Note that the main motivation in indexing mutual information terms by the past CSI, is to get a product form on the team policies. In the cooperative setup, we do not require a product form and therefore, the convex combination argument is not essential. However, we herein keep this indexing (see (\ref{eq:joi-dist-corr})) to avoid the use of a time sharing auxiliary random variable.
\begin{proof}[Converse Proof]
First observe that, since $X_t^b=\phi_{t}^{(b)} \left(W_a,W_b,S_{[t-1]}^b, S_t^b\right)$, we have
\begin{eqnarray}
T_t^b=\phi_{t}^{(b)}\left(W_a, W_b,S_{[t-1]}^b\right) \in {{\cal X}_{b}}^{|{\cal S}_b|}. \label{eq:sha-str-ind-enb-coop}
\end{eqnarray}
\begin{lemma}\label{lem:conv-corr}
Let $T_t^b \in {\cal T}_{b}$ be the Shannon strategy induced by $\phi_{t}^{(b)}$ as shown in (\ref{eq:sha-str-ind-enb-coop}). Assume that a rate pair $R=(R_a,R_b)$, with block length $n\geq1$ and a constant $\epsilon \in (0,1/2)$, is achievable. Then,
\begin{eqnarray}
R_b&\leq& \sum_{\mathbf{\mu} \in {\cal S}^{(n)}}\alpha_{\mathbf{\mu}} I(T_t^b;Y_t|X_t^a,S_t,S_{[t-1]}=\mathbf{\mu})+\eta(\epsilon)\label{eq:co-com-corr-2}\\
R_a+R_b&\leq& \sum_{\mathbf{\mu} \in {\cal S}^{(n)}}\alpha_{\mathbf{\mu}} I(X_t^a,T_t^b;Y_t|S_t,S_{[t-1]}=\mathbf{\mu})+\eta(\epsilon)\label{eq:co-com-corr-3}
\end{eqnarray}
where $\alpha_{\mathbf{\mu}}$ and $\eta(\epsilon)$ are defined in (\ref{eq:eta-epsilon}).
\end{lemma}
\begin{proof}
Let us first consider the sum-rate condition. Since,{\allowdisplaybreaks
\begin{eqnarray}
I(\mathbf{W};Y_{[n]},S_{[n]})&{\leq}&\sum_{t=1}^{n}\left[H(Y_t|S_{[t]})-H(Y_t|\mathbf{W},S_{[t]},Y_{[t-1]},X_t^a,T_t^b)\right]\nonumber\\
&\overset{(i)}{=}&\sum_{t=1}^{n}\left[H(Y_t|S_{[t]})-H(Y_t|S_{[t]},X_t^a,T_t^b)\right]\nonumber\\
&=&\sum_{t=1}^{n}I(X_t^a, T_t^b;Y_t|S_{[t]})\label{eq:conv-sr-corr-3},
\end{eqnarray}}where $(i)$ can be shown in a similar way as (\ref{eq:conv-cau-fl-6a}), we have,
\begin{eqnarray}
R_a+R_b\leq\frac{1}{n}\sum_{t=1}^nI(X_t^a,T_t^b;Y_t|S_{[t]})+\eta(\epsilon)\label{eq:conv-sr-corr-4}
\end{eqnarray}
and
\begin{eqnarray}
I(X_t^a,T_t^b;Y_t|S_{[t]})=n\sum_{\mathbf{\mu} \in {\cal S}^{(t-1)}}\alpha_{\mathbf{\mu}}I(X_t^a,T_t^b;Y_t|S_t, S_{[t-1]}=\mathbf{\mu}).
\end{eqnarray}
Substituting the above into (\ref{eq:conv-sr-corr-4}) yields (\ref{eq:co-com-corr-3}).

Let us now consider encoder $b$. With Fano's inequality and standard steps, we get
\begin{eqnarray}
R_b\leq \frac{1}{n}\log(|{\cal W}_b|)\leq \frac{1}{1-\epsilon}\frac{1}{n}\left(I(W_b;Y_{[n]}, S_{[n]})+H(\epsilon)\right)\label{eq:conv-enb-corr-2}.
\end{eqnarray}
Following similar reasonings as in (\ref{eq:conv-ena-inp-3}) we get,
{\allowdisplaybreaks
\begin{eqnarray}
I(W_b;Y_{[n]}, S_{[n]})&{\leq}&I(W_b;Y_{[n]}, S_{[n]}|W_a,S_{[n]}^a)\nonumber\\
&{=}&\sum_{t=1}^n\left[H(Y_t|S_{[t]},Y_{[t-1]},W_a,S_{[n]}^a)-H(Y_t|S_{[t]},Y_{[t-1]}, W_a, W_b, S_{[n]}^a)\right]\nonumber\\
&{=}&\sum_{t=1}^n\left[H(Y_t|S_{[t]},Y_{[t-1]},W_a,S_{[n]}^a,X_{[n]}^a)-H(Y_t|S_{[t]},Y_{[t-1]}, W_a, W_b, S_{[n]}^a,X_{[n]}^a)\right]\nonumber\\
&{\leq}&\sum_{t=1}^n\left[H(Y_t|S_{[t]},X_{t}^a)-H(Y_t|S_{[t]},Y_{[t-1]},W_a, W_b,S_{[n]}^a,X_{[n]}^a,T_{t}^b)\right]\nonumber\\
&\overset{(i)}{=}&\sum_{t=1}^n\left[H(Y_t|S_{[t]},X_{t}^a)-H(Y_t|S_{[t]},X_{t}^a,T_t^b)\right]\nonumber\\
&=&\sum_{t=1}^n I(T_t^b;Y_t|X_t^a,S_{[t]})\label{eq:conv-ena-corr-3}
\end{eqnarray}
where $(i)$ is valid since
\begin{eqnarray}
&&\hspace{-0.3in}P_{Y_t|S_{[t]},Y_{[t-1]},\mathbf{W}, S_{[n]}^a, X_{[n]}^a,T_t^b}(y_t|s_{[t]},y_{[t-1]},\mathbf{w}, s_{[n]}^a, x_{[n]}^a,t_t^b)\nonumber\\
&=&\sum_{s_t^b \in {\cal S}_b}P_{Y_t|S_t,S_t^b,X_t^a,T_t^b}(y_t|s_t,s_t^b,x_t^a,t_t^b) P_{S_t^b|S_{[t]},Y_{[t-1]},\mathbf{W}, S_{[n]}^a, X_{[n]}^a,T_t^b}(s_t^b|s_{[t]},y_{[t-1]},\mathbf{w}, s_{[n]}^a, x_{[n]}^a,t_t^b)\nonumber\\
&=&\sum_{s_t^b \in {\cal S}_b}P_{Y_t|S_t,S_t^b,X_t^a,T_t^b}(y_t|s_t,s_t^b,x_t^a,t_t^b) P_{S_t^b|S_t}(s_t^b|s_t)\nonumber\\
&=&P_{Y_t|S_t,X_t^a,T_t^b}(y_t|s_t,x_t^a,t_t^b)\label{eq:conv-sr-corr-3a}.
\end{eqnarray}}where the first equality is due to (\ref{eq:ch-recfl}) and the second equality is due to (\ref{eq:noi-recfl-mrk}) and (\ref{eq:sta-iid}). Following (\ref{eq:conv-cau-fl-6a}), we can directly verify (\ref{eq:co-com-corr-2}).
\end{proof}
We now need to show that the joint conditional distribution of channel state $S_t$, inputs $X_t^a, T_t^b$ and output $Y_t$ given the past realization $(S_{[t-1]}=\mathbf{\mu})$, i.e., $P_{X_t^a,T_t^b,Y_t,S_t|S_{[t-1]}}(x^a,t^b,y,s|\mathbf{\mu})$, factorizes as in (\ref{eq:joi-dist-corr}). This is straightforward. Let first
$\pi_{X^a, T^b}^{\mathbf{\mu}}(x^a,t^b):=P_{X_t^a,T_t^b|S_{[t-1]}}(x^a,t^b|\mathbf{\mu})$
and observe that
\begin{eqnarray}
&&\hspace{-0.5in}P_{X_t^a,T_t^b,Y_t,S_t|S_{[t-1]}}(x^a,t^b,y,s|\mathbf{\mu})\nonumber\\
&=&\sum_{s_t^b \in {\cal S}^b}P_{Y_t|X_t^a,X_t^b,S_t}(y|x^a,t^b(s_t^b),s)P_{S_t^b|S_t}(s_t^b|s_t)P_{S_t}(s)P_{X_t^a,T_t^b|S_{[t-1]}}(x^a,t^b|\mathbf{\mu})\nonumber\\
&=&\pi_{X^a, T^b}^{\mathbf{\mu}}(x^a,t^b)P_{S_t}(s)P_{Y_t|X_t^a,T_t^b,S_t}(y|x^a,t^b,s)\label{eq:conv-fact-corr-4}
\end{eqnarray}
where the equalities are verified by (\ref{eq:ch-recfl}), by (\ref{eq:noi-recfl-mrk}) and by the fact that $(X_t^a,T_t^b)$ is independent of $S_t$.
\end{proof}
We can now complete the converse proof of Theorem \ref{theo:main-corr}. With Lemma \ref{lem:conv-corr} it is shown that any achievable rate pair can be approximated by the convex combinations of rate conditions which are indexed by $\mathbf{\mu} \in {\cal S}^{(n)}$ and satisfy (\ref{eq:joi-dist-corr}) for joint state-input-output distributions. Hence, any achievable pair $(R_a,R_b) \in \overline{co}\big(\bigcup_{\hat{\pi}}\mathcal{R}_{C}(\hat{\pi})\big)$.
\section{Proof of ${\cal C}^G_{FS}={\cal C}_{AS}$}\label{app:3}
Let us first show that ${\cal C}^G_{FS}\subseteq {\cal C}_{AS}$. Recall that $T \in |{\cal T}|=|{{\cal X}}_b|^{|\cal S|}$ and $|{\cal U}|\leq|{\cal X}_a||{\cal X}_b||{\cal S}|+1$. Hence, we have either $|{\cal U}|>|{\cal T}|$ or else. In the case where $|{\cal U}|<|{\cal T}|$, we note that $|{\cal U}|$ is limited to a finite set without loss of generality. Hence, we can always take $|{\cal U}|$ at least $|{\cal T}|$ such that it satisfies (\ref{eq:rate-somek-enb}), (\ref{eq:rate-somek}) and (\ref{eq:dist-somek}). Then we can directly conclude that ${\cal C}^G_{FS}\subseteq {\cal C}_{AS}$ since $P_{X^b|S,T}(x^b|s,t)=P_{X^b|S,T}(x^b|s,t,x^a)=1_{\{x^b=t(s)\}}$ and this is a special case of $P_{X^b|U,X^a,S}(x^b|u,x^a,s)$.

In order to prove the other direction, i.e., ${\cal C}_{AS}\subseteq{\cal C}^G_{FS}$, let ${\cal C}_{AS}^E$ be the closure of all rate pairs $(R_a,R_b)$ satisfying
\begin{eqnarray}
R_b&<&I(U;Y|X^a) \label{eq:rate-somek-enb-eq}\\
R_b+R_a &<& I(U,X^a;Y) \label{eq:rate-somek-eq}
\end{eqnarray}
for some joint measure on ${\cal S}\times{\cal X}_a\times{\cal X}_b\times{\cal Y}\times{\cal U}$ having the form
\begin{eqnarray}
P_{Y|X^a,X^b,S}(y|x^a,x^b,s)1_{\{x^b=\mathbf{m}(s,x^a,u)\}}P_{S}(s)P_{X^a,U}(x^a,u)\label{eq:dist-somek-eq},
\end{eqnarray}
for some $\mathbf{m}:{\cal U}\times{{\cal X}_a}\times {\cal S}\rightarrow{\cal X}_b$, where $|{\cal U}|\leq|{\cal S}||{\cal X}_a||{\cal X}_b|+1$, and we first show that ${\cal C}_{AS}={\cal C}_{AS}^E$, and following this, we show that ${\cal C}_{AS}^E\subseteq{\cal C}^G_{FS}$.
\begin{lemma}\label{lem:somek-somekeq}
${\cal C}_{AS}={\cal C}_{AS}^E$.
\end{lemma}
\begin{proof}
Obviously ${\cal C}_{AS}^E\subseteq{\cal C}_{AS}$ and hence, we will show that ${\cal C}_{AS}\subseteq{\cal C}_{AS}^E$. Let $\bar{P}_{X^b,X^a,U,S}(x^b,x^a,u,s)$ be a joint distribution in the form of (\ref{eq:dist-somek}), i.e.,
\begin{eqnarray}
\bar{P}_{X^b,X^a,U,S}(x^b,x^a,u,s)=\bar{P}_{X^b|X^a,U,S}(x^b|x^a,u,s)P_{S}(s)\bar{P}_{X^a,U}(x^a,u)\label{eq:lem-somek-barp}.
\end{eqnarray}
Let $\bar{\mathbf{\Lambda}}$ denote a $|{\cal X}_a||{\cal U}||{\cal S}|$-by-$|{\cal X}_b|$ matrix where $\bar{\mathbf{\Lambda}}_{i, jkl}=\bar{P}_{X^b|X^a,U,S}(i|j,k,l)$, $1\leq i \leq |{\cal X}_b|$, $1\leq j \leq |{\cal X}_a|$, $1\leq k \leq |{\cal U}|$ and $1\leq l \leq |{\cal S}|$. Hence, $\bar{\mathbf{\Lambda}}$ is a $|{\cal X}_a||{\cal U}||{\cal S}|$-by-$|{\cal X}_b|$ row stochastic matrix, i.e., $\bar{\mathbf{\Lambda}}_{i, jkl}\geq0, \enspace \forall i,j,k,l$ and $\sum_{i=1}^{|{\cal X}_b|}\bar{\mathbf{\Lambda}}_{i, jkl}=1, \enspace \forall j,k,l$. Let $\mathbf{\Lambda}$ denote a $|{\cal X}_a||{\cal U}||{\cal S}|$-by-$|{\cal X}_b|$ binary stochastic matrix, that is a matrix with each row has exactly one non-zero element, which is $1$. Observe now that any row stochastic matrix can be written as a convex combination of binary stochastic matrices (e.g., see \cite[Lemma 5]{hognas} and \cite[Proposition IV.1]{niesen}). Therefore, we have
\begin{eqnarray}
\bar{\mathbf{\Lambda}}=\sum_{i=1}^k\lambda_i\mathbf{\Lambda}^{(i)}, \enspace \sum_{i=1}^k\lambda_i=1 \label{eq:lem-somek-convex},
\end{eqnarray}
where $\mathbf{\Lambda}^{(i)}$ is a binary stochastic matrix and by \cite[Lemma 5]{hognas}, $k\leq(|{\cal X}_a||{\cal U}||{\cal S}|)^2$.

Let, for the joint distribution $\bar{P}_{X^b,X^a,U,S}(x^b,x^a,u,s)$,
\begin{eqnarray}
\bar{R}_b&<&I(U;Y|X^a)_{\bar{\mathbf{\Lambda}}}, \label{eq:lem-somek-somekeq2}\\
\bar{R}_a+\bar{R}_b &<& I(U,X^a;Y)_{\bar{\mathbf{\Lambda}}} \label{eq:lem-somek-somekeq3}.
\end{eqnarray}
Therefore, $(\bar{R}_a,\bar{R}_b)\in{\cal C}_{AS}$. Now, observe that for a fix distribution $P_{X^a,U}(x^a,u)$, both $I(U,X^a;Y)$ and $I(U;Y|X^a)$ are convex in $P_{Y|X^a,U}(y|x^a,u)$ and hence, convex in $P_{X^b|X^a,U,S}(\cdot|x^a,u,s)$. This and (\ref{eq:lem-somek-convex}) imply that
\begin{eqnarray}
I(U;Y|X^a)_{\bar{\mathbf{\Lambda}}}&\leq& \sum_{i=1}^{k}\lambda_i I(U;Y|X^a)_{\mathbf{\Lambda}^{(i)}},\label{eq:lem-somek-somekeq4}\\
I(U,X^a;Y)_{\bar{\mathbf{\Lambda}}}&\leq& \sum_{i=1}^{k}\lambda_i I(U,X^a;Y)_{\mathbf{\Lambda}^{(i)}},\label{eq:lem-somek-somekeq5}
\end{eqnarray}
where $I(U;Y|X^a)_{\mathbf{\Lambda}^{(i)}}$ and $I(U,X^a;Y)_{\mathbf{\Lambda}^{(i)}}$ denote the mutual information terms induced by $\mathbf{\Lambda}^{(i)}$.

Now, let $(R_a^i,R_b^i)$, $1\leq i\leq k$, be such that
\begin{eqnarray}
R_b^i\leq I(U;Y|X^a)_{\mathbf{\Lambda}^{(i)}},\nonumber\\
R_b^i+R_a^i\leq I(U,X^a;Y)_{\mathbf{\Lambda}^{(i)}}\nonumber,
\end{eqnarray}
and hence, $(R_a^i,R_b^i)\in {\cal C}_{AS}^E$, $1\leq i\leq k$. Let $(R_a^f,R_b^f)=\sum_{i=1}^{k}\lambda_i(R_a^i,R_b^i)$. Since a convex combination of achievable rates is also achievable, so $(R_a^f,R_b^f)\in {\cal C}_{AS}^E$. This observation and inequalities (\ref{eq:lem-somek-somekeq2})-(\ref{eq:lem-somek-somekeq5}) complete the claim that $(\bar{R}_a, \bar{R}_b) \in {\cal C}_{AS}^E$.
\end{proof}
Up to now, we have shown that ${\cal C}^G_{FS}\subseteq{\cal C}_{AS}$ and ${\cal C}_{AS}^E={\cal C}_{AS}$. In order to prove that ${\cal C}^G_{FS}={\cal C}_{AS}$, it remains to show that ${\cal C}_{AS}^E\subseteq{\cal C}^G_{FS}$. Note that ${\cal C}_{AS}^E$ still depends on $P_{X^a,U}(x^a,u)$ in which $|{\cal U}|$ can be larger than $|{\cal T}|$. Hence, in the next lemma we basically show that for every $P_{X^a,U}(x^a,u)$, there exists a $\hat{\pi}_{T^a,U}(t^a,u)$ which induces the same rate constraints as induced by $P_{X^a,U}(x^a,u)$.
\begin{lemma}
${\cal C}_{AS}^E\subseteq{\cal C}^G_{FS}$.
\end{lemma}
\begin{proof}
Let us fix a joint distribution $P^{\ast}_{Y,X^a,X^b,U,S}(y,x^a,x^b,u,s)$ satisfying (\ref{eq:dist-somek-eq}), i.e.,
\begin{eqnarray}
P^{\ast}_{Y,X^a,X^b,U,S}(y,x^a,x^b,u,s)=P^{\ast}_{Y|X^a,X^b,S}(y|x^a,x^b,s)1_{\{x^b=\mathbf{m}(s,x^a,u)\}}P_{S}(s)P^{\ast}_{X^a,U}(x^a,u)\label{eq:somek-ind-func}.
\end{eqnarray}
Observe that for every $\mathbf{m}$ satisfying $x^b=\mathbf{m}(u,x^a,s)$, one can define
\begin{eqnarray}
x^b=\mathbf{m}(u,x^a,s)=\bar{\mathbf{m}}(x^a,u)(s), \enspace \bar{\mathbf{m}}(x^a,u) \in {\cal T}\label{eq:somek-ind-map},
\end{eqnarray}
where ${\cal T}$ is the set of all mappings from ${\cal S}$ to ${\cal X}_b$. Now, let
\begin{eqnarray}
\left(I(U;Y|X^a)_{{P^{\ast}_{Y,X^a,U}(y,x^a,u)}}, I(U,X^a;Y)_{{P^{\ast}_{Y,X^a,U}(y,x^a,u)}}\right)\label{eq:somek-mutinf},
\end{eqnarray}
denote the mutual information pair induced by $P^{\ast}_{Y,X^a,U}(y,x^a,u)$. We have{\allowdisplaybreaks
\begin{eqnarray}
&&\hspace{-0.8in}I(U,X^a;Y)_{{P^{\ast}_{Y,X^a,U}(y,x^a,u)}}\nonumber\\
&=&\sum_{u\in {\cal U}}\sum_{y\in {\cal Y}}\sum_{x^a\in {{\cal X}_a}}P^{\ast}_{Y,X^a,U}(y,x^a,u)\log\frac{P^{\ast}_{Y,U,X^a}(y,u,x^a)}{P^{\ast}_{Y}(y)P^{\ast}_{U,X^a}(u,x^a)}\nonumber\\
&\overset{}{=}&\sum_{t\in{\cal T}}\sum_{u\in {\cal U}}\sum_{y\in {\cal Y}}\sum_{x^a\in {{\cal X}_a}}P^{\ast}_{Y,X^a,U,T}(y,x^a,u,t)\log\frac{P^{\ast}_{Y,U,X^a}(y,u,x^a)}{P^{\ast}_{Y}(y)P^{\ast}_{U,X^a}(u,x^a)}\nonumber\\
&\overset{(i)}{=}&\sum_{t\in{\cal T}}\sum_{u\in {\cal U}}\sum_{y\in {\cal Y}}\sum_{x^a\in {{\cal X}_a}}P^{\ast}_{Y,X^a,U,T}(y,x^a,u,t)\log\frac{P^{\ast}_{Y,U,X^a,T}(y,u,x^a,t)}{P^{\ast}_{Y}(y)P^{\ast}_{U,X^a,T}(u,x^a,t)}\nonumber\\
&\overset{(ii)}{=}&\sum_{t\in{\cal T}}\sum_{u\in {\cal U}}\sum_{y\in {\cal Y}}\sum_{x^a\in {{\cal X}_a}}P^{\ast}_{Y,X^a,U,T}(y,x^a,u,t)\log\frac{P^{\ast}_{Y|X^a,T}(y|x^a,t)P^{\ast}_{U,T,X^a}(u,t,x^a)}{P^{\ast}_{Y}(y)P^{\ast}_{U,T,X^a}(u,t,x^a)}\nonumber\\
&=&\sum_{t\in{\cal T}}\sum_{u\in {\cal U}}\sum_{y\in {\cal Y}}\sum_{x^a\in {{\cal X}_a}}P^{\ast}_{Y,X^a,U,T}(y,x^a,u,t)\log\frac{P^{\ast}_{Y,X^a,T}(y,x^a,t)}{P^{\ast}_{Y}(y)P^{\ast}_{X^a,T}(x^a,t)}\nonumber\\
&=&\sum_{t\in{\cal T}}\sum_{y\in {\cal Y}}\sum_{x^a\in {{\cal X}_a}}P^{\ast}_{Y,X^a,T}(y,x^a,t)\log\frac{P^{\ast}_{Y,X^a,T}(y,x^a,t)}{P^{\ast}_{Y}(y)P^{\ast}_{X^a,t}(x^a,t)}\nonumber\\
&=&I(T,X^a;Y)_{P^{\ast}_{Y,X^a,T}(y,x^a,t)}\label{eq:somek-ind-mutinf},
\end{eqnarray}
where $(i)$ is valid since $\bar{\mathbf{m}}(x^a,u) \in {\cal T}$, i.e., for each $(x^a,u)$ there exists only one $t\in{\cal T}$ such that $P_{T|X^a,U}(t|x^a,u)=1$, $(ii)$ is valid since
\begin{eqnarray}
P^{\ast}_{Y|X^a,T,U}(y|x^a,t,u)&\overset{(iii)}{=}&\sum_{s \in {\cal S}}P^{\ast}_{Y|X^a,T,U,S}(y|x^a,t,u,s)P_{S}(s)\hspace{0.05in}\overset{(iv)}{=}\hspace{0.05in}\sum_{s \in {\cal S}}P_{Y|X^a,T,S}(y|x^a,t,s)P_{S}(s)\nonumber\\
&=&\sum_{s \in {\cal S}}P^{\ast}_{Y,S|X^a,T}(y,s|x^a,t)\hspace{0.05in}=\hspace{0.05in}P^{\ast}_{Y|X^a,T}(y|x^a,t),
\end{eqnarray}
where $(iii)$ is valid since $S$ and $(X^a,T,U)$ are independent and $(iv)$ is valid due to (\ref{eq:ch-recfl}). Similarly, we have
\begin{eqnarray}
&&\hspace{-0.8in}I(U;Y|X^a)_{{P^{\ast}_{Y,X^a,U}(y,x^a,u)}}\nonumber\\
&=&\sum_{u\in {\cal U}}\sum_{y\in {\cal Y}}\sum_{x^a\in {{\cal X}_a}}P^{\ast}_{Y,X^a,U}(y,x^a,u)\log\frac{P^{\ast}_{Y,U|X^a}(y,u|x^a)}{P^{\ast}_{Y|X^a}(y|x^a)P^{\ast}_{U|X^a}(u|x^a)}\nonumber\\
&\overset{}{=}&\sum_{u\in {\cal U}}\sum_{y\in {\cal Y}}\sum_{x^a\in {{\cal X}_a}}P^{\ast}_{Y,X^a,U}(y,x^a,u)\log\frac{P^{\ast}_{Y,U,X^a}(y,u,x^a)}{P^{\ast}_{Y|X^a}(y|x^a)P^{\ast}_{U,X^a}(u,x^a)}\nonumber\\
&\overset{(v)}{=}&\sum_{t\in{\cal T}}\sum_{u\in {\cal U}}\sum_{y\in {\cal Y}}\sum_{x^a\in {{\cal X}_a}}P^{\ast}_{Y,X^a,U,T}(y,x^a,u,t)\log\frac{P^{\ast}_{Y,U,X^a,T}(y,u,x^a,t)}{P^{\ast}_{Y|X^a}(y|x^a)P^{\ast}_{U,X^a,T}(u,x^a,t)}\nonumber\\
&\overset{(vi)}{=}&\sum_{t\in{\cal T}}\sum_{u\in {\cal U}}\sum_{y\in {\cal Y}}\sum_{x^a\in {{\cal X}_a}}P^{\ast}_{Y,X^a,U,T}(y,x^a,u,t)\log\frac{P^{\ast}_{Y|T,X^a}(y|t,x^a)P^{\ast}_{U,T,X^a}(u,t,x^a)}{P^{\ast}_{Y|X^a}(y|x^a)P^{\ast}_{U,T,X^a}(u,t,x^a)}\nonumber\\
&=&\sum_{t\in{\cal T}}\sum_{u\in {\cal U}}\sum_{y\in {\cal Y}}\sum_{x^a\in {{\cal X}_a}}P^{\ast}_{Y,X^a,U,T}(y,x^a,u,t)\log\frac{P^{\ast}_{Y,T|X^a}(y,t|x^a)}{P^{\ast}_{Y|X^a}(y|x^a)P^{\ast}_{T|X^a}(t|x^a)}\nonumber\\
&=&\sum_{t\in{\cal T}}\sum_{y\in {\cal Y}}\sum_{x^a\in {{\cal X}_a}}P^{\ast}_{Y,X^a,T}(y,x^a,t)\log\frac{P^{\ast}_{Y,T|X^a}(y,t|x^a)}{P^{\ast}_{Y|X^a}(y|x^a)P^{\ast}_{T|X^a}(t|x^a)}\nonumber\\
&=&I(T;Y|X^a)_{P^{\ast}_{Y,X^a,T}(y,x^a,t)}\label{eq:somek-ind-mutinfa},
\end{eqnarray}}
where $(v)$ and $(vi)$ follows from the same reasonings of $(i)$ and $(ii)$, respectively.  Now, let $R^{'}_{b}<I(U;Y|X^a)_{{P^{\ast}_{Y,X^a,U}(y,x^a,u)}}$ and $R^{'}_{b}+R^{'}_{a}<I(U,X^a;Y)_{{P^{\ast}_{Y,X^a,U}(y,x^a,u)}}$. Hence, $(R^{'}_{a},R^{'}_{b}) \in {\cal C}_{AS}^E$. Observe now that for a distribution in the form of $P^{\ast}_{Y,X^a,T}(y,x^a,t)$, one can define $\hat{\pi}_{X^a,T}(x^a,t)=P^{\ast}_{X^a,T}(x^a,t)$.  Therefore, since ${\cal C}^G_{FS}=\overline{co}\bigg(\bigcup_{\hat{\pi}}\mathcal{R}_{C}^{'}(\hat{\pi})\bigg)$, and due to (\ref{eq:somek-ind-mutinf}) and (\ref{eq:somek-ind-mutinfa}), $(R^{'}_{a},R^{'}_{b}) \in {\cal C}_{FS}^G$, which completes the claim.
\end{proof}

\end{document}